\documentclass[11pt]{article}
\usepackage{preamble}

\usepackage[margin=1in]{geometry}
\title{Noisy Quantum Learning Theory}
\mathtoolsset{showonlyrefs}

\allowdisplaybreaks

\usepackage[lined,boxed,ruled,norelsize,linesnumbered]{algorithm2e}

\usepackage{defs}
\usepackage{comment}
\usepackage{textgreek}
\usepackage{dsfont}
\usepackage{bm}

\author{
Jordan Cotler\thanks{Email: \url{jcotler@fas.harvard.edu}.} \\
Harvard University
\and
Weiyuan Gong\thanks{Email: \url{wgong@g.harvard.edu}.} \\
Harvard University
\and
Ishaan Kannan\thanks{Email: \url{ishaan@alumni.caltech.edu}.} \\
Harvard University
}

\newcommand{\DN}{\mathcal{D}^{\otimes n}_\lambda}

\newcommand{\nbqp}{{\textnormal{\textsf{NBQP}}}}
\newcommand{\nisq}{{\textnormal{\textsf{NISQ}}}}
\newcommand{\bqp}{{\textnormal{\textsf{BQP}}}}

\newcommand{\zv}{\mathbf{z}}

\renewcommand{\norm}[1]{\|#1\|}

\renewcommand{\i}{\mathrm{i}}

\usepackage[titles]{tocloft}

\usepackage{stackengine}
\stackMath
\newcommand\tsup[2][2]{%
 \def\useanchorwidth{T}%
  \ifnum#1>1%
    \stackon[-.5pt]{\tsup[\numexpr#1-1\relax]{#2}}{\scriptscriptstyle\sim}%
  \else%
    \stackon[.5pt]{#2}{\scriptscriptstyle\sim}%
  \fi%
}

\begin{document}
\pagestyle{empty}
{
  \renewcommand{\thispagestyle}[1]{}
  \maketitle

\begin{abstract}

We develop a framework for learning from noisy quantum experiments in which fault-tolerant devices access uncharacterized systems through noisy couplings. Introducing the complexity class $\textsf{NBQP}$ (``noisy BQP''), we model noisy fault-tolerant quantum computers that cannot generally error-correct the oracle systems they query. Using this class, we prove that while noise can eliminate the exponential quantum learning advantages of unphysical, noiseless learners, a superpolynomial gap remains between $\textsf{NISQ}$ and fault-tolerant devices. Turning to canonical learning tasks in noisy settings, we find that the exponential two-copy advantage for purity testing collapses under local depolarizing noise. Nevertheless, we identify a setting motivated by AdS/CFT in which noise-resilient physical structure restores this quantum learning advantage. We then analyze noisy Pauli shadow tomography, deriving lower bounds characterizing how instance size, quantum memory and noise jointly control sample complexity, and design algorithms with parametrically matching scalings. We study similar tradeoffs in quantum metrology, and show that the Heisenberg-limited sensitivity of existing error-correction-based protocols persists only up to a timescale inverse-polynomial in the error rate per probe qubit. Together, our results demonstrate that the primitives underlying quantum-enhanced experiments are fundamentally fragile to noise, and that realizing meaningful quantum advantages in future experiments will require interfacing noise-robust physical properties with available algorithmic techniques.
\end{abstract}

}

\clearpage
\pagestyle{plain}
\pagenumbering{arabic}

\tableofcontents

\newpage

\section{Introduction}
\label{sec:intro}
Quantum learning theory has advanced rapidly in recent years, and one of its key successes is recasting experimental protocols as learning problems, which admits a rich toolkit for determining when and how quantum computation-enhanced experiments can outperform conventional ones~\cite{Aaronson2006Learnability,HaahHarrowJiWuYu2017SampleOptimal,Wang_2017,Aaronson2018ShadowTomography,ArunachalamDeWolf2018QuantumSampleComplexity,Aaronson_2019,classical_shadow,cotler2020quantum,chen2021exponentialseparationslearningquantum,Huang_2021,Aharonov_Cotler_Qi_2022,nisq,Huang_2023}. From this perspective, partially uncharacterized experimental systems serve as oracles providing quantum data, and the experimentalist implements a learning protocol that reveals the system's properties through controlled interactions. When a quantum computer is available, it can be coherently coupled to the system of interest, enabling quantum computation on transduced data. In idealized settings, such quantum computation-enhanced experiments can provide a provable exponential advantage over conventional experimental protocols~\cite{Huang_2021, chen2021exponentialseparationslearningquantum, Huang_2022, Aharonov_Cotler_Qi_2022}.

However, for most known examples of quantum experimental advantage, the quantum computer and its coupling to the experimental system are assumed to be noiseless. While a noisy quantum computer can be error-corrected, the experimental system given to us by Nature does not come embedded into an error-correcting code. We are therefore in a setting where an error-corrected quantum computer can only access the experimental system via a noisy coupling -- one that must also mediate between the logical encoding of the quantum computer and the bare physical degrees of freedom of the experimental system. Moreover, since the experimental system corresponds to a partially unknown state or channel, it cannot generally be embedded into an error-correcting code by the experimentalist \cite{aharonov1996limitationsnoisyreversiblecomputation}. Noise thus fundamentally changes the character of quantum computation-enhanced experiments and may obviate many of their known advantages. 

In this paper, we develop a theory of \emph{noisy quantum learning}: for several prominent examples of quantum computation-enhanced experiments, we establish when noise eliminates quantum advantages and when such advantages persist.

Concretely, noisy quantum learning theory reveals two distinct mechanisms by which noise erodes ideal quantum advantages. First, nearly all known idealized separations rely on learning tasks governed by intrinsically complex, high-weight structures, including states or dynamics that resist succinct classical representations~\cite{ chen2021exponentialseparationslearningquantum, Huang_2022, nöller2025infinitehierarchymulticopyquantum,weiyuan_paulis, King_2025}. For example, the quantum-enhanced sample-complexity gap in estimating physically motivated Pauli observables grows with the weight of the observables \cite{weiyuan_paulis}, and classical simulation of quantum dynamics becomes intractable when high-weight Paulis dominate the evolution \cite{Schollw_ck_2011,begušić2023fastclassicalsimulationevidence}. Our results show, however, that precisely these high-complexity components are the most vulnerable to noise. As a result, the very structures that underwrite ideal quantum advantages also provide the primary mechanism through which noise destroys them.

Second, exponential quantum advantages achieved through multi-copy measurements typically require entanglement across a number of qubits that scales extensively with the problem size~\cite{Huang_2022, chen2021exponentialseparationslearningquantum,gong2024samplecomplexitypurityinner}. Operationally, such algorithms are built from a small set of standard primitives, most notably maximally entangled (Bell-basis) measurements and many-qubit SWAP operations, that underlie idealized superpolynomial speedups in quantum learning, property testing, and computation~\cite{Harrow:2012gwf, Lloyd_2014,bădescu2017quantumstatecertification,  montanaro2018surveyquantumpropertytesting}. In our setting, implementing these primitives on noisy, uncharacterized quantum systems leads to an exponential blow-up in sample complexity. Absent additional structure in the physical system that can be exploited for error correction, the primitives that enable ideal quantum speedups cease to provide robust advantages in the noisy regime.

Making these concepts rigorous requires ruling out adaptive learning protocols that might otherwise compensate for noise. Our framework for proving the inefficiency of noisy quantum learning protocols, even with adaptive strategies, builds on the learning tree formalism of~\cite{chen2021exponentialseparationslearningquantum}, and in particular techniques from~\cite{nisq}. In the latter work, noisy quantum computers \emph{incapable} of fault-tolerant error correction were formalized as the complexity class $\textsf{NISQ}$, and it was shown that there exist oracles $O$ for which $\textsf{NISQ}^O \subsetneq \textsf{BQP}^O$, meaning that even adaptive NISQ protocols are superpolynomially slower than noiseless quantum computers for certain learning problems. Our refinement is to introduce a more physically realistic complexity class, $\textsf{NBQP}$, corresponding to \emph{noisy} quantum computers with a constant error rate per qubit per circuit layer. The noise rate is taken to be below the threshold for fault-tolerant error correction, so that $\textsf{NBQP} = \textsf{BQP}$.  Accordingly, we refer to learning protocols captured by $\textsf{BQP}^O$ as \emph{ideal quantum learning}, and those captured by $\textsf{NBQP}^O$ as \emph{natural quantum learning}, which are depicted in Fig.~\ref{fig:idealnatural}. We demonstrate that for certain natural quantum learning problems, known to have an exponential advantage over conventional experiments with unentangled access, the corresponding oracles $O$ enforce $\textsf{NBQP}^O \subsetneq \textsf{BQP}^O$. A consequence is that noise can obviate the advantage of even adaptive quantum computing-enhanced experiments over conventional experiments.

\begin{figure}[t!]
    \centering
    \includegraphics[width=\linewidth]{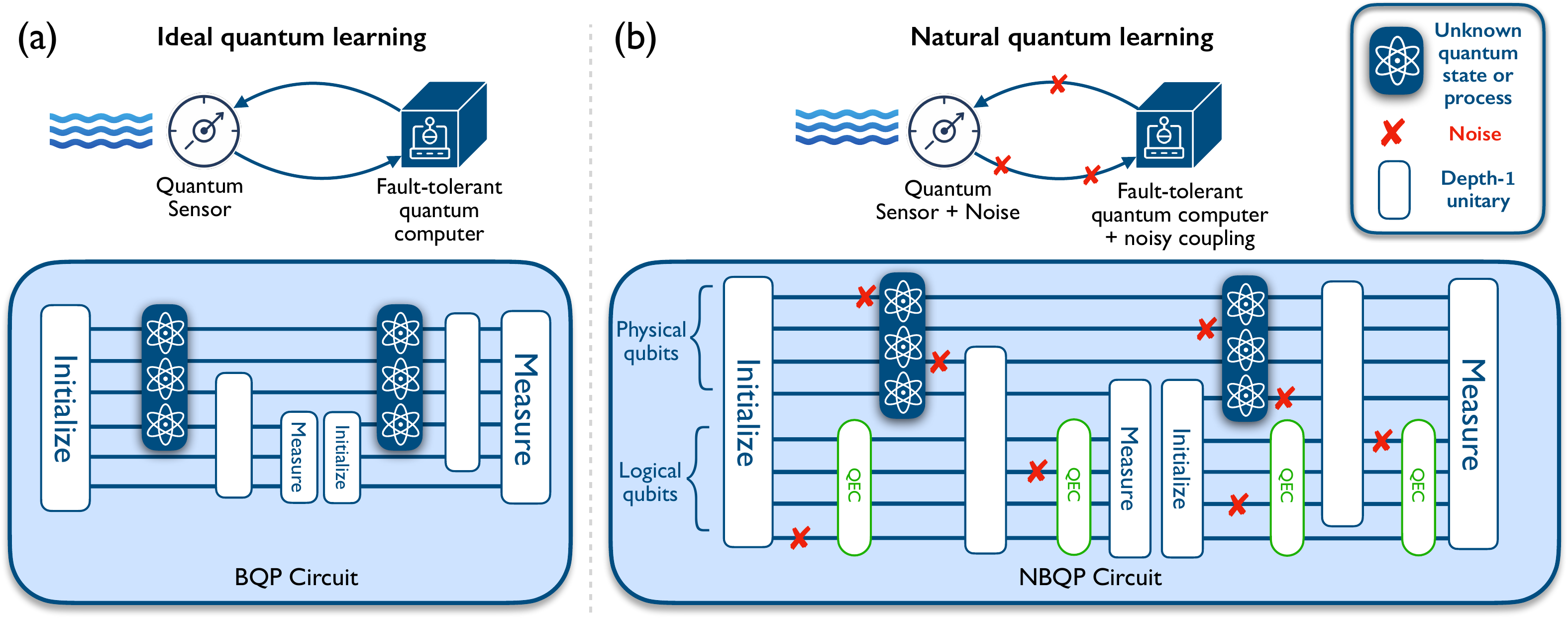}
    \caption{(a) \emph{Ideal quantum learning}: a noiseless quantum computer queries an unknown quantum state or channel through a sensor, yielding noiseless coherent oracle access modeled by $\textsf{BQP}^O$. (b) \emph{Natural quantum learning}: a noisy but fault-tolerant quantum computer queries an unknown quantum state or channel through a sensor, with noisy coupling occurring at the physical-qubit level. The quantum circuit must correct errors on the quantum computer and contend with noise induced by coupling to the experimental system, modelled by $\textsf{NBQP}^O$.}
    \label{fig:idealnatural}
\end{figure}

Given that several canonical results in quantum learning theory break down in noisy settings, what kinds of quantum advantages can still persist in the presence of such noise? Our work suggests two avenues to realize practically meaningful gains in noisy quantum-enhanced experiments. Many physical systems possess endogenous robustness: thermalizing systems have built-in noise resilience \cite{Chesi_2010, Brandao_2019, Kastoryano_2025}, and long-range correlations in many-body quantum systems can be protected by renormalization group structure \cite{happy_2015, Kim_2017, Lake_2025}. A useful insight is that the robustness of the physical system and the fault-tolerant architecture of the quantum computer can meet in the middle, allowing exponential quantum advantages to survive. We demonstrate this through illustrative examples in which a natural error-correcting structure present in a class of quantum states synergizes with the error correction of the quantum computer. Furthermore, we construct natural quantum learning problems for which $\textsf{NISQ}^O \subsetneq \textsf{NBQP}^O$, showing that error-correcting the quantum computer provides a superpolynomial advantage even when the coupling to the experimental system remains noisy.

Even in noisy regimes where asymptotic exponential quantum advantages disappear~\cite{Chen_2024_pauli, seif2024, Oh_2024, Liu_2025}, sub-exponential advantages can persist. Analyzing noisy quantum-enhacned experimental protocols requires tracking how complexity depends jointly on noise rates, quantum resources, and finite instance size, since practical experimental behavior is governed by these factors rather than asymptotic limits. Our Pauli-learning separation illustrates that even for a canonical quantum learning task, noise-sensitive effects yield meaningful quantum advantages that interpolate between idealized exponential separations and polynomial ones as the noise rate scales with instance size. Beyond quantum state learning, we apply noise- and resource-aware analyses to metrology, showing that single-parameter protocols achieving Heisenberg-limited sensitivity through active error correction can sustain this scaling only up to a timescale set by instance size, control speed, and device-probe coupling noise.

\section{Main Results}
\label{sec:results}

In Section~\ref{subsec:oraclesep}, we formalize the class $\nbqp$ and prove two superpolynomial oracle separations: $\nbqp^{O_1} \subsetneq \bqp^{O_1}$ and $\nisq^{O_2} \subsetneq \nbqp^{O_2}$.   Next, we move from complexity separations to concrete noisy learning tasks.  In Section~\ref{subsec:purity_testing}, we study the $\nbqp$ complexity of testing the purity of a quantum state, showing that an ideal exponential quantum advantage is fundamentally obstructed by noise unless the class of states in question has a latent error-correcting structure. We give a concrete example of such a structure in a physically-motivated toy setting: assessing the purity of a black hole microstate in the bulk of a tensor-network model of holographic duality using only noisy measurements of the boundary state. In Section~\ref{subsec:quantumlearningadv}, we analyze Pauli shadow tomography, demonstrating a tradeoff between sample complexity, the noise rate, and the complexity of the estimated observables. Using similar techniques, we show that error-correction-based metrology protocols retain Heisenberg-limited sensitivity only up to a sensing timescale that scales inverse-polynomially with sensor dimension and noise rate. While it is still possible to have a meaningful \emph{polynomial} advantage of certain quantum-enhanced experiments vis-\`{a}-vis conventional experiments in the noisy setting, we find that noise exponentially degrades the most prominent idealized quantum advantages. Technical details of the proofs are deferred to the appendix.

\subsection{Oracle separations between $\textsf{NBQP}$, $\textsf{BQP}$, and $\textsf{NISQ}$}
\label{subsec:oraclesep}

We now make precise the notions introduced above. An ideal quantum computer operates as if each qubit evolves noiselessly between operations; in practice, errors accumulate at each circuit layer. Fault-tolerant quantum computation tells us these errors can be corrected if the error rate per qubit per circuit layer is below a certain threshold, motivating the following complexity class.

\begin{definition}[$\textsf{NBQP}$ complexity class, informal] 
$\textnormal{\textsf{NBQP}}$ contains all problems solvable in polynomial time by a noisy quantum computer with polynomial-size circuits, where all operations are subject to constant depolarizing noise per qubit at a rate below the threshold of known quantum fault-tolerance schemes. 
\end{definition}

\noindent We work with depolarizing noise for concreteness, as it is a standard and well-studied error model. Our lower bounds hold for any noise model at least as strong as depolarizing noise, and our constructions can be adapted to other local stochastic noise models with constant error rate per qubit.

Under the above definition, an $\textsf{NBQP}$ quantum computer can implement error correction and perform recovery after each layer of gates on any register embedded into a code. The threshold theorem~\cite{aharonov1999faulttolerantquantumcomputationconstant} then implies $\textsf{NBQP} = \textsf{BQP}$. However, the situation changes when we introduce oracle access. To make this concrete, consider coupling a fault-tolerant quantum computer to a partially uncharacterized quantum material in the laboratory. Even if the quantum computer itself is fully error-corrected, the interaction with the material occurs at the physical-qubit level and is therefore noisy. We cannot simply transduce the material's state into an error-correcting code; even if this were technologically feasible, the encoding procedure could corrupt the very quantum information we are trying to learn about. The resulting asymmetry — that errors on the computation can be corrected but errors on oracle queries cannot — opens the possibility of nontrivial separations between $\textsf{NBQP}^O$ and $\textsf{BQP}^O$. While error mitigation may be possible for particularly structured oracles, we demonstrate that the inability to perform true error correction on the unknown system can starkly degrade our ability to learn from experiments, leading to superpolynomial oracle separations between $\textsf{NBQP}$ and $\textsf{BQP}$.

\begin{theorem}\label{thm:NBQP_vs_BQP}
    There exists an oracle $O_1$ such that $\textnormal{\textsf{NBQP}}^{O_1}\subsetneq \textnormal{\textsf{BQP}}^{O_1}$.
\end{theorem}

\noindent To prove this separation, we adapt the lifted Simon oracle construction of~\cite{nisq}. Conceptually, the standard Simon oracle is modified so that slight perturbations to its input result in completely uninformative outputs, making the oracle highly nonrobust to $\nbqp$-type noise. More concretely, given a function $f:\{0,1\}^n\to\{0,1\}^n$, consider the usual Simon's promise: either (i) $f$ is injective, or (ii) there exists a nonzero $s \in \{0,1\}^n$ such that $f(x) = f(x\oplus s)$ for all $x$ (a Simon's function).  The lift $\tilde f$ acts on $2n$ input bits but is nontrivial only when the last $n$ bits are $0^n$.  A noiseless quantum computer can simply restrict its queries to strings of the form $x\,0^n$ and then run the standard Simon's algorithm~\cite{Simon} on the first $n$ bits.  Thus the promise problem ``is $f$ injective or a Simon's function?'' lies in $\bqp^{O_1}$ with $O(n)$ queries.

For $\nbqp^{O_1}$, oracle calls act on \emph{physical} qubits, and in our model each such call is preceded and followed by a layer of depolarizing noise.  Even if the algorithm tries to prepare “good’’ queries of the form $x\,0^n$, these noise layers quickly flip some of the trailing $n$ bits, so with high probability the actual query lies outside the special subspace on which $\tilde f$ encodes $f$; on those inputs $O_1$ simply outputs $0$.  In effect, to a noisy learner $O_1$ is almost indistinguishable from a trivial oracle, and exponentially many queries in $n$ are required for such a learner to tell the difference.  Using the hybrid/distinguishability framework of~\cite{nisq}, namely their channel-level hybrid lemma and a node-perturbation argument in our learning-tree model (formalized as Lemmas~\ref{lemma:tree_id_vs_o} and~\ref{lemma:node_perturbation}), we show that any $\lambda$-noisy circuit making $N$ oracle calls has output distribution within $N^2 e^{-\Omega(\lambda n)}$ total variation distance of the distribution obtained by replacing every oracle call with the identity channel.  Under this identity oracle the two promise cases (injective versus Simon's $f$) induce exactly the same distribution, so any $\nbqp^{O_1}$ algorithm with $N = \mathrm{poly}(n)$ cannot achieve constant distinguishing advantage, whereas a $\bqp^{O_1}$ algorithm can.  This yields the strict separation $\nbqp^{O_1}\subsetneq \bqp^{O_1}$ claimed in the theorem.

Despite this limitation, an \textsf{NBQP} machine can implement polynomial-depth fault-tolerant quantum circuits, which are believed to be strictly more powerful than the logarithmic-depth circuits achievable by noisy intermediate-scale (NISQ) devices \cite{fenner2003boundspowerconstantdepthquantum, fanout, Haferkamp_2022, arora2023quantum}.  Leveraging this gap in coherent depth, we prove a superpolynomial oracle separation between $\textsf{NISQ}$ and $\textsf{NBQP}$, showing that fault-tolerant learners retain an advantage over near-term devices even when all access to the experimental system is noisy.

\begin{theorem} \label{Thm:NISQ_vs_NBQP}
    There exists an oracle $O_2$ such that $\textnormal{\textsf{NISQ}}^{O_2}\subsetneq \textnormal{\textsf{NBQP}}^{O_2}$.
\end{theorem}

\noindent For the proof of this theorem, we explicitly leverage the fundamental gap between $\nisq$ and $\nbqp$ machines: the latter has ability to perform quantum error correction, and thus can execute deep quantum circuits. Following~\cite{Chia_2023}, we start from the $d$-level Shuffling Simon’s Problem, a variant of Simon's problem challenging for shallow quantum circuits to solve. Here, a Simon function $f$ on $n$ bits is embedded and randomly permuted inside a much larger domain so that only a hidden subset of inputs carries any information about the secret $s$.  We then define an \emph{encoded} shuffling oracle $\mathcal{O}^{\mathrm{enc}}_{f,d}$ that acts like this shuffled Simon oracle on a logical code space of a fixed fault-tolerant quantum error correction scheme, and trivially outside it.  By Theorem~4.11 of~\cite{Chia_2023}, there is a noiseless depth-$O(d)$ circuit that recovers $s$, and an \textsf{NBQP} machine can simulate this circuit fault-tolerantly (for the noise rate $\lambda$ below threshold) with only polynomial overhead, and so $\textsf{Enc-}d\textsf{-SSP}\in\nbqp^{O_2}$. 

The lower bound against \textsf{NISQ} proceeds in two steps.  First, building on the analysis of~\cite{Chia_2023}, we show that any $\textsf{BPP}^{\textsf{QNC}_d}$ algorithm has exponentially small success probability on the encoded problem: even if it makes polynomially many depth-$d$ quantum queries, the domain in which the Simon's function is embedded is so large that the algorithm is very unlikely to query the relevant subspace. While classical advice between bounded-depth subroutines can reveal successively more information about this ``hidden domain", the number of permutations applied to the Simon's function is, by construction, too large for any $\textsf{BPP}^{\textsf{QNC}_d}$ algorithm to locate the domain. Mathematically, we carry out this argument by demonstrating that, with high probability, all oracle queries made by the algorithm can be replaced by ``shadow'' queries that agree with the true oracle outside a small wrapper set containing the hidden domain, but output no information on the domain itself; even after this swap, the output distribution of the algorithm is hardly altered. Upon this replacement, the algorithm learns nothing about $s$, as it no longer has access to the embedded Simon's function (Lemma~\ref{lemma:QNC_small_success}), and can thus do no better than random guessing. This establishes an exponentially small success probability for any $\textsf{BPP}^{\textsf{QNC}_d}$ algorithm.

Second, we show that any \textsf{NISQ} algorithm making polynomially many oracle calls can be simulated, up to vanishing total variation distance, by such a $\textsf{BPP}^{\textsf{QNC}_d}$ algorithm: we cut each noisy circuit at a fixed depth threshold and use KL-divergence bounds from~\cite{nisq} to argue that replacing deeper noisy circuits by shallow ones only changes the leaf distribution of the learning tree by $o(1)$.  Combining these two ingredients and applying Le Cam's two-point method, we obtain that no \textsf{NISQ} algorithm with $\mathrm{poly}(n)$ queries can recover the Simon's secret with success probability at least $2/3$, whereas an \textsf{NBQP} algorithm can, yielding the oracle separation $\textsf{NISQ}^{O_2}\subsetneq \textsf{NBQP}^{O_2}$.

\subsection{Purity testing in noisy quantum experiments} \label{subsec:purity_testing}
So far our results have established oracle separations between relativized complexity classes. We now turn to more concrete noisy learning tasks, asking how sample complexity scales with coherent quantum memory and access to joint measurements, and whether exponential quantum advantages can survive constant local noise. Our first result in this direction, Theorem \ref{thm:purity_testing_lb_informal}, addresses a canonical task in quantum learning theory and property testing, namely testing the purity of a quantum state, which has been highlighted as one of the few examples exhibiting quantum advantages in both sample and computational complexity \cite{montanaro2018surveyquantumpropertytesting, bubeck2020entanglementnecessaryoptimalquantum, Huang_2021, Aharonov_Cotler_Qi_2022,Huang_2022, chen2021exponentialseparationslearningquantum, weiyuan_paulis,gong2024samplecomplexitypurityinner}. In the ideal setting, distinguishing an 
$n$-qubit maximally mixed state from a fixed pure state requires $\Omega(2^{n})$ samples, whereas permitting joint measurements on two copies reduces the task to only $O(1)$ samples and constant computational time. This separation is powered by a quantum subroutine that appears throughout the literature on quantum speedups in learning, property testing, and computation: the multi-qubit \textsf{SWAP} operation~\cite{Harrow:2012gwf,Lloyd_2014,bădescu2017quantumstatecertification,Cincio_2018,montanaro2018surveyquantumpropertytesting,Huang_2022}. Here, we show that this super-exponential quantum speedup is completely degraded by only a single layer of noise. 

\begin{theorem}[No quantum advantage for purity testing in the presence of noise, informal] \label{thm:purity_testing_lb_informal}
Any algorithm that can test the purity of a quantum state using noisy two-copy measurements requires at least order $c(\lambda)^n$ samples, where $c(\lambda)>1$ is a constant depending only on the noise rate.
\end{theorem}
\noindent Our proof begins with the following hypothesis testing problem: given copies of a state $\rho$, guaranteed to be either a maximally mixed state, or a fixed pure state sampled from the Haar measure on $n$-qubits, can we identify the ground truth with probability at least $2/3$? We then construct a learning tree model for algorithms which can make adaptive, joint measurements on $\rho\otimes\rho$, interleaved with classical computation. As in the $\nbqp$ model, each pair $\rho\otimes \rho$ is corrupted by a depolarizing channel $\mathcal{D}_\lambda^{\otimes 2n}$ before performing an \textit{arbitrary} $2n$-qubit POVM; note that the latter part permits noiseless quantum circuits of any depth. Hence, our lower bound holds against any model of noisy computation stronger than a single layer of depolarizing noise applied at the outset. 

Our analysis proceeds by bounding the likelihood ratio between the leaf distributions generated by this learning tree when its measurement outcomes arise from the maximally mixed state versus from a fixed, Haar-random pure state. To do so, we focus on a fixed root-to-leaf path, utilizing the tree's multilinear structure to reinterpret the learning problem. Under our reinterpretation, it suffices to bound the concentration of the output distributions of intermediate nodes in the tree, which we achieve by proving several technical lemmas regarding the noise channel and using the martingale formalism from \cite{weiyuan_paulis}. Because each node concentrates by a factor exponentially small in the noise rate, every experiment is highly uninformative, and we obtain our lower bound.

\begin{figure}[t!]
    \centering
    \includegraphics[width=\linewidth]{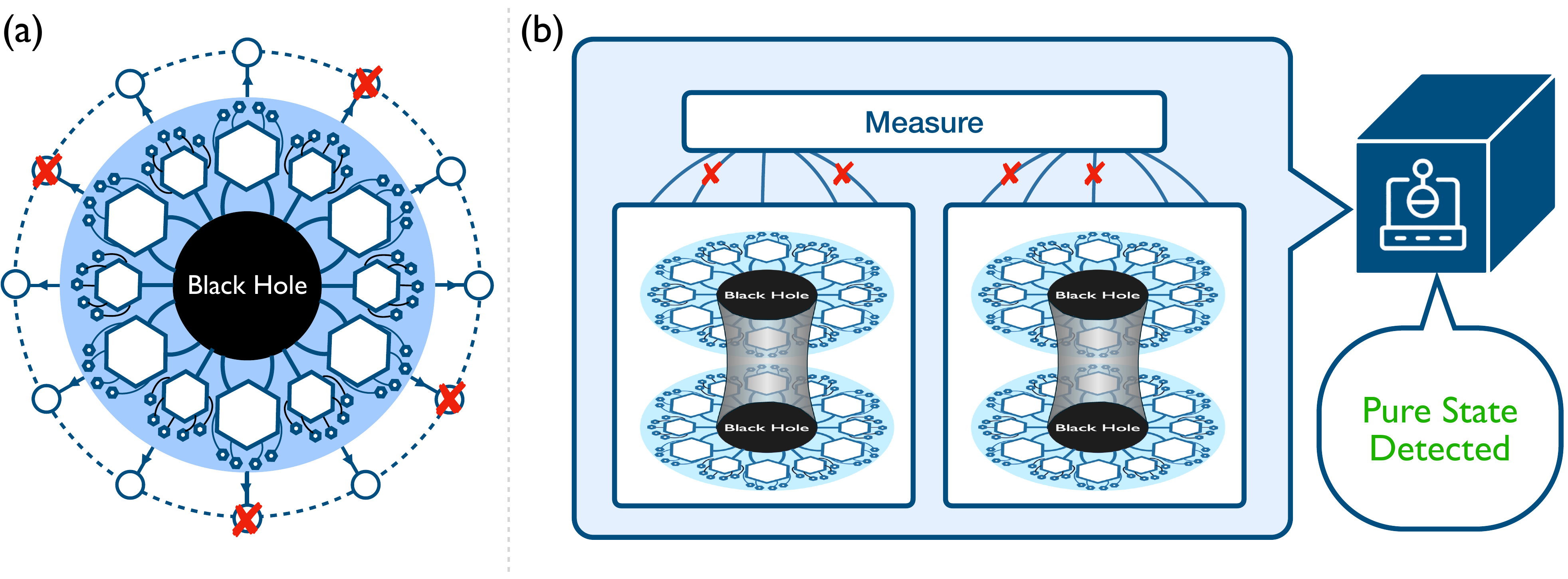}
    \vspace{2pt}
\caption{(a) \textit{Holographic encoding of a black hole}: a HaPPY tensor network with a central bulk region encoding a black hole state, shown as a black disk. The outer dangling legs represent boundary qubits of the dual CFT; red crosses indicate qubits lost to an erasure channel acting independently on each boundary site.
(b) \textit{Quantum-enhanced purity test}: two noisy boundary copies are first approximately decoded back toward the bulk and then used as input to a joint measurement (e.g.~a two-copy \textsf{SWAP} test) implemented by a quantum device. This protocol distinguishes whether the bulk black hole state is pure or mixed using only a constant number of copies, even in the presence of boundary erasures.}
    \label{fig:happy_code}
\end{figure}

While we have shown that assessing the purity of a generic quantum state is fundamentally obstructed by noise, there are natural classes of ``noise-robust'' states for which purity testing remains feasible. A toy but instructive example arises in tensor-network models of the AdS/CFT correspondence \cite{Maldacena_1999}, such as the HaPPY code \cite{happy_2015}. In essence, a HaPPY tensor network $T$ defines an isometric encoding $T : \mathcal{H}_{\text{bulk}} \to \mathcal{H}_{\text{bdy}}$ from a collection of ``bulk'' logical qudits on a finite hyperbolic lattice to ``boundary'' physical qudits arranged on a one-dimensional ring. In the AdS/CFT interpretation, $\mathcal{H}_{\text{bulk}}$ models a code subspace of low-energy quantum-gravitational degrees of freedom on a spatial slice of AdS, while $\mathcal{H}_{\text{bdy}}$ models the Hilbert space of the dual CFT on the boundary circle. The quantum error-correcting structure of $T$ discretely implements bulk-boundary reconstruction: local bulk operators can be reconstructed on many overlapping boundary regions, and the radial direction of the network plays the role of a renormalization group scale~\cite{Jahn_2021}.

In our setting, we contract all bulk legs with $\ket{0}$ states except those on an inner disk, and we use the remaining bulk legs on that disk to model a black hole (see Fig.~\ref{fig:happy_code}(a)). Concretely, on this inner disk we insert a state $\rho_{\text{BH}}$ on the corresponding bulk Hilbert space, which is either a fixed Haar-random pure state (a single microstate) or the maximally mixed state (a toy microcanonical ensemble). The resulting boundary state is then a holographic encoding of either a pure or mixed black hole in the bulk, and our operational task is to decide, from noisy measurements of the boundary CFT degrees of freedom alone, which case holds. Because the black hole degrees of freedom are protected by the HaPPY code and only indirectly exposed to noise through the boundary, a suitably designed SWAP-test-type procedure acting on (approximately) decoded bulk modes (see Fig.~\ref{fig:happy_code}(b)) is partially robust to noise, and we show that this protection suffices to recover a quantum advantage. This is captured in the theorem below.

\begin{theorem}[Purity testing for holographic black holes in the HaPPY code, informal] \label{thm:happy_code_informal}
Consider a holographic HaPPY tensor network of total radius $R$ with no uncontracted bulk legs, and remove all tiles within a smaller radius $r$, so that the resulting uncontracted bulk legs are replaced by either a fixed Haar-random pure state (a toy black hole microstate) or the maximally mixed state on the corresponding bulk Hilbert space (a toy black hole microcanonical ensemble). Given copies of the resulting boundary state, there exists a quantum-enhanced protocol which, even in the presence of constant-strength erasure error on every qubit at each circuit layer, uses only a constant number of copies together with joint measurements to distinguish these two cases with high probability.  By contrast, any conventional experiment restricted to single-copy measurements requires at least $2^{\exp(\Omega(r))}$ copies to do so.
\end{theorem}
\noindent This result should be viewed as a holographic counterpart of our noisy purity-testing lower bound. On the one hand, the HaPPY code converts local boundary erasures into highly suppressed logical errors on the bulk black hole degrees of freedom: as long as $R \gtrsim r + O(\log r)$ and the erasure rate is below threshold, a greedy decoder can approximately recover $\rho_{\mathrm{BH}}$ from the noisy boundary state with failure probability that is exponentially small in $4^{R-r}$. More broadly, it is believed (but not known) that HaPPY-type holographic codes may admit fully fault-tolerant realizations against local noise \cite{numerics_happy_decoding_2020, Jahn_2021, Farrelly_2021}; in our setting we only appeal to their rigorously understood erasure-correction properties. Composing the decoding map with a two-copy \textsf{SWAP} test on the recovered bulk region therefore reproduces, up to small decoding errors, the ideal two-copy purity test and yields a constant-copy quantum protocol. On the other hand, any protocol restricted to single-copy measurements on the boundary reduces, via bulk-boundary isometry and data-processing, to single-copy purity testing on an $L_r$-qubit system, where $L_r = \Theta(4^{r-1})$ is the number of bulk legs in the excised region. Our general lower bound for noisy single-copy purity testing then implies a sample complexity of order $2^{\Theta(L_r)} = 2^{\exp(\Theta(r))}$, establishing the separation claimed in the theorem.

Our HaPPY code example illustrates the first of two strategies for useful noisy quantum learning, namely targeting physical systems with latent error-resilient properties.

\subsection{Noise-dependent advantage in Pauli shadow tomography and metrology}
\label{subsec:quantumlearningadv}

To further understand how noise reshapes quantum learning advantages, we now consider the well-studied problem of Pauli shadow tomography, namely estimating expectation values of (potentially mutually noncommuting) Pauli observables, which in ideal settings exhibits an exponential memory-sample tradeoff \cite{chen2021exponentialseparationslearningquantum}. Any noiseless conventional protocol restricted to single-copy measurements, or to $k < n$ qubits of quantum memory, requires a number of samples exponential in $n - k$, whereas a quantum-enhanced learner with two-copy access (i.e.~$n$ ancillary qubits of quantum memory) succeeds with only $O(n)$ samples using Bell-basis measurements; these two settings are depicted in Fig.~\ref{fig:quantum_memory}(a) (memoryless or small-memory single-copy experiments) and Fig.~\ref{fig:quantum_memory}(b) (architectures with a quantum memory register enabling multi-copy measurements). In Theorem~\ref{thm:multi_advantage}, we refine this separation in the $\nbqp$ setting, quantifying the noise-dependence of lower bounds with and without ancillary quantum memory. Notably, even when $n$ qubits of memory are provided, enabling two-copy measurements, we show that order $(1-\lambda)^{-n}$ samples are necessary. We give a noisy 2-copy algorithm achieving this scaling up to a constant factor in the exponent, establishing a quantum advantage in noisy Pauli tomography that depends polynomially on the noise rate, and negating the exponential speedup of ideal Bell-measurement-based strategies.

\begin{theorem} [Complexity of noisy Pauli shadow tomography, informal] \label{thm:multi_advantage}
    In the presence of constant depolarizing noise per qubit, any quantum algorithm without ancillary quantum memory which can identify a Pauli-structured state with high probability requires order $2^nf(\lambda)^n$ measurements, where $f(\lambda) \in [1, \infty)$ for $\lambda \in [0, 1]$. Given an additional $k\leq n$ qubits of quantum memory, $\Omega(2^{n-k}(1-\lambda)^{-n})$ samples are still required. When $k=n$, there exists a quantum-enhanced learning algorithm with access to noisy two-copy measurements solving the Pauli-identification task using $\tilde{O}((1-\lambda)^{-4n})$ samples.
\end{theorem}

\noindent As an immediate corollary, we find a sample complexity separation between our two-copy algorithm and any single-copy strategy for the same Pauli identification task. This separation depends on the noise rate, and interpolates smoothly between the ideal exponential separation for Pauli tomography and a polynomial (in the noise rate) separation for $\Theta(1)$ local noise. A related separation is presented in Ref.~\cite{huang2022foundations}; that work considers diamond-norm gate noise and access to perfect state copies, a weaker noise model than $\nbqp$, and gives a lower bound which considers only noiseless conventional experiments.

\begin{corollary} [Noise-dependent quantum advantage in Pauli identification]
\label{cor:1} 
    Let $N_{SC}$ be the optimal sample complexity for any single-copy algorithm for the $n$-qubit Pauli identification task from Theorem \ref{thm:multi_advantage}, and let $N_{TC}$ be the sample complexity of the two-copy algorithm given in Theorem \ref{thm:multi_advantage}. Then $N_{SC} = \Omega\big(N_{TC}^{a(\lambda)}\big)$, where $a(\lambda)$ is a function of noise rate such that when $\lambda = \Theta(1)$ and $n$ is large, $a(\lambda) = \Theta(\lambda^{-1})$, and as $\lambda\rightarrow 0$, $N_{SC} = \exp(\Omega(n))$ while $N_{TC} = \textnormal{poly}(\Theta(n))$.
\end{corollary}

\noindent Our lower bounds span three models for quantum learning algorithms utilizing single-copy measurements. Namely, we consider (i) a learner without any ancillary quantum memory, (ii) one with $k<n$ qubits of quantum memory, but where the memory must be reset after a constant number of queries and classical advice may be passed between experiments, and (iii) $k\leq n$ qubits of memory with unbounded lifetime. To utilize the learning tree formalism, we introduce another hypothesis testing problem (the Pauli identification task discussed above): given copies of the state $\rho = (\mathds{1} + P)/\tr(\mathds{1}+P)$, can we distinguish between the case where $P$ is sampled uniformly from all $4^n-1$ non-identity Paulis and where $P$ is the $n$-qubit identity matrix? 

\begin{figure}[t!]
    \centering
    \includegraphics[width=\linewidth]{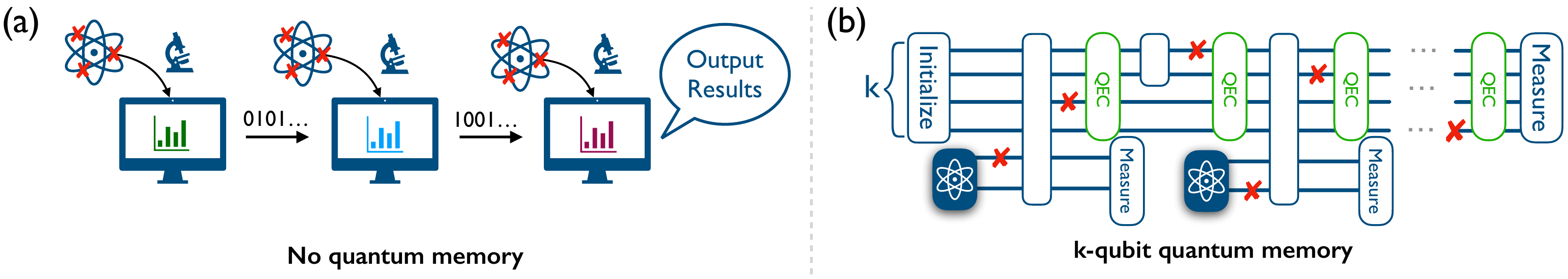}
    \vspace{1pt}
    \caption{(a) \textit{Memoryless protocol}: each noisy copy of the unknown state is measured immediately and only classical bit strings are stored, so different copies are never jointly entangled in the device.
(b) \textit{Protocol with $k$-qubit quantum memory}: a register of $k$ qubits is initialized, repeatedly interacts with fresh noisy copies of the state, and is stabilized by intermittent QEC cycles (green boxes) before a final joint measurement. Red crosses indicate local noise events on the physical qubits.}
    \label{fig:quantum_memory}
\end{figure}

We next define learning trees for each experimental model. For models (i) and (ii), the output of each node of the tree is a classical bitstring. For both, we show that the distribution over bitstrings induced by the results of each experiment, given access to a $\mathds{1}+P$-type state, concentrate around the distribution resulting from experiments which instead query a maximally mixed state. Hence, every experiment is uninformative in distinguishing the two hypotheses, and exponentially many measurements are required to achieve an algorithm output distribution under one hypothesis that is far in total variation from the other (Theorem \ref{thm:memoryless_lb} and Theorem \ref{thm:c,m_lower_bound}). The relationship between concentration of node-wise distributions and the final leaf output distribution is given by the martingale formalism developed in \cite{weiyuan_paulis}. In model (ii), obtaining this bound requires quantifying the correlation between copies of the unknown state in terms of the size of the quantum memory and its lifetime, which we accomplish using a reformulation of the learning model in terms of Matrix Product States (Lemma \ref{lemma:dc_moments}). 

For model (iii), nodes of the tree are joined by the state of the ancillary memory register after every measurement rather than a simple classical bitstring. Thus, we require a stronger approach: bounding the variation in entire root-to-leaf paths in the learning tree, following the approach from \cite{chen2021exponentialseparationslearningquantum}. Using a probabilistic argument, we control the total variation of the leaf output distribution by bounding the number of paths which diverge substantially from the output distribution induced by the maximally mixed input state. To account for depolarizing noise in every copy of the state, we simplify the tensor-network analysis of \cite{chen2021exponentialseparationslearningquantum} via a technical argument which provides a cleaner form of the action of our noise channel. With this approach, we find that the contribution of the noise to the sample complexity decouples from the number of memory qubits; hence, even given $n$ ancillary memory qubits, $\Omega((1-\lambda)^{-n/3})$ samples are still required in the noisy regime. We provide a two-copy algorithm which leverages noisy Bell sampling to match this asympotic lower bound up to constants in the exponent, and give a single-copy classical shadow algorithm for the broader Pauli shadow tomography which accounts for noise by using shallow-circuit unitary ensembles.  Up to polynomial factors and constants in the exponents, our upper and lower bounds for the Pauli-identification task exhibit the same exponential dependence on $n$, $k$, and the noise rate $\lambda$, matching the known noiseless bounds from \cite{chen2021exponentialseparationslearningquantum,weiyuan_paulis} as $\lambda \to 0$ and diverging to infinity as $\lambda \to 1$, where the task becomes information-theoretically impossible.

From the standpoint of near-term experiments, the most relevant regime is often not the asymptotic limit $n\to\infty$ at fixed noise. Instead, the instance size $n$ is effectively fixed by hardware, and the central question is how performance changes as other parameters, especially the noise rate $\lambda$, are improved. Our proof of Corollary~\ref{cor:1} supports this viewpoint: for any fixed $n$, it gives an explicit $\lambda$-dependent separation between the optimal single-copy sample complexity $N_{SC}$ and the two-copy sample complexity $N_{TC}$, and thus provides a quantitative condition on $\lambda$ for obtaining a meaningful two-copy advantage. This perspective is consistent with the experimental results of Ref.~\cite{Huang_2022}, which indicate that at modest instance sizes and sufficiently low noise, two-copy Pauli shadow protocols can substantially outperform single-copy baselines.

Our noise- and resource-aware analyses extend beyond state learning. Figure~\ref{fig:HL_sensing} contrasts ideal error-corrected metrology, which can sustain Heisenberg-limited scaling indefinitely, with the $\nbqp$ setting, where noisy probe-control couplings confine Heisenberg-limited performance to a finite sensing window before a crossover to the Standard Quantum Limit. In Theorem~\ref{thm:HL_sensing_informal}, we bound this precision-time tradeoff of Heisenberg-limited (HL) single-parameter metrology protocols under an effective noise model similar to \nbqp.

\begin{figure}[t!]
    \centering
    \includegraphics[width=\linewidth]{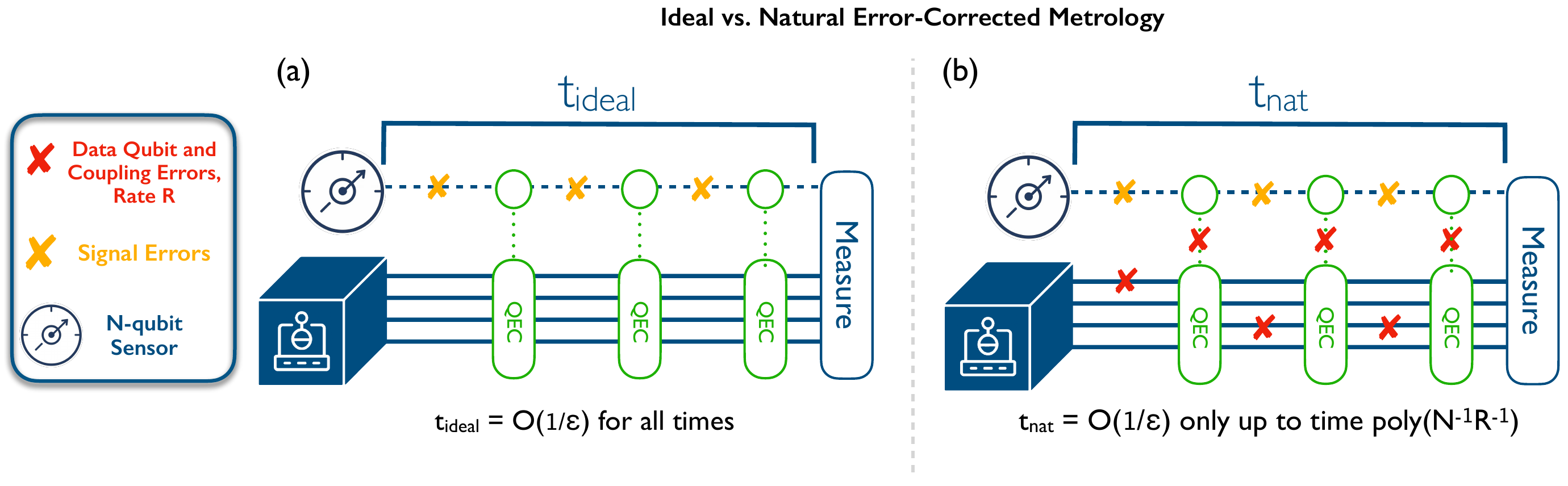}
    \caption{(a) \textit{Ideal error-corrected metrology}: interleaved quantum error correction (QEC) cycles protect an $N$-qubit probe during interrogation, sustaining Heisenberg-limited scaling in a noiseless sensing-control interface. (b) \textit{Natural error-corrected metrology}: constant-rate faults at the sensing-control interface (red crosses, rate $R$) accumulate and restrict Heisenberg-limited performance to a finite sensing window (scaling to leading order as $\mathrm{poly}(N^{-1}R^{-1})$), beyond which the protocol crosses over to the Standard Quantum Limit.}
    \label{fig:HL_sensing}
\end{figure}

\begin{theorem}[A threshold for Heisenberg-limited metrology, informal] \label{thm:HL_sensing_informal}
Consider a single-parameter quantum sensing protocol in which a Hamiltonian $H(\omega)$ acts on an $N$-qubit probe, achieving Heisenberg-limited sensitivity when the sensing–control interface is noiseless. Under depolarizing noise of rate $R(\tau, \eta)$ (determined by the speed of control $\tau$ and native error strength $\eta$) at this interface, and without additional error mitigation, the protocol achieves the Heisenberg limit only up to a total evolution time inverse-polynomial in $NR$, beyond which it crosses over to the Standard Quantum Limit.
\end{theorem}
\noindent This result illustrates the value of noise-aware analysis beyond state preparation and measurement, namely that SPAM-robust metrology protocols can still fail under interstitial noise. It also highlights the connection between noisy quantum learning theory and practical metrology, demonstrating that even with access to quantum error correction, Heisenberg-limited metrological gains are delicate and tightly coupled to hardware constraints. The formal statement (Theorem \ref{thm:HL_sensing}) expresses the threshold sensing time in terms of control speed, native noise strength, and quantum memory size. When all parameters except $N$ are $\Theta(1)$, HL sensitivity is essentially completely lost, as coherence-destroying errors accumulate on the same timescale as quantum control is executed in a single round. More generally, our bounds clarify the tradeoff between experimental resources and quantify how error rates must scale with $n$ to sustain an HL protocol over meaningful total time. For instance, an effective noise strength of $\eta \sim 1/N$ and $N\eta\tau \ll 1$ enables HL sensitivity for $O(1)$ time.

\section{Outlook}
\label{sec:outlook}

We have studied quantum-enhanced experiments in the presence of noise through the lens of quantum learning theory. Our results show that many of the most striking idealized quantum advantages in learning, property testing, and metrology disappear, or become inapplicable, once interstitial noise and unprotected oracle systems are taken into account. At the same time, we identified settings where advantages survive, often in a weakened but still meaningful form, and clarified how these surviving advantages depend on noise strength, memory resources, and problem structure.

Our work suggests several concrete directions for further development of noisy quantum learning theory. First, progress will require engaging more deeply with the \emph{physics} of quantum many-body systems, rather than treating oracles as abstract. Many natural systems possess built-in robustness: renormalization-group structure and locality can protect long-range correlations~\cite{Kim_2017,Furuya_2022,Furuya_RGAQECC, goldman2024lindbladianexactrenormalizationdensity,Lake_2025}, and thermalizing dynamics can generate noise-resilient macroscopic observables \cite{Kastoryano_2025, Brandao_2019, Chesi_2010}. It will be important to identify and characterize classes of states, channels, and observables whose relevant features are intrinsically stable under realistic error models, and to turn such structure into quantitatively sharp, physically natural examples of noisy quantum advantage.

Second, most known exponential quantum advantages in learning rest on highly entangled, global measurements, such as large-scale SWAP tests or Bell-basis measurements \cite{Aharonov_Cotler_Qi_2022, huang2022foundations, chen2021exponentialseparationslearningquantum}, that are simultaneously fragile to noise and misaligned with the local-control constraints of real devices which can couple to and manipulate experimental systems. This points toward a theory of noisy quantum learning under \emph{restricted resources}, in which allowed operations may be shallow, geometrically local, few-qubit, or constrained to a small number of probe systems. Quantum probe tomography \cite{chen2025quantumprobetomography} is an example of this philosophy, wherein a small number of probes interacting locally with a large system can still extract global information under noise. It will be important to understand systematically which restricted-access models admit robust (even if only polynomial) advantages over conventional experiments, and how those advantages trade off against architectural constraints.

There is a close connection between noisy quantum learning and quantum sensing. Many sensing protocols can be viewed as learning problems about Hamiltonian parameters or state observables under stringent access and noise constraints \cite{Zhou_2018,Huang_2023,Hu_2025}. Some ambitious proposals, such as quantum computation-enhanced sensing based on deep coherent control \cite{Huang_2023, Hu_2025} or Grover-type oracles \cite{allen2025quantumcomputingenhancedsensing}, lose their asymptotic advantage in the \textsf{NBQP} noise model. This raises a quantitative question rather than a purely asymptotic one: \emph{how} does the achievable quantum advantage degrade as a function of the noise rate, circuit depth, and available memory? Corollary~\ref{cor:1}, together with related work such as Ref.~\cite{huang2022foundations}, suggest that in realistic noise regimes one should expect noise-dependent polynomial improvements rather than noise-independent exponential ones, but those polynomial advantages can still be practically meaningful.

Taken together, these directions aim at a more operational understanding of what fault-tolerant quantum computers can teach us about real-world quantum systems that are neither fully protected nor perfectly characterized. Our results indicate that quantum advantages do remain available in this regime, but they are more delicate, more problem-dependent, and more tightly coupled to physical structure than in idealized oracle models. Developing a mature noisy quantum learning theory along these lines should inform both the design of near- and medium-term experiments and the long-term role of quantum computers as scientific instruments.

\subsection*{Acknowledgments}

The authors thank Sitan Chen, Soonwon Choi, Harald Putterman, Ruohan Shen and Nikita Romanov for valuable discussions. JC is supported by an Alfred P.~Sloan Foundation Fellowship. WG is supported by NSF Grant CCF-2430375. The authors do not declare any competing interests.

\appendix

\vspace{2.5em}
\noindent \textbf{\LARGE{}Appendices}
\vspace{0.25em}

\paragraph{Roadmap for Appendices.} In Appendix~\ref{sec:related}, we review salient prior work on quantum sensing, learning, and computation in noisy settings. In Appendix~\ref{app:defns}, we recall the definitions of several hybrid, relativized models of quantum computation used in our complexity-theoretic separations, and define the complexity class $\nbqp$. In Appendix \ref{Appendix:basic_quantum_prelim}, we review technical preliminaries on quantum information theory as well as quantum learning lower bounds. In Appendix \ref{app: complexity_sep}, we demonstrate our two superpolynomial oracle separations, proving Theorems \ref{thm:NBQP_vs_BQP} and \ref{Thm:NISQ_vs_NBQP}, which comprise the complexity-theoretic portion of our results. In Appendix \ref{appendix:purity_testing_lb}, we demonstrate that the well-studied exponential quantum speedup for purity testing completely degrades under local noise, proving Theorem \ref{thm:purity_testing_lb_informal}. We then turn to a physically-motivated reformulation of the purity-testing problem, demonstrating that the breakdown we showed previously can be rectified for the task of detecting a black hole microstate in the bulk of a tensor-network model for holographic duality, since the system has intrinsic error-correction properties. In Appendix \ref{app:quantum_advantage}, we address the degradation of Pauli shadow tomography with noisy experiments, proving Theorem \ref{thm:multi_advantage}. In particular, we prove tight sample complexity lower bounds under three models of noisy quantum experiments, and provide noisy single-copy and two-copy algorithms for the task. In Appendix~\ref{app:deferred}, we prove several technical lemmas stated in earlier sections. 

\section{Related Work}
\label{sec:related}

\paragraph{Quantum sensing enhanced by quantum information processing.}
Quantum sensing has a natural interpretation in the language of noisy quantum learning theory, since sensing targets are often Hamiltonian coefficients (e.g.\ field amplitudes or phases) or observables of quantum states (such as order parameters of quantum materials). In a future fault-tolerant setting, the most general sensing protocols could use a sensor to couple the experimental system, viewed as an oracle, to a quantum computer, thereby providing oracle access for an \textsf{NBQP} computation.  Many works have established Heisenberg-limited (HL) scaling for traditional sensing methods such as Ramsey spectroscopy and interferometry in idealized settings \cite{D_Ariano_2003, De_Martini_2003, Giovannetti_2004, Kok_2004}. However, under standard noise models, such as Markovian reservoirs or photon loss channels, the HL asymptotic scaling of these methods often reverts to the standard quantum limit (SQL) \cite{Yu_2004, Bellomo_2007,Monras_2007, Smirne_2016,Haase_2018, Jiao_2023}. 

Quantum information processing has been used to either recover HL sensitivity in noisy environments or surpass HL in ideal settings. Quantum error correction on sensor qubits, introduced in \cite{Kessler_2014}, underlies much of the first direction. Subsequent work has developed error-correction protocols that recover asymptotic HL scaling for particular noise channels, such as dephasing or bosonic loss, and in specific architectures such as trapped-ion sensors \cite{2017NatCo...8.1822R,Layden_2019}. Ref.~\cite{Zhou_2018} provides a necessary and sufficient condition for such recoverability, showing that when the noise operators span the observables of interest, HL scaling typically cannot be restored and the protocol reverts to the SQL. In the \textsf{NBQP} model, we therefore expect many error-corrected sensing protocols to fail to retain HL asymptotics, since e.g.~the depolarizing channel has a Kraus decomposition consisting of all $4^n$ Pauli strings. This intuition is captured in our Theorem \ref{thm:HL_sensing}, in which the well-known error-correction-based metrology scheme from \cite{Zhou_2018} is degraded to the SQL beyond a sensing time threshold. A concurrent work \cite{sahu2026achievingheisenberglimitusing} demonstrates that for sufficiently structured signals and noise channels, quantum error correction can recover HL sensitivity for noise rates below a threshold. In such specialized settings, our noisy HL sensing window can be extended indefinitely, at the expense of introducing a new threshold condition on the error rate. Practically, error-mitigation and approximate error-correction techniques remain promising for metrological gains despite asymptotic limitations \cite{Jeske_2014, Zhou_2020, Rossi_2020}, and careful co-design of error-correcting codes and sensing architectures with well-characterized noise can still enable HL sensing.

Sensing beyond the HL does not yet fit into a single unified framework. A valuable non-entanglement resource in some beyond-HL proposals is quadrature squeezing: by compressing uncertainty in a relevant quadrature while enlarging it in an irrelevant one, and measuring only the former, interferometric protocols can achieve super-HL scaling in idealized regimes \cite{Pezz__2008, Grote_2013, Oh_2024}. More recently, \cite{allen2025quantumcomputingenhancedsensing} combined Grover-type quantum speedups with sensing by searching for an ambient signal over discrete frequency bins, giving another beyond-HL framework that leverages deep quantum circuits. This work engineers a Grover phase oracle by wrapping the unknown signal with a quantum signal processing transform \cite{Martyn_2021}. In the $\nbqp$ noise model where calls to this oracle are interleaved with noise, this speedup breaks down due to known results regarding Grover's algorithm with noisy oracles \cite{ambainis_grover, rosmanis2023quantumsearchnoisyoracle, rosmanis2024addendumquantumsearchnoisy}. Refining this algorithm for practical applications will require novel error-mitigation strategies both within the oracle construction and surrounding the oracle queries. An important practical follow-up to our work is to connect our noise-aware analysis of learning from uncharacterized systems with these beyond-HL protocols, which are largely developed only in ideal settings.

\paragraph{Noise-robust quantum learning.}
Our Theorems~\ref{thm:NBQP_vs_BQP} and~\ref{thm:purity_testing_lb_informal} show that many examples of quantum advantage in learning from uncharacterized systems break down in the presence of noise. At their core, several advantages based on highly entangled multi-copy measurements, such as protocols built from \textsf{SWAP} tests or Bell-basis sampling, are intrinsically fragile to interstitial noise. This suggests that quantum learning algorithms which remain robust for generic inputs will need the number of entangling two-qubit gates to be independent of instance size. Such algorithms naturally employ measurements with limited entanglement, shallow circuits, or coherence-boosting resources other than entanglement. These properties are especially desirable in shadow tomography of realistic many-body systems.

In a related vein, the classical shadows algorithm \cite{classical_shadow} with the standard Clifford ensemble is not noise-robust: \cite{Koh_2022} shows that under product depolarizing and amplitude-damping noise, its sample complexity for estimating Pauli operators scales exponentially with operator weight. By contrast, randomized measurement protocols using only single-qubit local control (and hence generating no entanglement) are expected to be significantly more noise-robust, since corruption of a constant fraction of qubits remains localized. Examples include quantum overlapping tomography \cite{cotler2020quantum} and classical shadows with a unitary 1-design ensemble.

Hamiltonian learning from time dynamics is another quantum learning task where noise robustness is essential and closely tied to quantum metrology. Several works have proposed algorithms achieving $O(1/\epsilon)$, Heisenberg-limited scaling in time complexity that are robust to state-preparation and measurement (SPAM) errors \cite{Huang_2023, Hu_2025}. However, these protocols require deep circuits that interleave many layers of quantum control with queries to the unknown Hamiltonian. Moreover, \cite{Hu_2025} proves that such deep quantum control is \emph{necessary} for Heisenberg-limited, ansatz-free Hamiltonian learning. Under interstitial noise, accumulated errors then destroy coherence unless the Hamiltonian coincidentally acts on the logical codespace of an error-correcting code. In a different direction,~\cite{chen2025quantumprobetomography} introduces quantum probe tomography, where a constant number of probes couple to a few sites of a large system to extract its parent Hamiltonian. While restricted to structured Hamiltonian classes, this probe setting requires neither probe entanglement nor deep control, yielding an end-to-end noise-robust, Hamiltonian learning strategy.  However, it is not yet known if there are versions of quantum probe tomography which achieve the Heisenberg limit.  In any case, such protocols highlight that it is prudent for practical Hamiltonian learning protocols to leverage the substantial structural constraints of physical Hamiltonians.      

Quantum learning in bosonic continuous-variable systems has also recently attracted attention. Bosonic statistics allow squeezing operations that reduce canonical quadrature variance below the standard Heisenberg limit. Leveraging this, \cite{fanizza2025efficientlearningbosonicgaussian} provides an algorithm for learning bosonic Gaussian unitaries, while \cite{Oh_2024} studies bosonic random displacement channels; both algorithms exhibit a sample complexity that shrinks with the amount of squeezing used to prepare input states. Moreover, \cite{Oh_2024} demonstrates that the effects of realistic photon-loss, measurement, and crosstalk errors on estimation uncertainty are suppressed by a factor that grows exponentially with the squeezing parameter. Beyond such examples, the broader role of non-entanglement resources in quantum learning remains poorly understood.

Even without formal noise-robustness guarantees, small-scale experiments have demonstrated quantum learning-enhanced efficiency. For instance,~\cite{Cotler_2019} use a strategy based on \textsf{SWAP} tests to estimate low-temperature properties of a Bose-Hubbard model in an optical lattice. To measure local observables at low temperatures, one needs only perform a variant of the \textsf{SWAP} test on a constant number of sites. However, if the target low temperature is separated from the physical temperature by a phase transition, or if one wishes to access the phase transition directly, then one may need to perform a version of the \textsf{SWAP} test on a number of sites that scales with $n$. Our exponential lower bound on noisy two-copy purity testing in Appendix~\ref{appendix:purity_testing_lb} implicitly shows that such general \textsf{SWAP}-test-based strategies are exponentially hampered by local errors. We therefore expect a similar degradation for quantum virtual cooling when virtually cooling to at or below a phase transition, although the details will depend on the particular observables being measured and on whether the state-observable pair is intrinsically noise robust. To elaborate on this last point, it is possible that certain kinds of (macroscopic) observables are insensitive to certain kinds of local perturbations or errors, allowing reliable learning despite noise at every logical layer of a given quantum learning protocol. Clarifying the relationship between experimental learnability and the robustness of (macroscopic) observables to local errors remains an open question.

Finally, \cite{huang2022foundations} develops a theoretical framework for quantum learning with noisy quantum devices, in a setting distinct from ours. They study noisy device learning from noisy access, whereas we consider learning uncharacterized and noisy systems with a fault-tolerant quantum computer. Their work analyzes a task similar to the Dec-IP problem from Definition~\ref{def:dec_ip_problen}, which we use to establish quantum advantage with noisy access. However, \cite{huang2022foundations} considers access to perfect copies of the unknown state and diamond-error noise in the action of each gate for the upper bound; moreover, the given lower bounds neglect noise.

\paragraph{Quantum computation with noisy oracles.} The $\nbqp$ model also characterizes quantum computational problems making calls to noisy oracles, beyond the context of learning from quantum experiments. Ref.~\cite{Cross_2015} shows that a modification of the Bernstein-Vazirani problem retains a quantum speedup when only the output of the oracle is corrupted by depolarizing noise; it is simple to see that, under interstitial noise, this speedup vanishes. 

Several works study Grover's search algorithm \cite{grover1996fastquantummechanicalalgorithm} with various models of oracle noise, including phase error \cite{Shenvi_2003}, simple oracle call failure \cite{ regev2012impossibilityquantumspeedupfaulty, ambainis_grover}, and phase inversions \cite{long_grover}, all finding that the time complexity under constant-strength noise reverts to asymptotically linear in the database size. Furthermore,~\cite{nisq} demonstrates such a slowdown for \textsf{NISQ} algorithms. Most aligned with our work, \cite{rosmanis2023quantumsearchnoisyoracle} shows that if Grover oracle calls are sandwiched between layers of constant-strength global depolarizing or dephasing noise with all other operations perfect (a stronger computational model than \textsf{NBQP}), then the speedup once again breaks down, and in \cite{rosmanis2024addendumquantumsearchnoisy} the result is strengthened to hold when depolarizing noise is applied only to a single qubit before and after the oracle query. While Grover oracles for database search tasks may eventually be implemented fault-tolerantly, these results suggest that achieving asymptotic metrological gains from the recently proposed Grover-enhanced sensing strategy in \cite{allen2025quantumcomputingenhancedsensing} may be difficult when ambient noise acts on the sensor qubits, effectively implementing an error channel before each oracle query.

We remark that it is crucial to consider interstitial noise rather than noise applied only after oracle queries. Physically, if an experimental probe cannot be easily embedded into an error-correction scheme, the noise incurred when coupling to a quantum computer is unavoidably propagated through the computation; this includes errors occuring \textit{before} the application of the oracle. Mathematically, note that our Theorem \ref{Thm:NISQ_vs_NBQP} separating $\textsf{NISQ}$ and $\textsf{NBQP}$ relies on the fact that due to noise before the oracle query, it is exponentially unlikely that we successfully query the oracle within the relevant subspace; if instead errors only occurred after the oracle, a simple majority-vote strategy would obviate the separation. The aforementioned conception of interstitial errors is at the heart of our remark on ``logical locality" in Appendix \ref{app:remark}.

\section{Definitions} \label{app:defns}
In this Section, we recall the definition of the complexity class \textsf{NISQ}, then formally define the complexity class \textsf{NBQP}.

We begin by recounting the definition of a classical oracle.
\begin{definition}[Classical oracle]
A classical oracle $O$ is a function from $\{0, 1\}^n\rightarrow \{0, 1\}^m$ for $n,m \in \mathbb{N}$. The quantum instantiation of $O$ is the unitary $U_O$ acting on computational basis states $\ket{x}\!, \ket{y}$ as $U_O\ket{x}\ket{y} = \ket{x}\ket{y\oplus O(x)}$.
\end{definition}

\noindent For the definition of the $\textsf{NISQ}$ complexity class, we will need to make reference to a noise model for quantum computation. The most convenient is constant depolarizing noise per qubit, although as emphasized in~\cite{nisq} other noise models are suitable as well.  To fix notation, a single-qubit depolarizing channel will be defined as follows.
\begin{definition}[Single-qubit depolarizing channel] The single-qubit depolarizing channel with strength $\lambda\in [0, 1]$ is
\begin{equation}
    \mathcal{D}_\lambda(\rho) = (1-\lambda)\rho + \lambda\, \frac{I}{2}
\end{equation}
for any single-qubit density matrix $\rho$. The depolarizing channel is self-adjoint with respect to the Hilbert-Schmidt inner product, and on a general 2-by-2 matrix $A$, (the adjoint of) $\mathcal{D}$ acts as
\begin{equation}
    \mathcal{D}_\lambda(X) = (1-\lambda)A + \lambda\,\frac{\tr(A)\,I}{2}\,.
\end{equation}
\end{definition}

An additional useful primitive is sampling a noisy quantum circuit which has access to a classical oracle.  We make this precise below.
\begin{definition}[Sampling a noisy quantum circuit with classical oracle access]
Let $O$ be a classical oracle taking $n$-bit inputs. We denote the sampling of a noisy quantum circuit call with access to $O$ by $\text{\rm NQC}^O_\lambda(n', \{U\})$; this object takes in an integer $n' \geq n$, and a sequence of $T$ $n'$-qubit unitaries $\{U\} = \{U_1,...,U_T\}$ where each $U_i$ is either a depth-1 circuit or equal to $U_O\otimes I_{2^{n'-n'}}$, and outputs a random $n'$-bitstring $s$ sampled from the distribution 
\begin{equation}
    p(s) = \bra{s}\DN(U_T\DN (U_{T-1} \cdots (U_2\DN(U_1\DN(\ketbra{0^n}{0^n})U_1^\dagger)U_2^\dagger) \cdots U_T^\dagger)\ket{s}\,.
\end{equation}
This is the probability distribution induced by measuring the outcome of the circuit in the computational basis.  Each call to $\text{\rm NQC}^O_\lambda$ takes time $\Theta(T)$, wherein each call to the oracle takes unit time.
\end{definition}

\noindent With this definition, we can define a \textsf{NISQ} algorithm with oracle access.

\begin{definition}[\textsf{NISQ} algorithm with oracle access]
A \textnormal{\textsf{NISQ}}$_\lambda$ algorithm $A_\lambda^O$ with access to a classical oracle $O$ on $n$ bits is a probabilistic polynomial-time classical algorithm with $\mathrm{poly}(n')$ memory for $n' \geq n$ such that: \textnormal{(i)} the algorithm can query $O$ on any $n$-bit input, and \textnormal{(ii)} the algorithm can call noisy quantum circuits $\textnormal{NQC}^O_\lambda(n', \{U\})$ of at most polynomial depth.  The algorithm $A_\lambda^O$ runs in total time $T_c + \sum_i T_q^{(i)}$, where $T_c$ is the run time of the classical computation in the algorithm and $T_q^{(i)}$ is the (at most) polynomial running time of the $i$-th call to a noisy quantum circuit. Since the overall algorithm must run in polynomial time, the number of such calls (i.e.~the range of the index $i$) is itself bounded by a polynomial in $n'$.
\end{definition}

Using the notion of a \textsf{NISQ}$_\lambda$ algorithm with oracle access, we now define the corresponding (functional) relativized complexity class. Informally, \textsf{NISQ}$^O$ consists of all functions that can be computed with bounded error and in polynomial time by a \textsf{NISQ} algorithm for some fixed noise rate $\lambda > 0$, using at most polynomially many qubits.

\begin{definition}[\textsf{NISQ}$^O$ complexity class]
Let $f : \{0,1\}^\star \rightarrow \{0,1\}^\star$ be a (total) function. We say that $f$ is in \textnormal{\textsf{NISQ}}$^O$ if there exist a constant $\lambda > 0$, a polynomial $p(\cdot)$, and a \textnormal{\textsf{NISQ}}$_\lambda$ algorithm $A_\lambda^O$ (as in the previous definition) such that, for every input $x \in \{0,1\}^\star$ of length $n = |x|$:
\begin{enumerate}
    \item All calls made by $A_\lambda^O$ to noisy quantum circuits are of the form $\textnormal{NQC}^O_\lambda(n', \{U\})$ with $n \leq n' \leq p(n)$ and with circuit depth at most $p(n)$;
    \item The total running time of $A_\lambda^O(x)$ (its classical computation plus all noisy circuit calls) is at most $p(n)$, and its output has length at most $p(n)$;
    \item $A_\lambda^O(x)$ outputs $f(x)$ with probability at least $2/3$.
\end{enumerate}
\end{definition}
\noindent As is standard in complexity theory, we formulate \textnormal{\textsf{NISQ}}$^O$ here as a functional class; the associated decision version is obtained by restricting $f$ to have single-bit outputs and interpreting that bit as the yes/no answer.

In the $\textsf{NISQ}$ model, each use of the quantum device consists of a single noisy circuit call $\textnormal{NQC}^O_\lambda(n',\{U\})$ followed by a complete measurement of all qubits in the computational basis; the only memory across calls is classical. In particular, the algorithm cannot perform intermediate quantum measurements, reset a subset of qubits, or introduce fresh clean ancillas in the middle of a coherent computation. Quantum error-correction and fault-tolerance schemes rely precisely on such operations: one repeatedly measures stabilizers, discards or resets noisy ancillas, and thereby pumps entropy out of the encoded state. Since none of this is available in the \textsf{NISQ} model, one can prove (see e.g.~\cite{nisq}) that the effective noise on the quantum state cannot be suppressed over time, and the class is inherently incapable of implementing full fault-tolerant quantum error correction.

Next we turn to defining \textsf{NBQP}, which is a less restrictive model of noisy quantum quantum computation that does enable fault-tolerant quantum error correction.  We begin by defining a noisy quantum algorithm with oracle access.

\begin{definition}[Noisy quantum algorithm with oracle access]
\label{def:NBQP_circ}
Let $U_O$ be a quantum oracle acting on $n$ qubits. A $\lambda$-noisy quantum algorithm $Q_\lambda^O$ with access to $U_O$ is a uniform family of quantum channels $\{C_n^{U_O}\}_{n\in\mathbb{N}}$, where for each input length $n$ the channel $C_n^{U_O}$ acts on $n' \leq \mathrm{poly}(n)$ qubits. We refer to $n'$ as the total number of qubits used by the algorithm on inputs of size $n$.  Each $C_n^{U_O}$, which we call a $\lambda$-noisy quantum circuit, has the form
\begin{equation}
C_n^{U_O}[\rho] = V_{k,n}\,\mathcal{D}_\lambda\bigl( V_{k-1,n}\,\mathcal{D}_\lambda\bigl(\cdots V_{2,n}\,\mathcal{D}_\lambda\bigl(V_{1,n}\,\rho\,V_{1,n}^\dagger\bigr)V_{2,n}^\dagger \cdots \bigr) V_{k-1,n}^\dagger \bigr) V_{k,n}^\dagger,
\end{equation}
for some integer $k \leq \mathrm{poly}(n)$. For each $t = 1,\dots,k$, the unitary $V_{t,n}$ is either \textnormal{(i)} a depth-1 unitary on the $n'$ qubits, or \textnormal{(ii)} an oracle layer of the form $U_O \otimes I_{2^{n'-n}}$, where $U_O$ acts on some chosen subset of $n$ of the $n'$ qubits and the identity acts on the remaining $n' - n$ qubits.  For fixed $n$, the output state of $Q_\lambda^O$ is $C_n^{U_O}\bigl[\ketbra{0^{n'}}{0^{n'}}\bigr]$, and the runtime of the algorithm on inputs of length $n$ is $\Theta(k)$.
\end{definition}
\noindent When $\lambda$ is small, below the threshold of known fault-tolerant quantum error-correction schemes, Definition~\ref{def:NBQP_circ} permits the mid-circuit measurements, state preparations, and ancilla resets needed to implement full quantum error correction. In particular, one can encode information into a code subspace and perform (effectively) error-free logical quantum computation using fault-tolerant implementations of the required logical gates. However, a general oracle $O$ need not preserve any chosen code space or admit a fault-tolerant implementation. Thus, even though an $\textsf{NBQP}$ algorithm can protect most of its internal computation, a noisy quantum algorithm that seeks to learn properties of an arbitrary $O$ may still be fundamentally more limited than a noiseless quantum algorithm that seeks to do the same.  With this in mind, we define the \textsf{NBQP}$^O$ complexity class.

\begin{definition}[\textsf{NBQP}$^O$ complexity class]\label{def:NBQP}
Let $f : \{0,1\}^\star \rightarrow \{0,1\}^\star$ be a (total) function. We say that $f$ is in \textnormal{\textsf{NBQP}}$^O$ if there exist a constant $\lambda > 0$, a polynomial $p(\cdot)$, and a $\lambda$-noisy quantum algorithm $Q_\lambda^O$ with access to $U_O$ (as in Definition~\ref{def:NBQP_circ}) such that, for every input $x \in \{0,1\}^\star$ of length $n = |x|$:
\begin{enumerate}
    \item The algorithm $Q_\lambda^O$ on input $x$ acts on $n' \leq p(n)$ qubits, uses at most $p(n)$ layers (i.e.~$k \leq p(n)$ in Definition~\ref{def:NBQP_circ}), and hence has total running time at most $p(n)$;
    \item Measuring all qubits of the output state $C_n^{U_O}\bigl[\ketbra{0^{n'}}{0^{n'}}\bigr]$ in the computational basis yields a classical bit string of length at most $p(n)$;
    \item The resulting measurement outcome equals $f(x)$ with probability at least $2/3$.
\end{enumerate}
\end{definition}
\noindent As with \textnormal{\textsf{NISQ}}$^O$, we have formulated \textnormal{\textsf{NBQP}}$^O$ as a functional class.

\begin{remark}
We will sometimes denote $\textnormal{\textsf{NISQ}}^O$ and $\textnormal{\textsf{NBQP}}^O$ by $\textnormal{\textsf{NISQ}}_\lambda^O$ and $\textnormal{\textsf{NBQP}}_\lambda^O$ when we want to make a noise rate $\lambda$ explicit.
\end{remark}

\begin{remark}[Property testing with measure-and-prepare oracles]
In many applications, the oracle $O$ is a measure-and-prepare channel that, on each invocation, produces an $n$-qubit state $\rho_O$. In Definition~\ref{def:NBQP_circ}, we can replace the unitary oracle layers $U_O$ by applications of a measure-and-prepare channel.

For a two-property testing problem, suppose that for each $n$ there are two disjoint classes of states $\mathcal{P}_0(n)$ and $\mathcal{P}_1(n)$, and we are promised that $\rho_O$ belongs to exactly one of them. For each such oracle $O$, define a function $f^O : \{0,1\}^\star \to \{0,1\}$ by
\begin{align}
f^O(1^n) =
\begin{cases}
0 & \textnormal{if } \rho_O \in \mathcal{P}_0(n)\\
1 & \textnormal{if } \rho_O \in \mathcal{P}_1(n)
\end{cases}\,,
\end{align}
where the input is simply the unary string $1^n$ encoding the system size. An \textnormal{\textsf{NBQP}}$^O$ algorithm for this property-testing task is then a $\lambda$-noisy quantum algorithm $Q_\lambda^O$ that, on input $1^n$ and with access to the measure-and-prepare oracle $O$, outputs $f^O(1^n)$ with probability at least $2/3$ using only polynomial resources. Thus, standard two-property quantum state testing problems naturally fit into our functional notion of \textnormal{\textsf{NBQP}}$^O$.
\end{remark}

When not taken relative to an oracle, Definitions~\ref{def:NBQP_circ} and~\ref{def:NBQP} yield \textsf{NBQP} = \textsf{BQP}. In particular, if the noise rate $\lambda$ is chosen to be below a fault-tolerant quantum error correction threshold, all physical errors can be made fault-tolerantly correctable, so the model is computationally equivalent to standard noiseless quantum computation. In contrast, in the relativized setting the oracle in \textsf{NBQP}$^O$ represents an unknown quantum state or process that we wish to probe experimentally. In this regime, \textsf{NBQP}$^O$ naturally captures quantum learning and property-testing tasks, and it is precisely the errors occurring during oracle interactions (which cannot, in general, be coherently corrected) that lead to the separations we establish.

To prove Theorem \ref{Thm:NISQ_vs_NBQP} from the main text we will need to define complexity classes of hybrid algorithms with access to bounded-depth noiseless quantum computation. 

Following the convention in \cite{Chia_2023},
\begin{definition}[\textsf{QNC}$_d$]
A \textnormal{\textsf{QNC}}$_d$ circuit family is a collection $\mathcal{C} = \{\mathcal{C}_n\}_{n\in \mathbb{N}}$, where each $\mathcal{C}_n$ is a quantum circuit on $n$ input qubits (and possibly additional ancillas initialized to $\ket{0}$) consisting of $d$ layers of depth-1 unitaries. On input $x \in \{0,1\}^n$, we prepare the first $n$ qubits in $\ket{x}$, all ancillas in $\ket{0}$, apply the $d$ layers, and then measure all (or a designated subset of) qubits in the computational basis to obtain a classical bit string $s$. The circuit $\mathcal{C}_n$ has depth $d$ and thus runs in time $\Theta(d)$.

We say that a (total) function $f : \{0,1\}^\star \to \{0,1\}^\star$ is in \textnormal{\textsf{QNC}}$_d$ if there exists a \textnormal{\textsf{QNC}}$_d$ circuit family $\{\mathcal{C}_n\}$ such that, for every input $x \in \{0,1\}^n$, the output $s$ of $\mathcal{C}_n$ on input $x$ equals $f(x)$ with probability at least $2/3$.
\end{definition}
\noindent \textnormal{\textsf{QNC}}$_d^O$ circuit families are defined analogously to Definition~\ref{def:NBQP_circ}, except that the total number of layers, counting both depth-1 unitary layers and oracle layers using $U_O$, is bounded by $d$.

To formalize hybrid classical-quantum algorithms that may make adaptive, bounded-depth noiseless quantum queries, we introduce a \textsf{BPP}-type model with access to \textsf{QNC}$_d^O$ circuits as subroutines.

\begin{definition}[\textsf{BPP}$^{\textsf{QNC}_d}$ algorithm with oracle access]
A $\textnormal{(\textsf{BPP}}^{\textnormal{\textsf{QNC}}_d}\textnormal{)}^O$ algorithm $H^O$ with access to a classical oracle $O$ is a probabilistic polynomial-time classical algorithm that, on input $x \in \{0,1\}^n$,
\begin{enumerate}
\item Can query $O$ on any (polynomially bounded) classical string; and
\item May, at any point during its computation, specify and call a (possibly different) $\textnormal{\textsf{QNC}}_d^O$ circuit on some number of input qubits $n' \leq \mathrm{poly}(n)$ to obtain a classical output bit string.
\end{enumerate}
We refer to each such execution of a $\textnormal{\textsf{QNC}}_d^O$ circuit as a \emph{quantum query}. Let $Q$ denote the number of quantum queries made on input $x$, and let $T_c$ be the classical running time of $H^O$ excluding the time spent inside these quantum subroutines. Since each $\textnormal{\textsf{QNC}}_d^O$ circuit has depth at most $d$, a single quantum query takes time $\Theta(d)$, and hence the total running time of $H^O$ is $T_c + Q \cdot \Theta(d)$; that is, the sum of its classical running time and the time spent in its (adaptively chosen) bounded-depth quantum subroutine calls.
\end{definition}

\section{Preliminaries}
\subsection{Quantum information theory toolkit} \label{Appendix:basic_quantum_prelim}
Here we collect the standard notions from quantum information theory used in this work, and record several lemmas capturing their key properties. Throughout, we use $\mathds{1}_m$ to denote the $2^m \times 2^m$ identity matrix, and use $\mathds{1}$ or $\mathds{1}_n$ interchangeably for the $2^n \times 2^n$ case.
\begin{definition}[POVM]
An $n$-qubit Positive Operator-Valued Measure (POVM) is given by a set of matrices $\{F_s\}$ such that all $F_s$ are positive semi-definite and $\sum_s F_s = \mathds{1}_{n}$. Given a density matrix $\rho$, when we say we measure $\{F_s\}$ on $\rho$, we obtain the classical outcome $s$ sampled from the distribution $\textnormal{Pr}[s] = \tr(F_s\rho)$.
\end{definition}
\noindent A well-known fact is that the classical outcome distribution of an arbitrary POVM can be simulated by a POVM consisting only of rank-1 matrices:
\begin{lemma}[Simulating arbitrary POVMs with rank-1 POVMs, e.g.~Lemma 4.8 in \cite{chen2021exponentialseparationslearningquantum}]
If we neglect the post-measurement quantum state, the outcome distribution of any arbitrary $k$-qubit POVM can be simulated (using classical postprocessing) by a POVM of the form $\{w_s2^n\ketbra{\psi}{\psi}\}$, where $\ket{\psi}$ is a pure quantum state and $\sum_s w_s = 1$. 
\end{lemma}
\noindent This lemma tells us that in quantum learning tasks where we discard a state after measurement, we only need to consider rank-1 POVMs.

An object we use to prove lower bounds for protocols with bounded memory are Matrix Product States (MPS). 
\begin{definition}[MPS]
    An $n$-qubit, $c$-qudit matrix product state (MPS) with bond dimension $k$ is a quantum state of the form
    \begin{equation}
        \ket{\psi} = \sum_{\{s\}}\tr[A_1^{(s_1)}A_2^{(s_2)} \cdots A_c^{(s_c)}]\ket{s_1 \cdots s_c}
    \end{equation}
    where every $s_i\in \{0,...,2^n-1\}$, $A_1^{(s_1)}$ is a $1\times k$ matrix, $A_i^{s_i}$ is $k\times k$, and $A_c^{(s_c)}$ is $k\times 1$ for all $i$.
    The set of all states of this form is $\textnormal{MPS}(n, k, c)$. For any $r\in [c]$, such an MPS can be written as
    \begin{equation}
        \ket{\psi} = \sum_{i=1}^{2^k} \sqrt{\lambda_k}\ket{\alpha_i}\otimes \ket{\beta_i}
    \end{equation}
    where the sets $\{\alpha_i\}$, $\{\beta_i\}$ are orthonormal bases supported on the first $r$ and last $c-r$ qudits. Then the set of all $n$-qubit MPS of bond dimension $k$ is $\textnormal{MPS}(n, k) =\bigcup_{c=2}^n \textnormal{MPS}(n, k, c)$.
\end{definition}
The bond dimension of an MPS captures the amount of irreducible entanglement contained within every qudit subsystem. In this work, we will consider models of learning in which an algorithm can entangle copies of an unknown state with a quantum memory register multiple times, generating entanglement between the two. When this is done sequentially, the memory register acts to simulate virtual entanglement between many copies of the state, even when no explicit quantum gate is applied simultaneously to the copies. The following definition and lemma formalize this concept.
\begin{definition}[$\mathcal{M}_{k}^{cn}$ POVMs]
    The set $\mathcal{M}_{k}^{cn}$ is the set of POVMs of the form $\{2^{cn}w_s\ketbra{L_s}{L_s}\}$, with the requirement that all $\ket{L_s}\in \textnormal{MPS}(cn, n+k)$. 
\end{definition}
\begin{lemma}[$\mathcal{M}_{k}^{cn}$ POVMs are $c$-query $k$-qubit circuits, Section 8.1 in \cite{weiyuan_paulis}] 
    Consider a quantum algorithm with access to copies of an $n$-qubit state $\rho$ and $k$ additional qubits of quantum memory initialized in the state $\Sigma_0$. The algorithm can perform quantum gates on individual copies of $\rho$ and the $k$-qubit memory register, but cannot jointly process multiple copies of $\rho$ at once (as in Figure \ref{fig:quantum_memory}(a)). Suppose a quantum circuit run by the algorithm measures at most $c$ copies of $\rho$. Then the outcome of any such algorithm is equivalent to measuring some POVM from $\mathcal{M}_k^{cn}$ on the state $\rho^{\otimes c}\otimes \Sigma_0$.
\end{lemma}
Another common tool we use are identities relating the Pauli operators to permutations.
\begin{definition}[Pauli Operators]
The single-qubit Pauli operators are
\begin{equation}
    X = \begin{pmatrix}0 & 1\\ 1 & 0\end{pmatrix},\quad
    Y = \begin{pmatrix}0 & -i \\i & 0\end{pmatrix}, \quad
    Z = \begin{pmatrix}1 & 0\\ 0 & -1\end{pmatrix},\quad
    I = \begin{pmatrix}1 & 0\\ 0 & 1\end{pmatrix},\quad
\end{equation}
The $n$-qubit Pauli group $\mathcal{P}_n$ is the set of $4^n$ elements $\{I, X, Y, Z\}^{\otimes n}$. We also denote $I^{\otimes n}$, the $2^n$-dimensional identity matrix, by $\mathds{1}_n$.
\end{definition}
\noindent Then the following are standard facts about Pauli matrices which we use without proof.
\begin{fact}\label{fact:pauli_to_swap}
    Let $\textnormal{\textsf{SWAP}}_n := \textnormal{\textsf{SWAP}}^{\otimes n}$ be the swap operator acting on $2n$ qubits. Then the following identities hold:
    \begin{align}
        \sum_{P\in \mathcal{P}_n}&P\otimes P = 2^n\textnormal{\textsf{SWAP}}_n\,,  \quad
        \sum_{P\in \{X,Y,Z\}^{\otimes n}}P\otimes P = (2\textnormal{\textsf{SWAP}}_1 - \mathds{1}_2^{\otimes 2})^{\otimes n}\,.
    \end{align}
\end{fact}
\noindent We also make use of the \textsf{SWAP} trick:
\begin{lemma}[\textsf{SWAP} trick]
Given an $n$-qubit density matrix $\rho$ and a subset of qubits $I\subseteq [n]$,
\begin{equation}
    \tr(\textnormal{\textsf{SWAP}}_I (\rho\otimes\rho)) = \tr(\tr_{[n] \textnormal{\textbackslash} I}(\rho)^2)
\end{equation}
where $\textnormal{\textsf{SWAP}}_I$ is the $\textnormal{\textsf{SWAP}}_{|I|}$ operator acting on the sites specified by $I$.
\end{lemma}

We will make use of standard statistical divergences between discrete probability distributions: The total variation distance $d_{\text{\rm TV}}(p, q) = 1/2 \sum_x |p(x)-q(x)|$, the $\chi$-squared divergence $\chi^2(p\|q) = \sum_x (p(x)-q(x))^2/q(x)$, and the Kullback-Leibler (KL) divergence $\text{KL}(p\|q) = \sum_x p(x)\log(p(x)/q(x))$. 

\subsection{Tree representations for learning lower bounds}
A powerful tool for proving information-theoretic lower bounds on the sample complexity of quantum learning tasks is the tree representation of a quantum learning algorithm. In this work, we bring together the learning tree formalisms proposed in previous works for memoryless algorithms, bounded-memory algorithms with limited memory lifetime, and bounded-memory algorithms with unbounded coherent lifetime \cite{weiyuan_paulis, chen2021exponentialseparationslearningquantum, nisq}. These frameworks fall under the following definition.

\begin{definition}[General tree representation for a quantum learning algorithm]\label{def:general_learning_tree}
A quantum learning algorithm with access to a fixed $n$-qubit state $\rho$, $m$ qubits of quantum memory, and $\textnormal{poly}(n)$ classical memory can be represented as a rooted tree $\mathcal{T}$ with the following properties:
\begin{itemize}
    \item Each node $u$ in $\mathcal{T}$ is associated with a POVM $\{M_s^u\}$ on $k \leq n+m$ qubits, which may be drawn from some subset of all possible $n+m$-qubit POVMs.
    \item Each node $u$, at depth $d$ in the tree, has an associated probability $p_\rho(u)$ denoting the probability that the state of the algorithm is represented by $u$ after $d$ measurements.
    \item Each non-leaf node $u$ is joined to its children by edges $e_{u, s}$, where $s$ corresponds to the classical outcome of the measurement performed at $u$. For a child node $v$, the transition rule is given by
    \begin{equation}
        p_\rho(v) = p_\rho(u) \,\tr(\rho(M_s^u\otimes \mathds{1}_{n+m-k}) )\,,
    \end{equation}
    where the identity acts on any qubits not measured under the POVM. 
    \item Every root-to-leaf path has $T$ edges. 
\end{itemize}
To specify a learning tree, we specify the size of the quantum register and the set of allowed POVMs at each node.
\end{definition}

Our general strategy in proving lower bounds on the sample complexity of a learning task is using a reduction to a hypothesis distinguishing task. An algorithm for learning a property of some quantum state to high accuracy can always be used to distinguish two quantum states with different values of the property; hence, a bound on the cost of distinguishing implies one on the cost of learning.

\begin{definition}[Many-vs.-one distinguishing problem]
Given a quantum state $\rho$, suppose the following two events are realized with equal probability.
\begin{itemize}
    \item $\rho$ is the maximally mixed state on $n$ qubits, $\mathds{1}/2^n$.
    \item $\rho$ is sampled from some known distribution $\mathcal{D}$ over a specified set of candidate states $\{\rho_i\}$. 
\end{itemize}
A many-vs.-one distinguishing task is to decide which event occurred with high accuracy.
\end{definition}
We will leverage learning trees to prove query lower bounds for property testing problems (such as the many-vs.~one distinguishing problem) using  Le Cam's two point method~\cite{Yu1997LeCam}. The potential outcomes of a quantum learning algorithm for many-vs.-one distinguishing will be stored in classical memory, and encoded in the leaf nodes of the learning tree $\mathcal{T}$. Some leaves will correspond to outputting the maximally-mixed hypothesis, while others will guess the alternative event. If the two distributions over leaves $\ell$, namely $p_{\rho_x}(\ell)$ (given $\rho_x$ sampled from $\mathcal{D}$) and $p_{\mathds{1}/2^n}(\ell)$ (given the maximally mixed state) are very close to one another, our learning algorithm cannot successfully distinguish the two hypotheses with high accuracy. Formally:
\begin{lemma}[Le Cam's Two-Point Method]\label{lemma:le_cam}
Given a many-vs.-one distinguishing problem and learning tree $\mathcal{T}$ representing a quantum algorithm for this problem, the probability that the algorithm selects the correct hypothesis is upper bounded by
\begin{equation}
    \frac{1}{2}\sum_{\ell\in \textnormal{leaf}(\mathcal{T})} \Big|\mathbb{E}_\mathcal{D}[p_{\rho_x}(\ell)] - p_{\mathds{1}/2^n}(\ell)\Big|
\end{equation}
This divergence is the total variation distance $d_{\text{\rm TV}}(\mathbb{E}_\mathcal{D}[p_{\rho_x}(\ell)],p_{\mathds{1}/2^n}(\ell))$.
\end{lemma}
\noindent For a review of information-theoretic lower bound techniques including Le Cam's method, see e.g. \cite{Yu1997LeCam} and its application to learning trees in~\cite{chen2021exponentialseparationslearningquantum}. Our goal will be to argue that for the total variation bound in Lemma \ref{lemma:le_cam} to become appreciably large, the tree depth $T$ must grow exponentially in $n$. To do so, it suffices to bound the one-sided likelihood ratio.  Let us first define this ratio, and then subsequently explain why bounding it is useful for us.
\begin{definition}[One-sided likelihood ratio]
Assume a learning tree $\mathcal{T}$ and a many-vs.-one distinguishing problem with distribution $\mathcal D$. For every $\ell\in \textnormal{leaf}(\mathcal{T})$, the one-sided likelihood ratio is
\begin{equation}
    L(\ell) = \frac{\mathbb{E}_{\rho_x\sim \mathcal{D}}[p_{\rho_x}(\ell)]}{p_{\mathds{1}/2^n}(\ell)}\,.
\end{equation}
Given a concrete instance of $\rho_x$, we also define fixed-state edge and leaf likelihood ratios:
\begin{equation}
    L_{\rho_x}(s|u) = \frac{p_{\rho_x}(s|u)}{p_{\mathds{1}/2^n}(s|u)}\,, \quad
        L_{\rho_x}(\ell) = \frac{p_{\rho_x}(\ell)}{p_{\mathds{1}/2^n}(\ell)}
\end{equation}
\end{definition}

With these definitions in place, we regularly draw upon the following toolbox for quantum learning lower bounds.
\begin{lemma}[Toolbox for learning tree lower bounds]\label{lemma:toolbox}
Suppose $\mathcal{T}$ is a learning tree for a many-vs.-one distinguishing problem with distribution $\mathcal D$.
    \begin{enumerate}
        \item \textnormal{(Lemma 5.4 in \cite{chen2021exponentialseparationslearningquantum})} If $L(\ell) \geq 1-\delta$ for all $\ell\in \textnormal{leaf}(\mathcal{T})$, then we have the bound $d_{\text{\rm TV}}(\mathbb{E}_\mathcal{D}[p_{\rho_x}(\ell)],p_{\mathds{1}/2^n}(\ell)) \leq \delta$.
        \item \textnormal{(Lemma 7 in \cite{weiyuan_paulis})} If the condition
        \begin{equation}
            \mathbb{E}_{\rho\sim \mathcal D}\,\mathbb{E}_{s\sim p^{\mathds{1}/2^n}(s|u)} \left[\left(L(s|u) - 1\right)^2\right]\leq \delta
        \end{equation}
        is satisfied for all $u$ in $\mathcal{T}$, then:
        \begin{enumerate}
          \item For any $\rho_x$, there exists a constant $C>0$ such that $\textnormal{Pr}_{\ell\sim p^{\mathds{1}/2^n}(\ell)}[L_{\rho_x}(\ell)\!\leq\! 0.9] \leq 0.1 + C\delta T$.
          \item  $\mathcal{T}$ must have depth $T >\Omega(1/\delta)$ for the algorithm to succeed with probability $>2/3$. 
        \end{enumerate}
        \item \textnormal{Lemma 6 in \cite{weiyuan_paulis}} For any $\delta \in (0, 1)$, $$d_{\text{\rm TV}}(\mathbb{E}_\mathcal{D}[p_{\rho_x}(\ell)],p_{\mathds{1}/2^n}(\ell)) \leq \textnormal{Pr}_{\ell\sim p^{\mathds{1}/2^n}(\ell)}[L(\ell) \leq 1-\delta] + \delta\,.$$
        \end{enumerate}
\end{lemma}
\noindent The first lemma in our toolbox states that a lower bound on the one-sided likelihood ratio for all possible outcomes of our algorithm is sufficient to bound the total variation distance. While powerful in simple cases, this lemma requires us to track the likelihood ratio over all paths in the tree. The second and third lemmas relax this requirement in different ways.

The second lemma tells us that if the likelihood ratio induced by \textit{intermediate nodes} of the tree sharply concentrates (indicating that a single experiment carries little information about the unknown state), the ratio over leaves is likely close to $1$, and the tree depth must then be large to distinguish the hypotheses. Finally, our third lemma allows us to bound the total variation distance by a \textit{probabilistic} leaf likelihood ratio. That is, if we consider a fixed but arbitrary root-to-leaf path and argue that, with high probability the likelihood ratio on \textit{that path} is close to $1$, we can utilize Lemma \ref{lemma:le_cam}. This approach is useful when the bound on intermediate nodes fails, which may happen when a small, nonzero number of intermediate nodes has large likelihood ratio fluctuations despite only comprising a vanishing fraction of the eventual outcomes. For a detailed overview of these ``edge-based" and ``path-based" approaches, see \cite{chen2021exponentialseparationslearningquantum}.

\subsection{Symplectic stabilizer formalism}\label{sec:symplectic_paulis}
It is well known that Clifford unitaries on $n$ qubits are precisely those operations that preserve the standard symplectic form on $\mathbb{F}_2^{2n}$. In this work, we briefly recall the basic notation from this formalism.

Any $n$-qubit Pauli operator $P$ can be represented by a binary vector $p = (p^x \mid p^z) \in \mathbb{F}_2^{2n}$, where $p^x, p^z \in \mathbb{F}_2^n$, such that
\begin{equation}
P = i^{p^x \cdot p^z} X^{p^x} Z^{p^z}\,.
\end{equation}
Products of Pauli operators are governed by the symplectic inner product
\begin{equation}
\langle p, q \rangle = p^x \cdot q^z - p^z \cdot q^x\,,
\end{equation}
with all arithmetic performed over $\mathbb{F}_2$. We also define the phase term $\langle p \rangle = p^x \cdot p^z$.

In this notation, the POVM elements of the $n$-qubit Bell measurement, written in the Pauli basis, are
\begin{equation}
\label{eq:Bell_POVM}
\Pi_s = \frac{1}{4^n}\sum_{P \in \mathcal{P}_n} (-1)^{\langle s, p \rangle + \langle p \rangle} \, P \otimes P\,,
\end{equation}
where each $\Pi_s$ is labeled by a symplectic bitstring $s \in \mathbb{F}_2^{2n}$, and for each Pauli $P$ in the sum we write $p \in \mathbb{F}_2^{2n}$ for its corresponding symplectic representation. Equivalently, one may regard the sum as running over all $p \in \mathbb{F}_2^{2n}$.

The following measurement subroutine will be very useful.
\begin{definition}[Bell measurement subroutine]
\label{def:BM_subroutine}
The Bell POVM on $2n$ qubits, $\{\Pi_s\}_{s \in \mathbb{F}_2^{2n}}$, has $4^n$ elements, each characterized by a bitstring $s \in \mathbb{F}_2^{2n}$ and defined by~\eqref{eq:Bell_POVM}. The subroutine \textnormal{\textsc{BellMeasure}}$(k,\rho)$ applies the Bell POVM on $2k$ qubits to a state $\rho$ on $2k$ qubits, and returns a bitstring $s \in \mathbb{F}_2^{2k}$ sampled from the distribution
\begin{equation}
\mathrm{Pr}[s] = \mathrm{tr}(\Pi_s \rho)\,.
\end{equation}
\end{definition}

\section{Complexity-Theoretic Separations}
\label{app: complexity_sep}
\subsection{Separating \textsf{NISQ} and \textsf{NBQP}}
The fact that $\textsf{NISQ}\subseteq \textsf{NBQP}$ is immediate, since $\textsf{NBQP}_\lambda$ can already run any $\lambda$-noisy quantum circuit. Given a classical step in a $\textsf{NISQ}$ algorithm, $\textsf{NBQP}$ can perfectly simulate the classical computation on a $\lambda$-noisy quantum register for any constant $\lambda<\lambda_\text{th}$, where $\lambda_\text{th}$ is the threshold for a constant-rate fault-tolerant quantum error correction (FT-QEC) scheme as in $\cite{aharonov1999faulttolerantquantumcomputationconstant}$, with polynomial overhead. Now we will show that this inclusion is actually strict relative to an oracle, proving Theorem \ref{Thm:NISQ_vs_NBQP}.

\subsubsection{Encoded Shuffling Simon's Problem}

The oracle we use to separate $\textsf{NISQ}$ and $\textsf{NBQP}$ is a modification of the Shuffling Simon's Problem from \cite{Chia_2023}. First we give some preliminary definitions.

\begin{definition}[Simon's function]
A two-to-one function $f:\mathbb{Z}_2^n\rightarrow \mathbb{Z}_2^m$ is a Simon's function if there is some hidden secret $s\in \mathbb{Z}_2^n$ such that for every $x\in \mathbb{Z}_2^n$, $f(x) = f(x\oplus s)$.
\end{definition}
The standard Simon's search problem is: given oracle access to a Simon's function $f$, recover the hidden string $s$. We will consider a more difficult task in which the domain of the function $f$ is unknown.

\begin{definition}[$d$-Level Shuffling of $f$]
\label{def:dLevelShuffling}
Consider a function $f:\mathbb{Z}_2^n\rightarrow \mathbb{Z}_2^m$. A d-level shuffling of $f$ is a sequence of functions $\mathcal{F} = (f_0, f_1, ..., f_d)$ with the following property. For all $i = 0, 1, ..., d-1$, $f_i$ is a permutation acting on $\mathbb{Z}_2^{n(d+2)}$. Let $S_{d}$ denote the image of $f_{d-1}\circ \cdots \circ f_0$ acting upon the first $2^n$ lexicographically-ordered bitstrings in $\mathbb{Z}_2^{n(d+2)}$. Then, $f_d$ is defined so that for all $x\in S_{d}$, we have
$$f_d(x) = f((f_{d-1}\circ \cdots \circ f_0)^{-1}(x))$$
where the lexicographically-ordered bitstrings in $ \mathbb{Z}_2^{n(d+2)}$ are indentified with $\mathbb{Z}_2^{n}$ in the natural way. For any $x\notin S_{d}$, $f_d(x) = \ \perp$, outputting no logical information. Let $\textnormal{\textbf{SHUF}}(f, d)$ denote all d-level shufflings of $f$, and let $D(f, d)$ denote the distribution over $\textnormal{\textbf{SHUF}}(f, d)$ obtained by sampling all $f_0,...,f_{d-1}$ uniformly at random from the set of permutations. 
\end{definition}
A $d$-level shuffling of $f$ embeds the domain of $f$ in a space $2^{n(d+1)}$ times larger, then permutes the embedding to hide the domain. In the problem we construct, we will need quantum oracle access to a $d$-level shuffling defined according to the following.

\begin{definition}[$d$-Level Quantum shuffling oracle] \label{def:shuffling_oracle}
Fix a function $f:\mathbb{Z}_2^n\rightarrow \mathbb{Z}_2^m$. The quantum shuffling oracle $\mathcal{O}_{f, d}$ is the quantum channel that acts on an input state $\ket{\psi}$ of the form 
\begin{equation}
    \ket{\psi} = \left(\bigotimes_{i=0}^d \ket{i ,x_i}\right) \otimes \ket{y}
\end{equation}
with all $x_i\in \mathbb{Z}_2^{n(d+2)}$ as
\begin{equation}
    \mathcal{O}_{f, d}(\ketbra{\psi}{\psi}) = \mathbb{E}_{\mathcal{F}\sim D(f, d)}[\mathcal{F}\ketbra{\psi}{\psi}\mathcal{F}^\dagger]
\end{equation}
where, in a slight abuse of notation, we let $\mathcal{F}$ be a unitary corresponding to a shuffling which acts as
\begin{equation}
     \mathcal{F}\ket{\psi} = \bigotimes_{i=0}^d \ket{i ,x_i}\ket{y\oplus f_i(x_i)}
\end{equation}
Note that $\mathcal{F}$ is classically sampled from $D(f, d)$, and $O_{f, d}$ applies the same shuffling $\mathcal{F}$ upon each query. Here and below, we assign $f_i(x_i)$ a fixed bit string outside of all legitimate outputs when $f_i(x_i)=\perp$ to make $\ket{y\oplus f_i(x_i)}$ a valid quantum state.
\end{definition}

We will now use shuffling to construct a version of Simon's problem that is challenging for a quantum device with bounded depth, as done in \cite{Chia_2023}, but for which an $\textsf{NBQP}$ algorithm can learn the hidden secret even when noise is applied to the quantum shuffling oracle.

\begin{definition}[Encoded $d$-level shuffled Simon's oracle]
Let $f:\mathbb{Z}_2^n\rightarrow \mathbb{Z}_2^n$ be a Simon function, and let $d$ denote the shuffling level. Moreover, let $\mathsf{QEC}$ be a fixed, uniform fault-tolerant quantum error correction (FT-QEC) scheme with encoding unitary channel $\mathcal{U}_{\mathrm{enc}} = U_{\mathrm{enc}}^\dagger (\cdot){U}_{\mathrm{enc}}$ acting on $m(n) = \textnormal{poly}(n)$ physical qubits, of depth at most $\ell = \ell(n)$, and threshold $\lambda_{\mathrm{th}}$. These parameters are chosen such that \textnormal{\textsf{QEC}} can protect a logical quantum circuit on $n$ qubits of depth $3d$ with probability at least $0.99$. Let $\mathcal{C} := \text{\rm span}\{\mathcal{U}_{\mathrm{enc}}\ket{x}\ket{0}^{\otimes m-n}: x\in \{0, 1\}^n\}$ be the encoded logical basis, and let $\mathcal{C}^\perp$ be the orthogonal complement to the logical codespace.
    
    For each $i\in\{0,\ldots,d\}$ define the \emph{encoded oracle channel}
    \begin{align}
    \mathcal{O}^{\mathrm{enc}}_{f, d} := \mathcal{U}_{\mathrm{enc}} \circ (\mathcal{O}_{f, d}\otimes\ketbra{0}{0}^{m-n})\ \oplus \mathds{1}_{\mathcal{C}^\perp}\,,
    \end{align}
    where $\mathcal{O}_{f, d}$ is the standard quantum shuffling oracle for $f$.
\end{definition}

The encoded oracle acts on the logical codespace of an FT-QEC scheme the same way the standard oracle acts on computational basis states, and acts trivially on states that lie outside the codespace. With this, we can present the oracle problem we use to separate $\textsf{NISQ}$ from $\textsf{NBQP}$.

\begin{definition}[Encoded $d$-Shuffling Simon's Problem]\label{def:encoded-codehash}
Let $f:\mathbb{Z}_2^n\rightarrow \mathbb{Z}_2^n$ be an unknown Simon function with random secret $s$ and $d = \textnormal{poly}(n)$. Given quantum access to the encoded d-level shuffled Simon's oracle $\mathcal{O}^{\mathrm{enc}}_{f, d}$, the goal of the Encoded $d$-Shuffling Simon's Problem $\textnormal{(\textsf{Enc-}d\textsf{-SSP})}$ is to learn $s$.
\end{definition}

The following lemma is immediate.
\begin{lemma}\label{lemma:NBQP_upper}
$\textnormal{\textsf{Enc-}}d\textnormal{\textsf{-SSP}} \in \textnormal{\textsf{NBQP}}^{\mathcal{O}_{f, d}}$.
\end{lemma}
\begin{proof}
    Theorem 4.11 of \cite{Chia_2023} gives a quantum algorithm that solves the \textnormal{\textsf{Enc-}}$d$\textnormal{\textsf{-SSP}} problem with a noiseless quantum circuit of depth $2d+1$. For any constant $\lambda<\lambda_{\textnormal{th}}$, a $\lambda$-noisy quantum circuit can use quantum fault-tolerance \cite{aharonov1999faulttolerantquantumcomputationconstant} to simulate this noiseless algorithm with logarithmic overhead in the depth of the quantum circuit, recalling that the encoding scheme \textsf{QEC} can, with high probability, execute a circuit of depth $3d> 2d+1$ without logical errors.
\end{proof}

In the remainder of this section, we show that no $\textsf{NISQ}$ algorithm making polynomially many queries can solve \textsf{Enc-}d\textsf{-SSP} with probability $\geq 2/3$. The proof strategy will be to first show that any $\textsf{BPP}^{\textsf{QNC}_d}$ has a success probability exponentially small in $n$, following the argument in \cite{Chia_2023}. Then we argue that for $d$ and a number of queries both polynomial in $n$, the outcome distribution of any \textsf{NISQ} algorithm can be simulated by a $\textsf{BPP}^{\textsf{QNC}_d}$ algorithm to vanishingly small total variation distance.  

\subsubsection{Bounding quantum success probability with classical combinatorics}
We first prove necessary lemmas presented in \cite{Chia_2023} for $d\textsf{-SSP}$, now for the encoded version of the problem. The definitions we present are largely adapted from \cite{arora2023quantum, Chia_2023}. 

Note that all information about $f$ is contained within $S_{d}$, and an algorithm trying to learn the secret $s$ must perform queries in this unknown subspace. Looking at the preimage of $S_{d}$ under $f_{d-1}$, then the preimage of this subpace under $f_{d-2}$, and iteratively continuing to look at preimages under each random permutation until $f_0$, would reveal to us the subset of the initial query subspace which gets mapped to $S_{d}$. Hence, all of these preimages form ``hidden domains," such that an algorithm that does not query within the hidden domains can never learn anything about $f$ by construction. This intuition is formalized by the following definition.

\begin{definition}[Hidden domains and wrappers] Fix a function $f: \mathbb{Z}_2^n\rightarrow\mathbb{Z}_2^n$ and let $\mathcal{F} = (f_0,...,f_d)$ be some $d$-level shuffling of $f$. Let $S_0 = \{0, 1,...,2^n-1\}$ denote the first $2^n$ lexicographically ordered bitstrings in $\mathbb{Z}^{n(d+2)}$. For $i = 1,2,...,d$ the hidden domain $S_i$ is given by $f_i\circ f_{i-1}\circ \cdots \circ f_0(S_0)$.

The level-$k$ hidden wrappers $S_0^{(k)},...,S_d^{(k)}$ for $k = 0, 1,... d$ are defined as follows.
\begin{enumerate}
    \item $S_i^{(0)} = \mathbb{Z}^{n(d+2)}$, the entire enlarged space, for all $i = 0, 1, ..., d$.
    \item For $k = 1,...,d$, for $j = k,...,d$, the wrapper set $S_j^{(k)}$ is chosen to satisfy $S_j^{(k)} \subset S_j^{(k-1)}$ such that $|S_j^{(k)}|/ |S_j^{(k-1)}| \leq 2^{-n}$, $S_j\subseteq S_j^{(k)}$, and $f_j(S_{j-1}^{(k)}) = S_j^{(k)}$.
\end{enumerate}   
For our purpose, define the collection of level-k hidden wrappers $\bar{S}^{(k)} = (S_k^{(k)}, S_{k+1}^{(k)},...,S_d^{k})$
\end{definition}
Hidden domains, intuitively, track the small subset of queries which reveal actual information about $f$. Hidden wrappers are nested sets of exponentially growing size, all containing the corresponding hidden domain. These wrappers have the additional property that under the permutation $f_i$, the $i-1^{st}$ hidden wrapper is mapped onto the $i^{th}$ for every level of the nesting. The domains are the crucial hidden information, while the wrappers are proof tools we utilize in conjunction with \textit{shadow shufflings}, defined next.
\begin{definition}[Shadow shuffling and functions]
Fix a $d$-level shuffling $\mathcal{F} = (f_0,...,f_d)$ of $f$. Then fix collections of level-k hidden wrappers $\bar{S}^{(k)}$ for all $k = 0, 1,...,d$. The shadow shuffling $\mathcal{F}_{\text{\rm sh}}^{(k)}$ is the sequence
\begin{equation}
\mathcal{F}_{\text{\rm sh}}^{(k)} = (f_0,...,f_k,g_{k+1}^{(k)},..., g_d^{(k)})
\end{equation}
where each $g_i$ is a shadow function, defined as 
\begin{align}
g_i^{(k)}(x) = \begin{cases}
 f_i(x) & \text{\rm if } x \notin S_i^{(k)} \\
\perp & \text{\rm else }
\end{cases}.
\end{align}
\end{definition}
\noindent The shadow shuffling $\mathcal{F}_{\text{\rm sh}}^{(k)}$ is a map that, when any of its first $k$ shadow functions are queried within a hidden wrapper, provides no information about the secret $s$; hence, if an algorithm only had access to $\mathcal{F}_{\text{\rm sh}}^{(k)}$, it would need to already know something about the hidden domain $S_k$ in order to proceed. Using this technology, we can argue that any $\textsf{BPP}^{\textsf{QNC}_d}$ algorithm can have its queries to the shuffler replaced by shadow shufflers, without changing its output distribution substantially; but such an algorithm, by construction, knows nothing about the domain of $f_d$.

To proceed, let us formally state the action of the quantum unitary corresponding to a fixed shuffling Simon's function $\mathcal{F} = (f_0,...,f_d)$ as in Definition \ref{def:shuffling_oracle}. When referring to a particular shuffling $\mathcal{F}$ rather than the classical mixture oracle $O_{f,d}$, we utilize the encoded unitary shuffling oracle $\mathcal{F}^{\textnormal{Enc}}$ defined analogously as
\begin{equation}\label{eq:f_oracle}
    \mathcal{F}^{\textnormal{Enc}} = \mathcal{U}_{\textnormal{Enc}} \circ (\mathcal{F} \otimes \ketbra{0}{0}^{m-n}) \oplus \mathds{1}_{\mathcal{C}^\perp} 
\end{equation}

Then we are armed with all the necessary definitions to prove a lemma that will allow us to bound the success probability of a quantum algorithm using purely classical combinatorics over the underlying classical vector space on $n(d+2)$-bitstrings. Using this, we will no longer need to concern ourselves with the enlarged code Hilbert space.

\begin{lemma}[Oneway-to-Hiding (O2H) Lemma for Encoded Oracle] \label{lemma:o2h}
    Pick some $k, d \in \mathbb{N}$ with $k<d$. 
    \begin{enumerate}
    \item Let $\mathcal{F} = (f_0,...,f_d)$ be a d-level shuffling Simon's function, and fix collections of $k$-level hidden wrappers $\bar{S}^{(k)}$ for $k= 1,...,d$. Then let $\mathcal{F}_{\text{\rm sh}}^{(k)}$ be the shadow of $\mathcal{F}$.
    \item Suppose $\textnormal{Pr}[x\in S_i^{(k)}|x \in S_i^{(k-1)}] \leq p$ for all $i, k$. 
    \item Define any input state with density matrix $\rho$, and any unitary $U$ which makes $q$ queries to $\mathcal{F}$, such that $\rho$ and $U$ are uncorrelated to $\bar{S}^{(k)}$ and $\mathcal{F}$ restricted to $\bar{S}^{(k)}$.
    \end{enumerate}
    We let $\mathcal{F}^{\textnormal{Enc}}$ denote the corresponding encoded quantum oracle as defined in Definition \ref{def:shuffling_oracle} and Eq.~\eqref{eq:f_oracle}, with encoded quantum shadow $\mathcal{F}_{\text{\rm sh}}^{(k), \textnormal{Enc}}$. Moreover we we take $\mathcal{U}(\rho) := U \rho U^\dagger$.  Let $\Pi_s$ denote a projector onto encoded computational basis string $s$. Then 
    \begin{equation}\label{eq:o2h}
        \left|\tr[\Pi_s \mathcal{F}^{\textnormal{Enc}} \mathcal{U} (\rho)] - \tr[\Pi_s \mathcal{F}_{\text{\rm sh}}^{(k), \textnormal{Enc}} \mathcal{U} (\rho)]\right| \leq \sqrt{2qp}\,.
    \end{equation}
    While we only need $\Pi_s$ to be a computational basis measurement, this inequality holds for any generic POVM elements. We have taken the oracle and $U$ to act as channels in Eq.~\eqref{eq:o2h}.
\end{lemma}
\noindent We defer the proof to Appendix \ref{app:deferred}. This lemma is the crucial step that allows us to relate quantum success probability to classical combinatorics. The O2H bound tells us that the output distribution of a quantum algorithm making a bounded number of queries to the true oracle differs from one querying a shadow oracle which knows nothing about the crucial function $f$ by an amount controlled by the classical probability of picking a bitstring that lies in a smaller hidden wrapper nested within a larger one. Notably, this is completely independent of the FT-QEC scheme. Since every round of nesting makes the space bigger by a factor of $2^n$, O2H intuitively allows us to bound these difference between these output distributions by an exponentially small number. This intuition leads to the following bound on the success probability of an algorithm with limited quantum depth.

\begin{theorem}[Proven in Sections 5, 8 of \cite{Chia_2023}] \label{thm:nisq_sep_master}
Let $f$ be a Simon's function on $n$ bits and let $\mathcal{A}^{\mathcal{O}_\mathcal{F}}$ be any $\textnormal{\textsf{BPP}}^{\textnormal{\textsf{QNC}}_d}$ algorithm with quantum access to some unitary oracle $\mathcal{O}_\mathcal{F}$ corresponding to a randomly chosen shuffling Simon's function $\mathcal{F} = (f_0,...,f_d)$. If the following conditions are satisfied:
\begin{enumerate}
    \item $\mathcal{F}$ is sampled from $D(f, d)$ (see Definition~\ref{def:dLevelShuffling}).
    \item $\mathcal{O}_\mathcal{F}$ is defined in such a way that when the three conditions of Lemma \ref{lemma:o2h} are satisfied, Eq.~\eqref{eq:o2h} holds for any arbitrary POVM element $\Pi_s$. Note that the three conditions are dependent only on the classical data of $\mathcal F$ and hidden wrappers $\bar{S}^{(k)}$ and not on the quantum implementation $\mathcal{O}_\mathcal{F}$.  
\end{enumerate}
Then the probability that $\mathcal{A}^{\mathcal{O}_\mathcal{F}}$ outputs a bitstring $s$ equal to the true hidden bitstring is bounded by $d\sqrt{\textnormal{poly}(n)/2^n}$.
\end{theorem}
The formulation of $\textnormal{\textsf{Enc-}}d\textnormal{\textsf{-SSP}}$ assumes $\mathcal{F}\sim D(f, d)$. Moreover, in Appendix \ref{app:deferred} we prove Lemma \ref{lemma:o2h}, implying that the second condition of Theorem \ref{thm:nisq_sep_master} holds for $\mathcal{O}_\mathcal{F} = \mathcal{F}^\textnormal{Enc}$ defined in Eq.~\eqref{eq:f_oracle}. Hence, we immediately obtain the following lemma.
\begin{lemma}\label{lemma:QNC_small_success}
    The success probability of any $\textnormal{\textsf{BPP}}^{\textnormal{{\textsf{QNC}}}_d}$ algorithm with $d = \textnormal{poly}(n)$ for $\textnormal{\textsf{Enc-}}d\textnormal{\textsf{-SSP}}$ is bounded by $O(\textnormal{poly}(n)/2^{n/2})$. 
\end{lemma}

\subsubsection{Reducing to \textsf{NISQ}}
We can now prove the following lemma.

\begin{lemma}\label{lemma:NISQ_lower}
    $\textnormal{\textsf{Enc-}}d\textnormal{\textsf{-SSP}} \notin \textnormal{\textsf{NISQ}}^{\mathcal{O}_{f, d}}$.
\end{lemma}
\begin{proof}
    Let $\mathcal{T}$ be the learning tree representation of a \textsf{NISQ}$_\lambda$ algorithm with access to $O_{f, d}$, recalling that $d = \textnormal{poly}(n)$. For every non-leaf node $u$ in the tree, the algorithm either runs a classical algorithm, which cannot make oracle queries, or runs a $\lambda$-noisy quantum circuit on $m =$ poly$(n)$ qubits. Now we will transform $\mathcal{T}$ into a learning tree $\mathcal{T}'$ representing a $\textsf{BPP}^{\textsf{QNC}_d}$ algorithm.
    
    Start with $\mathcal{T}' = \mathcal{T}$. Let $\bar{D}$ be a parameter representing a depth threshold in the noisy quantum circuit. If at node $u$ in $\mathcal{T}$, the noisy quantum circuit performs more than $\bar{D}$ queries to $O_{f, d}$, replace $u$ in $\mathcal{T}'$ with a noiseless quantum circuit that measures the $n$-qubit maximally mixed state in the computational basis. If the circuit makes less than $\bar{D}$ queries, replace $u$ with the noiseless quantum circuit that simulates the $\lambda$-noisy circuit by simulating depolarizing noise at each layer. With $\mathcal{T}'$ defined, we call upon the following lemma from \cite{nisq}:
    \begin{lemma}[In \cite{nisq}, Lemmas D.15 and D.16] \label{lemma:nisq_kl}
        Let $\rho$ be any $m$-qubit state output by a $\lambda$-noisy depth-$D$ quantum circuit and choose any POVM. Let $p, q$ be the distributions induced by measuring $\rho$ or the maximally mixed state with the chosen POVM, respectively. Then $\text{KL}(p\|q)\leq (1-\lambda)^Dm$.
    \end{lemma}
    By Pinsker's inequality and Lemma \ref{lemma:nisq_kl}, the total variation distance between the induced conditional distribution on the children of a node $u$ when making the replacement we described for circuits with more than $\bar{D}$ queries is at most $(1-\lambda)^{\bar{D}/2}O(\text{poly}(n))$. Now, let $N$ be the depth of $\mathcal{T}$ and $\mathcal{T}'$. Then the total variation distance between leaf output distributions of the two trees is bounded by $N(1-\lambda)^{\bar{D}/2}O(\text{poly}(n))$. We take $\bar{D} = \text{poly}(n) - (\text{log}(N)/\log(1-\lambda))$, making this bound $O(\text{poly}(n)(1-\lambda)^{\text{poly}(n)}$). Applying Lemma \ref{lemma:QNC_small_success}, along with a union bound and Lemma \ref{lemma:le_cam}, the success probability of the $\textsf{NISQ}$ algorithm is bounded by $O((2^{-n/2} + (1-\lambda)^{\text{poly}(n)})\text{poly}(n))$. Hence, the \textsf{NISQ} algorithm cannot solve $\textnormal{\textsf{Enc-}d\textsf{-SSP}}$ with probability $\geq 2/3$ whenever $N\leq \Theta(\min(\{2^{n/2}, (1-\lambda)^{-\text{poly}(n)}\})/\textnormal{poly}(n)) = \text{superpoly(n)}$.
\end{proof}
\noindent Lemma \ref{lemma:NBQP_upper} and Lemma \ref{lemma:NISQ_lower} combine to give us Theorem \ref{Thm:NISQ_vs_NBQP}.

\subsection{Separating \textsf{NBQP} and \textsf{BQP}} \label{app:remark}
A noiseless quantum circuit can simulate any $\lambda$-noisy quantum circuit with overhead polynomial in the depth of the circuit by applying depolarizing noise after each layer. Hence, $\nbqp^O\subseteq\bqp^O$ for any quantum oracle $O$. Moreover, the existence of constant-rate FT-QEC implies that without an oracle, $\nbqp = \bqp$. In this Section, we construct an oracle to demonstrate the strict separation in Theorem \ref{thm:NBQP_vs_BQP}.

\subsubsection{Remark on criteria for noise-robust quantum learning algorithms}

Before proceeding with the argument, we make some conceptual remarks. The class $\nbqp^O$ with an oracle is intended to model a fault-tolerant quantum computer that can perform protected quantum computation while interacting with a quantum system in Nature, viewed abstractly as a quantum oracle $O$. Each ``oracle query'' corresponds to implementing some controlled interaction between the quantum computer and this physical system and then reading out its response. Crucially, while the quantum computer's internal registers can be encoded in a known code state and processed using logical gates within a FT-QEC scheme, the oracle system itself is not error-corrected, and its coupling to the computer is subject to physical noise. The relevant quantum information about the system is therefore stored in unprotected physical degrees of freedom, not in the computer's logical qubits. Given this, where can we expect to find meaningful separations between $\nbqp^O$ and $\bqp^O$ machines in quantum learning tasks, highlighting gaps between idealized and physically motivated learning models?

It is plausible that even when the oracle system is uncharacterized, if its behavior is constrained to a small, structured ansatz class, certain simple error patterns remain recoverable. An extreme example would be if the system in Nature was somehow the output of an error-corrected quantum simulation, which is guaranteed to lie within the codespace of a known quantum error-correcting code.  In such settings, a fault-tolerant learner could treat the oracle output as another logical state. Thus, one strategy for a separation between $\nbqp^O$ and $\bqp^O$ is if the oracle system (and its interaction with the computer) is sufficiently unstructured so that it is difficult to design an effective error-correction strategy for it.

A more interesting challenge for a quantum learning algorithm arises when local physical errors on the oracle system, or during its coupling to the computer, can affect the encoded quantum information in a highly nonlocal way. If a few sparse errors on the query register only corrupt a bounded amount of the logical information we are trying to extract, then repetition and majority-vote type procedures can still be used to distill a high-fidelity estimate of the oracle's behavior \cite{buhrman2022quantummajorityvote, iwama_2005, chakraborty_2005, H_yer_2003, suzuki2006robustquantumalgorithmsepsbiased}.
On the other hand, suppose an oracle loads a single computational-basis bitstring encoding a value between $0$ and $2^n - 1$. In this case, a single local error may be enough to completely corrupt the stored information, since an error on the $k$-th qubit can change the measured value by $2^{k-1}$. While this toy example, and the corresponding oracle we construct for our exponential separation, are deliberately designed to impede an $\nbqp$ learner, the underlying intuition is that noise-robust quantum learning theory often requires some notion of logical locality in how information is encoded.

Consequently, it is likely that physical systems that permit learning-enhanced experiments will (a) be restricted to a small ansatz class, or (b) have properties that are largely independent of sparse, local perturbations.

\subsubsection{Hybrid argument}
A key ingredient in the separation will be the following lemma.
\begin{lemma}[Learning from similar oracles, Lemma B.4 in \cite{nisq}]
Consider two quantum channels $\mathcal{O}_1$ and $\mathcal{O}_2$, and suppose that for all pure states $\sigma$, $\|(\mathcal{O}_1 - \mathcal{O}_2)[\sigma]\|_{\textnormal{tr}} \leq \epsilon$. Let $\mathcal{C}_D(\mathcal{O})$ be a quantum circuit of the form
\begin{equation}
    \mathcal{C}_D(\mathcal{O}) = \mathcal{U}_D\circ \mathcal{O}\circ\mathcal{U}_{D-1}\circ \mathcal{O} \circ \cdots \circ \mathcal{U}_1[\ketbra{0}{0}^{\otimes n}]\,.
\end{equation}
Then let $p_1$ and $p_2$ be the classical probability distributions over $n$-bitstring outputs of $\mathcal{C}_D(\mathcal{O}_1)$ and $\mathcal{C}_D(\mathcal{O}_2)$. It holds that $d_{\text{\rm TV}}(p_1, p_2)\leq \epsilon D$.    
\end{lemma}
We find that the ``lifted Simon's oracle" used in \cite{nisq} to separate $\nisq$ from $\bqp$ suffices to separate $\nbqp$ as well. 
\begin{definition}[Lifted Simon's oracle]
Given a function $f:\{0, 1\}^n\rightarrow\{0, 1\}^n$, define the lift $\tilde{f}:\{0, 1\}^{2n}\rightarrow\{0, 1\}^n$ as
\begin{equation}
    \tilde{f}(x) =\begin{cases}
        f(x_{1:n}) &\text{\rm if} \   x_{n+1:2n} = 0^n \\
        0 &\text{\rm else}
    \end{cases}
\end{equation}
Given a Simon's function $f$, we denote the lifted Simon's oracle by $O_{\tilde{f}}$, where $O_{\tilde{f}}$ acts on classical bitstrings as $O_{\tilde{f}}(x) = \tilde{f}(x)$, and abusing notation, the quantum oracle $O_{\tilde{f}}$ acts as $O_{\bar{f}}\ \ket{x}\ket{y} = |x\rangle|y\oplus \tilde{f}(x)\rangle$.    
\end{definition}
\noindent Then the decision lifted Simon's problem is, given oracle access to $O_{\tilde{f}}$, to determine whether a $f$ is 1-to-1 or 2-to-1.

The lifted Simon's problem is in $\bqp^{O_{\tilde{f}}}$, because one can simply set the last $n$ qubits of each query to $\ket{0}^{\otimes n}$ and run the standard Simon's algorithm \cite{Simon}. To show that the lifted Simon's problem is not in $\nbqp^{O_{\tilde{f}}}$, we use an argument similar to the one used by $\cite{nisq}$ to show that it is not in $\nisq^{O_{\tilde{f}}}$. We call on the following lemmas from their work:
\begin{lemma}[Distinguishability of noisy output distributions, Lemma C.22 in \cite{nisq}]\label{lemma:tree_id_vs_o}
Let $\mathcal{C}_n^O$ be a $\lambda$-noisy quantum circuit on $n$ qubits making $N$ oracle queries to some oracle $O$. If $p_{\tilde{f}}$ is the output distribution over bitstrings when $O = \mathcal{O}_{\tilde{f}}$ and $p_{I}$ is the output distribution when $O = \mathds{1}$ acting as the identity, $d_{\text{\rm TV}}(p_{\tilde{f}}, p_I) \leq N\exp(-\Omega(\lambda n))$.  
\end{lemma}
\begin{lemma}[Distinguishability propagation in a learning tree, $\nbqp$ variant of Lemma B.2 in \cite{nisq}] \label{lemma:node_perturbation}
Let $\mathcal{T}$ be a learning tree for an $\nbqp_\lambda$ algorithm making $N$ queries. Suppose at each node $u$ of the tree where an oracle query is made, the $\lambda$-noisy quantum circuit $\mathcal{C}_u$ is replaced by a circuit $\mathcal{C}_u'$ such that the induced distribution over children of $u$ is at most $\epsilon$-far in total variation distance from the original distribution. Let the tree obtained from this procedure be $\mathcal{T}'$. Then the distributions over leaves of $\mathcal{T}, \mathcal{T'}$ are at most $\epsilon N$-far in total variation distance. 
\end{lemma}
\begin{proof}
   The number of nodes along any root-to-leaf path in $\mathcal{T}$ making an oracle query is at most $N$. Letting $\mathcal{T}^{(0)} = \mathcal{T}$, we construct intermediate trees $\mathcal{T}^{(i)}$ for $i=1,...,N$ as follows. Start with any root-to-leaf path of $\mathcal{T}^{(i-1)}$, find the $i$-th node that makes an oracle query ( runs some circuit $\mathcal{C}$), and replace it with an $\epsilon$-total variation close circuit $\mathcal{C}'$. Do this for all root-to-leaf paths. The resulting tree is $\mathcal{T}^{(i)}$, and note that $\mathcal{T}' = \mathcal{T}^{(N)}$. By the triangle inequality, it suffices to show that the total variation between leaf distributions of $\mathcal{T}^{(i-1)}$ and $\mathcal{T}^{(i)}$, denoted $p^{(i-1)}, p^{(i)}$, is at most $\epsilon$.

The leaf distribution $p^{(i-1)}$ can be decomposed into a convex combination of distributions conditioned on reaching each node in the $i$-th level. That is, suppose $v$ is the $i$-th oracle-querying node on a particular root-to-leaf path, and $p_v$ is the leaf distribution of $\mathcal{T}^{(i-1)}$ conditioned on reaching $v$. This conditional probability is obtained by running the $\nbqp$ algorithm corresponding to $\mathcal{T}'$ for $i-1$ oracle-querying steps. In $\mathcal{T}^{(i)}$, this node $v$ has the same conditional probability of being reached, but its output distribution is changed only by replacing the circuit at $v$ with one from $\mathcal{T'}$. The induced distribution on leaves from $v$, by assumption, is at most $\epsilon$-far from the distribution in $\mathcal{T}^{(i-1)}$. Since this applies to all $v$ in the $i^{th}$ level, the total variation over leaves, which is a convex combination of at most $\epsilon$-altered conditional distributions, can only change by at most $\epsilon$ per step. A union bound gives the final $\epsilon N$ total variation bound.
\end{proof}
Now we proceed with proving Theorem \ref{thm:NBQP_vs_BQP} by demonstrating that the lifted Simon's problem cannot be solved with polynomial queries by an $\nbqp$ machine.
\begin{theorem}
    No $\nbqp^{O_{\tilde{f}}}$ algorithm making at most \textnormal{poly}$(n)$ queries to ${O_{\tilde{f}}}$ can solve the lifted Simon's problem with probability at least $2/3$.
\end{theorem}
\begin{proof}
Let $\mathcal{T}$ be a learning tree for an $\nbqp_\lambda$ algorithm making at most $N$ queries to $O_{\tilde{f}}$. Combining Lemma \ref{lemma:tree_id_vs_o} and \ref{lemma:node_perturbation}, replacing every $\lambda$-noisy quantum circuit in $\mathcal{T}$ that queries $\mathcal{O}_{\tilde{f}}$ with one that queries the identity oracle results in a leaf distribution $p_I$ that is at most $N^2\exp(-\Omega(\lambda n))$ far in total variation from the original $p_{\tilde{f}}$.

Now, note that the distribution $p_{I}$ results from an algorithm that makes no queries to the oracle, so $\mathbb{E}_{f \text{ 1-to-1}}[p_I] =  \mathbb{E}_{f \text{ 2-to-1}}[p_I]$. By the triangle inequality, we see that
$$d_{\text{\rm TV}}(\mathbb{E}_{f \text{ 1-to-1}}[p_{\tilde{f}}], \mathbb{E}_{f \text{ 2-to-1}}[p_{\tilde{f}}]) \leq N^2\exp(-\Omega(\lambda n))\,.$$
Using Lemma \ref{lemma:le_cam}, when $N \leq \Omega(\exp(\lambda n))$, the $\nbqp_\lambda$ algorithm cannot succeed with probability $\geq 2/3$.
\end{proof}
This concludes the proof of Theorem \ref{thm:NBQP_vs_BQP}. Note that we did not need to use the fact that every depth-1 layer of a circuit employed by an $\nbqp$ algorithm is noisy; the oracle constructed is so sensitive to noise on its inputs that a single layer of noise at the outset is sufficient for this exponential separation.

\section{Quantum Advantage in Noisy Purity Testing} \label{appendix:purity_testing_lb}
A well-studied candidate for achieving quantum advantage in learning from entangled measurements is the problem of estimating the purity of a quantum state. Several works demonstrate that without at least $n$ additional qubits of quantum memory, the sample complexity of the simpler \textit{purity testing} problem (distinguishing between a $n$-qubit Haar-random pure state and a maximally mixed state) scales exponentially with $n$ \cite{aharonov1999faulttolerantquantumcomputationconstant, chen2021exponentialseparationslearningquantum, gong2024samplecomplexitypurityinner, weiyuan_paulis}. On the other hand, running a simple \textsf{SWAP} test on two copies of the state a constant number of times suffices to distinguish these cases, because while the inner product of a pure state with itself is $1$, the output of a \textsf{SWAP} test on two copies of a maximally mixed state is exponentially close to zero. In Section  \ref{app:noisy_purity_testing}, we show that this supposed quantum advantage breaks down in the presence of order-1 local depolarizing noise.

Implicitly, this result demonstrates that the \textsf{SWAP} test, a crucial subroutine in many quantum-enhanced algorithms and learning/property-testing protocols  \cite{Huang_2022, Cincio_2018, Lloyd_2014, bădescu2017quantumstatecertification, Harrow:2012gwf, montanaro2018surveyquantumpropertytesting}, is degraded in the presence of noise by a factor exponential in system size. Moreover, no adaptive strategy can be used to circumvent this fact. To recover quantum advantage when performing \textsf{SWAP} operations on entire copies of an uncharacterized quantum state or process, the underlying physical system must exhibit an intrinsic robustness to noise, effectively enabling a form of quantum error-correction.  In Section \ref{app:happy_code}, we construct such a system using a tensor-network model for holographic duality known as the HaPPY code \cite{happy_2015}. The purity testing problem is reformulated as detecting the presence of a bulk black hole microstate from noisy measurements of the dual boundary state.

\subsection{Breakdown of advantage in noisy two-copy purity testing} \label{app:noisy_purity_testing}
In this section, we prove an exponential-in-$n$ lower bound on the sample complexity of purity testing with access to two noisy copies of an unknown $n$-qubit quantum state. Any algorithm for this task can be represented by a learning tree where, at each node, the algorithm can perform arbitrary noisy joint measurements on two copies of the unknown state. We formalize this as follows, making reference to the general learning tree representation from Definition \ref{def:general_learning_tree}.

\begin{definition}[Learning tree for noisy two-copy algorithm]
    Any $\lambda$-noisy quantum algorithm with query access to two noisy copies of the state $\rho$, i.e. $\DN[\rho] \otimes \DN[\rho]$, can be represented by a learning tree $\mathcal{T}$, where at each node $u$ of $\mathcal{T}$ the algorithm can measure an arbitrary $2n$-qubit POVM $M = \{F_s\}_u$. Since every such POVM can be simulated by a rank-1 POVM of the form $\{2^{2n}w_s^u\ketbra{\psi_s^u}{\psi_s^u}\}$ with all $w_s \geq 0$ and $\sum_s w_s = 1$, the transition rule is
\begin{equation}
    p_\rho(v) = p_\rho(u)2^{2n}w_s^u\tr(\DN[\rho]\otimes \DN[\rho]\ketbra{\psi_s^u}{\psi_s^u})\,.
\end{equation}
\end{definition}
We remark that this model is actually more powerful than an algorithm using $\lambda$-noisy quantum circuits, because the POVMs need not include noise layers between every depth-1 unitary in their construction. However, because the set of all $n$-qubit POVMs contains the set of POVMs corresponding to $\lambda$-noisy circuits, any lower bound obtained with this learning tree applies to the interstitial-noise model. Our lower bounds thus hold against any algorithm which experiences only a single layer of depolarizing noise at the outset. We implicitly use this fact in all learning tree lower bounds hereafter.

\begin{theorem}[Formal version of Theorem \ref{thm:purity_testing_lb_informal}] Let $\rho$, with equal probability, be either the maximally mixed state $\mathds{1}/2^n$ or $\ketbra{\psi}{\psi}$ where $\ket{\psi}$ is sampled from the Haar measure over $n$-qubit pure states. Any algorithm with the ability to perform arbitrary measurements on $\DN[\rho]\otimes \DN[\rho]$ requires 
\begin{align}
\Omega\left(\min\!\left\{2^{n/2},\left(\frac{4}{1+3(1-\lambda)^4}\right)^n\right\}\right)
\end{align}
samples to distinguish the two cases with high success probability.
\end{theorem}
\begin{proof}
    Our learning tree toolbox, Lemma \ref{lemma:toolbox}, admits two types of strategies. First, one can bounding the fluctuations of the likelihood ratio at every intermediate node of the learning tree. When this bound is insufficient, the more sophisticated approach requires bounding the likelihood ratio over entire root-to-leaf paths of the tree, and arguing that this bound holds with high probability for almost all paths. Our approach for purity testing will begin with the path-based strategy; however, in doing so we will reformulate the problem in terms of a separate learning tree bound which can be completed using the simpler node-based strategy.

    We begin by fixing a root-to-leaf path of our 2-copy learning tree, which is specified by a set of $2n$-qubit POVMs $\{F_s^t\}_{s,t}$ with $t=1,...,T$ specifying the layer of the tree. The final node on this path will be the leaf node $\ell$. The likelihood ratio along this path is
    \begin{align}
    L_{\psi}(\ell)=\prod_{t=1}^T\frac{\tr\!\left(F_s^t(\DN\otimes\DN)[\ketbra{\psi}{\psi}^{\otimes 2}]\right)}{\tr(F_s^t(\mathds{1}/2^n)^{\otimes 2})}=\frac{\tr\!\left(\bigotimes_{t=1}^TF_s^t\cdot \mathcal{D}_\lambda^{\otimes 2nT}[\ketbra{\psi}{\psi}^{\otimes 2T}]\right)}{\tr(\bigotimes_{t=1}^TF_s^t\cdot(\mathds{1}/2^n)^{\otimes 2T})}.
    \end{align}
    Computing the Haar average, we obtain:
    \begin{align}
    L(\ell)&=\mathbb{E}_{\psi}[L_{\psi}(\ell)]\\
    &=\frac{(2^n)^{2T}}{2^n(2^n+1)...(2^n+2T-1)}\frac{\tr\!\left(\bigotimes_{t=1}^TF_s^t\cdot \mathcal{D}_\lambda^{\otimes 2nT}(S_{2T})\right)}{\tr(\bigotimes_{t=1}^TF_s^t)}\\
    &\geq\left(1-\frac{4T^2}{2^n}\right)\frac{\tr\!\left(\bigotimes_{t=1}^TF_s^t\cdot \mathcal{D}_\lambda^{\otimes 2nT}(S_{2T})\right)}{\prod_{t=1}^T\tr(F_s^t)},
    \end{align}
    where $S_x$ is the sum over all permutation operators on $x$ copies of the $n$-qubit Hilbert space. From~\cite[Lemma 16]{weiyuan_paulis} we have
    \begin{align}
    \tr(\rho_x\otimes\rho_y\cdot S_{x+y})\geq\tr(\rho_xS_x)\tr(\rho_yS_y)
    \end{align}
    for any positive semi-definite matrix $\rho_x$ and $\rho_y$.  Using $\tr(A\otimes B) = \tr(A)\tr(B)$, we have
    \begin{align}
    L(\ell)&\geq\left(1-\frac{4T^2}{2^n}\right)\prod_{t=1}^T\frac{\tr(F_s^t\cdot \mathcal{D}_\lambda^{\otimes 2n}(S_{2}))}{\tr(F_s^t)}\,.
    \end{align}
    
    At this stage, we notice that this expression can be recast as a likelihood ratio of distinguishing two specific states. Since $S_2 = \textsf{SWAP}_n + \mathds{1}_{2n}$ has $\tr(S_2) = 2^n(2^n+1)$ and is positive semi-definite, we see that $\rho_S = S_2/(2^n(2^n+1))$ is a valid quantum state. Defining $\rho_m = \mathds{1}_{2n}/2^{2n}$, we thus have
    \begin{equation}
        L(\ell) \geq \left(1-\frac{4T^2}{2^n}\right)\prod_{t=1}^T\frac{\tr(F_s^t\cdot \mathcal{D}_\lambda^{\otimes 2n}(S_{2}))}{\tr(F_s^t)} \geq \left(1-\frac{4T^2}{2^n}\right)\prod_{t=1}^T\frac{\tr(F_s^t\cdot \mathcal{D}_\lambda^{\otimes 2n}[\rho_S])}{\tr(F_s^t\rho_m)}\,,
    \end{equation}
    where the product in the final expression is the likelihood ratio $L'(\ell)$ for another learning tree $\mathcal{T'}$ distinguishing between input states $\mathcal{D}_\lambda^{\otimes 2n}[\rho_S]$ and $\rho_m$\,. This is the crucial reformulation; now, the learning tree bound on $\mathcal{T}'$ can be completed using the node-based approach. By Lemma \ref{lemma:toolbox}, we have that if for all nodes $u$ in $\mathcal{T}'$,
    \begin{equation}
            \mathbb{E}_{s\sim p^{\rho_m}(s|u)} \left[\left(L'(s|u) - 1\right)^2\right]\leq \delta\,,
    \end{equation}
    then there exists some $C>0$ such that $\textnormal{Pr}_{\ell\sim p^{\rho_m}(\ell)}[L'(\ell) > 0.9] \geq 0.9 - C\delta T$, where $T$ is the depth of $\mathcal{T}'$. Hence, our next step is to bound this node-wise concentration of the likelihood ratio in our new tree $\mathcal{T}'$.
    
    We omit $t$ in the POVM elements for convenience. Then, letting $\gamma = {2^n(2^n+1)}/{2^{2n}}$
    \begin{align}
        \mathbb{E}_{s\sim p^{\rho_m}(s|u)} \left[\left(L'(s|u) - 1\right)^2\right] &= \sum_s \tr(F_s\rho_m)\left(\gamma\,\frac{\tr(F_s\mathcal{D}_\lambda^{\otimes 2n}[S_2])}{\tr(F_s)} - 1\right)^2
        \\
        &\leq 2^{-2n + 1} + 2\sum_s \tr(F_s\rho_m)\left(\frac{\tr(F_s\mathcal{D}_\lambda^{\otimes 2n}[\textsf{SWAP}_n])}{\tr(F_s)
        }\right)^2 \\
        &\leq 2^{-2n + 1} + 2\sum_s \frac{\tr(F_s \mathcal{D}_\lambda^{\otimes 2n}[\textsf{SWAP}_n])^2}{2^{2n}\tr(F_s)}\,.
    \end{align}    
    We can bound the remaining sum as follows. Let $A$ be any Hermitian operator and $\{F_s\}$ any POVM with $F_s\succeq 0$ and $\sum_s F_s=I$. Write $\tr(F_s A)=\tr(\sqrt{F_s}\sqrt{F_s}A)$. By Cauchy-Schwarz for the Hilbert-Schmidt inner product,
    \begin{align}
    \bigl(\tr(F_s A)\bigr)^2
    \le \tr(F_s)\tr\bigl(\sqrt{F_s}A\sqrt{F_s}A\bigr)
    = \tr(F_s)\tr(F_s A^2)\,.
    \end{align}
    Dividing by $\tr(F_s)$ and summing over $s$, we have
    \begin{align}
    \sum_s \frac{(\tr(F_s A))^2}{\tr(F_s)}
    \le \sum_s \tr(F_s A^2)
    = \tr\Bigl((\sum_s F_s)A^2\Bigr)
    = \tr(A^2).
    \end{align}
    Let us apply the above to $A=\mathcal D_\lambda^{\otimes 2n}(\textsf{SWAP}_n)$. Doing so, we find
    \begin{align}
    \sum_s \frac{\bigl(\tr(F_s A)\bigr)^2}{2^{2n}\tr(F_s)}\le\frac{\tr(A^2)}{2^{2n}}.
    \end{align}
    So it suffices to compute $\tr(D_\lambda^{\otimes 2n}[\textsf{SWAP}]^2)$.
    Since $A$ factorizes over each site,
    \begin{align}
    A=D_\lambda^{\otimes 2n}(\textsf{SWAP}_{n})=\bigotimes_{i=1}^n (\mathcal D_\lambda\otimes \mathcal D_\lambda)(\textsf{SWAP}_1) =: \bigotimes_{i=1}^n a\,,
    \end{align}
    with $a$ acting on two qubits (one pair). Hence, $\tr(A^2)=\bigl(\tr(a^2)\bigr)^{n}$. Recalling that $a=\lambda(2-\lambda) \mathds{1} + (1-\lambda)^2 \textsf{SWAP}$, we define
    \begin{align}
    \alpha := (1-\lambda)^2,\quad \beta := \tfrac{1-\alpha}{2}=\tfrac{2\lambda-\lambda^2}{2}\,.
    \end{align}
    With this notation,
    \begin{align}
    a^2=(\beta I+\alpha \textsf{SWAP})^2=(\beta^2+\alpha^2)I+2\alpha\beta\,\textsf{SWAP}\,.
    \end{align}
    Using $\beta=(1-\alpha)/2$, we find
    \begin{align}
    \tr(a^2) = 1+3\alpha^2 = 1+3(1-\lambda)^4\,,
    \end{align}
    giving us
    \begin{align}
    \sum_s \frac{\bigl(\tr(F_s A)\bigr)^2}{2^{2n}\tr(F_s)}\le\frac{\tr(A^2)}{2^{2n}}\leq\left(\frac{1+3(1-\lambda)^4}{4}\right)^n.
    \end{align}
    Substituting into our likelihood ratio bound, we obtain
    \begin{align}
        \mathbb{E}_{s\sim p^{\rho_m}(s|u)} \left[\left(L'(s|u) - 1\right)^2\right] \leq 2^{-2n+1}+2 \left(\frac{1+3(1-\lambda)^4}{4}\right)^n\,.
    \end{align}
    By Lemma \ref{lemma:toolbox}, there is a constant $C$ such that 
    \begin{equation}
        \textnormal{Pr}_{\ell\sim p^{\rho_m}(\ell)}\left[L(\ell) =  \left(1-\frac{4T^2}{2^n}\right)L'(\ell) > 0.9 \left(1-\frac{4T^2}{2^n}\right)\right] \geq 0.9 - CT\left[2^{-2n+1}+2 \left(\frac{1+3(1-\lambda)^4}{4}\right)^n\right].
    \end{equation}
    If the tree has depth
    \begin{equation}
        T < \min\left(\frac{2^{n/2}}{20}, \frac{1}{400C}\left(\frac{4}{1+3(1-\lambda)^4}\right)^n\right),
    \end{equation}
    we find     
    \begin{equation}
        CT\left[2^{-2n+1}+2 \left(\frac{1+3(1-\lambda)^4}{4}\right)^n\right] \leq 0.01\,,
    \end{equation}
    and $0.9 (1-{4T^2}/{2^n}) \geq 0.89$. Applying this limit on $T$, we obtain the bound 
    \begin{equation}
        \textnormal{Pr}_{\ell\sim p^{\rho_m}(\ell)}\left[L(\ell)  > 0.89 \right] \leq 0.89\,.
    \end{equation}
    Applying (3) from Lemma~\ref{lemma:toolbox} with Lemma \ref{lemma:le_cam}, the probability that the algorithm succeeds in distinguishing the pure and maximally mixed cases is upper bounded by $2(1-0.89) = 0.22$. Hence, with this restriction on $T$, no algorithm can succeed with probability at least $2/3$.
    This yields the final lower bound on the tree depth of 
    \begin{align}
    \Omega\left(\min\left\{2^{n/2},\left(\frac{4}{1+3(1-\lambda)^4}\right)^n\right\}\right).
    \end{align}

\end{proof}

\subsection{Black hole detection in holographic duality}\label{app:happy_code}

We have demonstrated that the vanilla purity test, including any (adaptive) strategy based on \textsf{SWAP} tests, is exponentially degraded by local noise. To study when an ideal purity test can remain robust to noise in a concrete, physically motivated setting, we take as a testbed a tensor-network model of holographic duality. In such models, the microscopic degrees of freedom describing a black hole in the bulk are encoded nonlocally into a dual conformal field theory on the boundary. Operationally, an experimentalist coupled only to the boundary system faces the following learning problem: given noisy boundary data, can one tell whether the bulk black hole is in a single pure microstate or in a mixed ensemble?

We instantiate this scenario using the HaPPY code~\cite{happy_2015}, a holographic tensor-network code that implements an isometric bulk-boundary encoding and exhibits quantum error-correcting structure. In our construction, the inner bulk legs of the network encode the black hole degrees of freedom, all remaining bulk legs are fixed in a reference state, and the learner has noisy access only to boundary qubits. The task is then to distinguish, using only such noisy boundary measurements, whether the bulk black hole region is prepared in a single Haar-random pure state or in the maximally mixed state on the same Hilbert space. We now spell out this setup more precisely.

\subsubsection{Defining the task} \label{appendix:defining_happy_code_prob}
In this construction, we work with the so-called \textit{hexagonal} HaPPY code. As depicted in Figure \ref{fig:happy_code}(a), each tensor in the network has six (two-dimensional) legs. In this model, the network is built from a central (level-0) tile, which is contracted to six neighboring level-1 tiles. Each level-1 tile is then contracted to two shared level-2 tiles, and three level-2 tiles which only share a leg with a single level-1 tile (see Figure \ref{fig:happy_code}). The \textit{radius} of the tensor network is the least number of legs between a tile at the boundary of the network and the center of the network.

We then remove all tiles within radius $r$ of the center. This reveals a number of uncontracted internal legs, which are then contracted to either a maximally mixed state or a Haar-random pure state on the correct Hilbert space dimension. For later use, we now quantify the size of this inserted state as a function of $r$.

The hexagonal tiling has two types of tiles: those which are connected to $1$ tile of a lower radius and $5$ tiles of a higher radius, and those connected to $2$ lower and $4$ higher tiles. Letting $x_k$ be the number of $1, 5$-type tiles at radius $k$ and $y_k$ be the number of $2, 4$ tiles, we have the recurrence relations
\begin{align}
    x_k &= 3x_{k-1} + 4y_{k-1} \\
    y_k &= x_{k-1}
\end{align}
with $x_1 = 6, x_2 = 18$. Solving for a closed form in $x$ gives
\begin{equation}
    x_k = \frac{24}{5}4^{k-1} + \frac{6}{5}(-1)^{k-1}\,.
\end{equation}
The number of unconctracted legs upon excising all tiles up to radius $r$ is $5x_k + 4x_{k-1}$ because $y_k = x_{k-1}$, each $1,5$ tile leaves behind $5$ uncontracted legs, and $2,4$ tile leaves behind $4$ uncontracted legs. So the number of legs $L_k$ is
\begin{equation}
    L_k = \frac{144}{5}4^{k-1} + \frac{6}{5}(-1)^{k-1} = \Theta(4^{k-1})\,.
\end{equation}
Hence the number of legs contracted to the inserted black hole will go like $6, 30, 114, 462, ...$ and each leg corresponds to a qubit dimension. So when we say we insert a ``black hole microstate'' of radius $r$ or a ``maximally mixed state'' of radius $r$, this corresponds to a state on $L_r$ qubits, which grows exponentially in the radius. We let $\mathcal{H}_r$ denote the Hilbert space of dimension $2^{L_r}$, corresponding to a radius$-r$ tensor network state

With the model understood we consider the following task. Given a hexagonal HaPPY code state $\rho_{\text{phys}}$, subjected to an erasure channel of strength $\lambda$ on every out-leg, can we detect whether the logical state of radius $r$ embedded into the network is pure? In this section, we argue that given access to only single copies of $\rho_{\text{phys}}$, a number of copies doubly-exponential in $r$ is required to solve the distinguishing task, while only a constant number is required using quantum processing and two-copy measurements.
\begin{theorem}{\textnormal{(Formal statement of Theorem \ref{thm:happy_code_informal})}} \label{thm:happy_code_formal}
     Consider a hexagonal HaPPY code of radius $R$ with all tiles within a radius $r$ excised and replaced by a quantum state $\rho$ defined on $\mathcal{H}_r$. We are given that $R \geq r + O(\log r)$. 
     The quantum state which is inserted into the bulk is, with equal probability, either a maximally mixed state $\mathds{1}/2^{L_r}$, or a pure state $\ketbra{\psi}{\psi}$ sampled from the Haar measure over $L_r$-qubit states. This state is mapped to a boundary state, which is then corrupted by an erasure channel $\mathcal{E}_\lambda$ erasing every  out-leg with an independent probability $\lambda < 1/48$.
     
     Then there is a quantum algorithm, using joint measurements on two copies of the corrupted boundary state, which can distinguish the pure and mixed bulk hypotheses with high probability, using only $O(1)$ copies of the state. However, any quantum algorithm using only single copies requires $\Omega(2^{\exp(r)})$ measurements to do so.
\end{theorem}

The remainder of this section is dedicated to the proof of Theorem \ref{thm:happy_code_formal}. While we consider the erasure channel for a provable guarantee, we remark that the hexagonal HaPPY code exhibits numerical evidence of error thresholds against more general Pauli noise channels \cite{numerics_happy_decoding_2020, Farrelly_2021, Jahn_2021}, and it is likely that the two-copy distinguishing procedure will remain efficient under these channels (including the local depolarizing noise model considered in the rest of this work), by using more sophisticated decoders than the one considered in our upper bound. 

\subsubsection{Intractability with single copies}
First, we argue that any algorithm, even in the noiseless setting, which only measures single copies of the boundary state at a time requies at least $\Omega(2^{L_r/2})$ measurements to determine whether a black hole lies in the bulk with high probability, where we recall that $L_r = \Theta(4^{r-1})$ is the base-2 logarithm of the black hole's Hilbert space dimension.

Consider the many-vs.-one distinguishing problem from Appendix~\ref{appendix:purity_testing_lb}, i.e.~distinguishing between an $n$-qubit state $\rho$ that is either maximally mixed or sampled from the Haar measure over pure states. \cite{chen2021exponentialseparationslearningquantum} demonstrates that any noiseless algorithm without quantum memory (with a learning tree representation given by Definition \ref{def:memoryless_learning_tree}, without the noise channels), with access to individual copies of $\rho$, requires $\Omega(2^{n/2})$ samples to solve this discrimination problem with high probability. Note that the decision-version of our task of detecting the presence of a black hole can be reformulated precisely as the purity testing problem, with the distinction that the unknown state is Haar random over a logical subspace that is then mapped by isometries to a state on an exponentially larger physical Hilbert space. The crucial observation is that this task reduces to the information-theoretic lower bound obtained from a memoryless learning tree for standard purity testing.

The reasoning is as follows. To align notation, let $n = L_r$. Let $\rho_{\text{phys}}$ denote the noiseless holographic boundary state on $m>n$ qubit dimensions. Moreover, let $\rho_{\text{log}}$ denote the $n$-qubit logical quantum state inserted into the holographic tensor network, which is either a Haar-random pure state or a maximally mixed state on $n$ qubits. Finally, let $\Lambda$ denote the isometry mapping $\rho_{\text{log}}\rightarrow \rho_{\text{phys}}$. Then the bulk-boundary channel instantiated by the tensor network is a CPTP map $\mathcal{E}_\Lambda$ such that
\begin{equation}
    \mathcal{E}_\Lambda(\rho_{\text{log}} \otimes \ketbra{0}{0}^{m-n}) = \rho_{\text{phys}}\,.
\end{equation}
Any single-copy measurement on $\rho_{\text{phys}}$ is described by an arbitrary $m$-qubit POVM $\{F_s^{\text{phys}}\}$. The distribution over outcomes under this POVM is thus given by
\begin{equation}
    \text{Pr}[s] = \tr(F_s^{\text{phys}} \mathcal{E}_\Lambda(\rho_{\text{log}} \otimes \ketbra{0}{0}^{m-n}))\,.
\end{equation}
We can then precompose each $F_s^{\text{phys}}$ with $\mathcal{E}_\Lambda$ to obtain POVM elements $F_s^{\text{phys} '}$ such that 
\begin{equation}
    \text{Pr}[s] = \tr (F_s^{\text{phys}'} (\rho_{\text{log}} \otimes \ketbra{0}{0}^{m-n}))\,,
\end{equation}
noting that linearity implies that $\{F_s^{\text{phys} '}\}_s$ is still a POVM. Now define the set of operators $\{F_s^{\text{log}}\}_s$ according to
\begin{equation}
    F_s^{\text{log}} = \tr_{>n}\!\Big(F_s^{\text{phys}'}\big(\mathds{1}_n \otimes \ketbra{0}{0}^{\otimes m-n}\big)\Big),
\end{equation}
where the trace acts on all but the qubits corresponding to $\rho_{\text{log}}$. It is simple to see that $\{F_s^{\text{log}}\}_s$ is a POVM:
\begin{align}
    \sum_s  F_s^{\text{log}} &= \sum_s\tr_{>n}\!\Big(F_s^{\text{phys}'}\big(\mathds{1}_n \otimes \ketbra{0}{0}^{\otimes m-n}\big)\Big)\\
    &= \tr_{>n}\!\Big(\mathds{1}_{m}\big(\mathds{1}_n \otimes \ketbra{0}{0}^{\otimes m-n}\big)\Big)\\
    &= \mathds{1}_n\,.
\end{align}
Crucially, we immediately obtain that measuring $\{F_s^{\text{log}}\}_s$ on $\rho_{\text{log}}$ gives us precisely the same distribution over classical outcomes as measuring  $\{F_s^{\text{phys}}\}_s$ on $\rho_{\text{phys}}$. By this argument, any arbitrary POVM on the physical boundary state has an outcome distribution that can be simulated virtually by a POVM on the bulk logical state corresponding to the black hole or maximally mixed state, when we are discarding the quantum state after each experiment. All such POVMs are included in the memoryless learning tree representation for purity testing. Thus, the purity testing lower bound of $\Omega(2^{n/2})$ measurements immediately ports to the task of black hole detection. This is also a lower bound for the noisy setting, because a layer of depolarizing noise can be included in the description of a POVM on the physical state. Since $n$ is exponential in $r$ as quantified in \ref{appendix:defining_happy_code_prob}, the sample complexity is at least doubly-exponential in the radius of the hidden state.

\subsubsection{Noise-robust detection using joint measurements}
As in the lower bound, we utilize a basic reduction to the standard purity-testing problem to establishes our $O(1)$ two-copy upper bound. The first step is to use the intrinsic error-correction properties of the tensor network. In particular, \cite{happy_2015} proposes a simple hierarchical majority-vote strategy for decoding the boundary state, known as the greedy decoder. While less robust than other decoding strategies, the greedy decoder allows us to obtain rigorous decoding guarantees. In particular, \cite{happy_2015} gives us the following lemma.
\begin{lemma} \label{lemma:happy_error_prop}
     Consider a hexagonal HaPPY code of radius $R$ with all tiles within a radius $r$ excised and replaced by an arbitrary quantum state $\rho$ defined on $\mathcal{H}_r$. Let the boundary state, corrupted by an erasure channel $\mathcal{E}_\lambda$, be $\mathcal{E}_\lambda(\rho_{b})$.
     Then there is an efficient decoding map $\mathcal{G}:\mathcal{H}_R\rightarrow\mathcal{H}_r$ such that 
     \begin{equation}
         \textnormal{Pr}[\mathcal{G}(\mathcal{E}_\lambda(\rho_{b})) \neq \rho] \leq \frac{30\cdot 4^{r-1}}{12}(12\lambda)^{\varphi^{R-r}} \,,
     \end{equation}
     where $\varphi = (1+\sqrt{5})/2$.
\end{lemma}
\begin{proof}
    Under an erasure channel of strength $\lambda$, each out-leg of the boundary state is corrupted independently with probability $\lambda$. \cite{happy_2015} shows that when the greedy decoder is applied to $t$ radial layers of the tensor network, the probability that an erasure error has propagated to a particular leg at radius $R-t$ is $\leq \frac{1}{12}(12\lambda)^{\varphi^{R-t}}$. Applying a union bound over all legs at radius $r$, which we have shown is 
    \begin{equation}
        L_r = \frac{144}{5}4^{k-1} + \frac{6}{5}(-1)^{k-1} \leq 30\cdot 4^{r-1}\,,
    \end{equation}
    we obtain the bound on the error probability of the overall bulk state. 
\end{proof}
Now, let the error rate be $\lambda \leq 1/48$. Using Lemma \ref{lemma:happy_error_prop}, the probability that the greedy decoder fails to decode perfectly is bounded by $2.5\cdot4^{r-1}\cdot(1/4)^{\varphi^{R-r}} =  2.5 \cdot 4^{r-1-\varphi^{R-r}}$ which is asymptotically exponentially small in $r$ when $R \geq r + c\log r$ with a relatively small constant $c>1$.

Hence, given two copies of $\mathcal{E}_\lambda(\rho_{b})$, we can apply the greedy decoder to both, and with probability $\geq 1-O(\exp\exp(-r(c-1)))$, we obtain two copies of $\rho$. In our problem, $\rho$ is either maximally mixed or a Haar-random pure state $\ketbra{\psi}{\psi}$. Our problem is then reduced to the well-studied purity testing problem, for which there is a simple two-copy algorithm. One can simply run a $\textsf{SWAP}$ test $\textsf{SWAP}(\rho, \rho)$ (e.g., as described in \cite{montanaro2018surveyquantumpropertytesting}); on two arbitrary states $\rho, \sigma$, the test accepts with probability 
\begin{equation}
    \frac{1}{2} +\frac{\tr(\rho\sigma)}{2}\,,
\end{equation}
and rejects otherwise. If $\rho$ is a fixed pure state, the test always accepts, whereas if it is maximally mixed, the test accepts with probability exponentially close to $1/2$. Assuming the two copies of $\rho$ have been decoded perfectly, $O(\log 1/\delta)$ repetitions of the $\textsf{SWAP}$ test guarantees a success probability in detecting the pure state of at least $1-O(\delta)$. Since we have shown the decoding fails with probability at most exponentially small in the radius of the embedded state, $O(1)$ samples are sufficient, using greedy decoding and \textsf{SWAP} tests, to detect a pure state in the bulk with high probability.

\section{Pauli Tomography in Noisy Quantum Learning Theory} \label{app:quantum_advantage} 
In this section, we construct a quantum state-discrimination task that is strictly easier than performing Pauli shadow tomography \cite{weiyuan_paulis}. We show that any quantum algorithm restricted to single-copy measurements requires at least $(2/f(\lambda))^n$ measurements to succeed with high probability, where $f(\lambda)\in[0,1]$ depends inversely on the noise rate. More generally, even with $k$ qubits of quantum memory, any learning strategy still needs $2^{\,n-k}(1-\lambda)^{-n}$ samples. These bounds immediately yield lower bounds for noisy Pauli shadow tomography. Since learning Pauli observables (or stabilizer states \cite{montanaro2017learningstabilizerstatesbell}) is the canonical application of Bell sampling, our information-theoretic lower bounds — which require only a single invocation of the noise channel — demonstrate that Bell-measurement-based strategies are exponentially degraded by local noise.

We then give a single-copy algorithm for noisy Pauli shadow tomography and a two-copy algorithm which can solve the state-discrimination task. The latter has sample complexity matching the lower bound for $k=n$ qubits of quantum memory up to nonleading factors of $n$ and constants in the exponents. This establishes a polynomial separation, dependent on the noise rate, between our two-copy algorithm and any single-copy protocol. As the noise rate tends to zero, this polynomial advantage becomes exponential, recovering the ideal quantum advantage for Pauli shadow tomography. We begin by defining the broader learning problem.

\begin{definition}[Pauli shadow tomography] \label{def:pauli_shadows}Given copies of an unknown quantum state $\rho$, estimate all $4^n$ values $\tr(P\rho)$ for all $P\in \mathcal{P}_n$ to within additive error $\epsilon$, where $\epsilon$ is a positive constant.
\end{definition}
This physically motivated learning task has been extensively studied, and optimal bounds on its sample complexity in the noiseless setting are given in \cite{weiyuan_paulis}. Here, we treat this problem in the context of $\nbqp$-type errors. Our bounds neglect the well-established $\epsilon$-dependence of this task for clarity and to emphasize the interplay between noise rate and instance size, as $\epsilon$-dependence is not altered by noise. For our lower bounds, we consider three models of $\lambda$-noisy quantum algorithms: 
\begin{itemize}
    \item Algorithms without any quantum memory, but the ability to perform arbitrary measurements on individual copies of the unknown state
    \item Algorithms with $k$ qubits of quantum memory, but each circuit can only perform a constant number of queries before measurement (i.e.~quantum memory with a limited lifetime), and classical advice can be passed between quantum circuits.
    \item Algorithms on a $k$ qubit quantum computer coupled to arbitrarily many sequential copies of the unknown state
\end{itemize}
We provide separate lower bounds on the sample complexity of these three models, all of which are at least exponential in $n$ (or $n-k$ given $k$ qubits of memory). To prove these bounds, we state the following many-vs.-one decision problem.

\begin{definition}[Decision I+P problem] \label{def:dec_ip_problen}
Let $D$ be a distribution over $\mathcal{P}_n$. Let $O$ be a state preparation oracle which loads a fixed state given by either $\rho = (\mathds{1} + P)/2^n$ with $P$ sampled from $D$ or $\rho = \mathds{1}/2^n$, where $\mathds{1}$ denotes the $n$-qubit identity matrix. The decision I+P problem $\textnormal{Dec-IP}(n, D)$ is to distinguish the two cases with success probability at least $2/3$, given query access to $O$.
\end{definition}

\noindent For later use, we define $\overline{\mathcal{P}}_n := \mathcal{P}_n \setminus \{\mathds{1}\}$.  We call $w$ the \textit{weight} of $P$, and denote it by $w = |P|$. Observe that any algorithm that solves Pauli shadow tomography with e.g. $\epsilon < 1/3$ and success probability at least $2/3$ can solve $\textnormal{Dec-IP}(n, \textnormal{Unif}(\overline{\mathcal{P}}_n))$. The reasoning is simple; given an algorithm for Pauli shadow tomography, a learner attempting to solve Dec-IP can simply run the tomography algorithm to obtain estimates for all Pauli expectations, then choose the one with the largest absolute expectation. If this value rounds to $0$, we output the maximally mixed hypothesis, and if it rounds to $1$, we output the alternate hypothesis. Thus, any algorithm requiring at least $M$ queries to $O$ to solve $\textnormal{Dec-IP}(n, \textnormal{Unif}(\overline{\mathcal{P}}_n))$ requires at least $M$ copies of $\rho$ to solve Pauli shadow tomography. The same holds when the distribution is restricted to uniform over $\{X, Y, Z\}^{\otimes n}$, a fact we use in the following lower bound.

\subsection{Lower bound for single-copy measurements without quantum memory}
Now we demonstrate a lower bound on the first of our three single-copy learning models, where the algorithm has no access to quantum memory. This setting is depicted in Figure \ref{fig:quantum_memory}(a). Any such algorithm can be formally represented by a learning tree (in the setting without a quantum memory), in which each copy of the uncharacterized state will be subject to a layer of noise before the algorithm can process it.
\begin{definition}[Learning tree without quantum memory]\label{def:memoryless_learning_tree}
Any $\lambda$-noisy learning algorithm with access to a fixed quantum state $\DN[\rho]$ and no quantum memory can be represented by a learning tree $\mathcal{T}$, without an ancillary quantum register. At each node, the algorithm can select an arbitrary POVM on $n$-qubits. 
Since every such POVM can be simulated by a rank-1 POVM of the form $\{2^nw_s^u\ketbra{\psi_s^u}{\psi_s^u}\}$ with all $w_s \geq 0$ and $\sum_s w_s = 1$, the transition rule is
\begin{equation}
    p_\rho(v) = p_\rho(u)2^nw_s^u\tr(\DN[\rho]\ketbra{\psi_s^u}{\psi_s^u})\,.
\end{equation}
\end{definition}

Before proceeding, we state two lemmas characterizing the operator norm of $\textsf{SWAP}$ operators under the depolarizing channel, whose proofs we defer to Appendix \ref{appendix:depol_lemma_proofs}.
\begin{lemma}\label{lemma:depol_on_swap}
    Define $f(\lambda) := 1-\lambda+\lambda^2/2$ for $\lambda\in [0, 1]$. Then for any $n$-qubit pure state $\ket{\phi},$
    \begin{equation}
        \tr(|\phi\rangle \langle \phi|^{\otimes 2} (\mathcal{D}_\lambda^{\otimes n} \otimes \mathcal{D}_\lambda^{\otimes n})[\textnormal{\textsf{SWAP}}_n]) \leq {f(\lambda)}^n\,,
    \end{equation}
    where we note that $\tfrac{1}{2}\,f(\lambda) \in [\tfrac{1}{2}, 1]$.
\end{lemma}
\begin{lemma}\label{lemma:depol_on_swap_2}
    For any $n$-qubit pure state $\ket{\phi}$,
    \begin{equation}
        \tr(|\phi\rangle \langle \phi|^{\otimes 2} (\mathcal{D}_\lambda^{\otimes n} \otimes \mathcal{D}_\lambda^{\otimes n})[(2\textnormal{\textsf{SWAP}}_1 - \mathds{1}_2^{\otimes 2})^{\otimes n}]) \leq \frac{(1+3^n)(1-\lambda)^{2n}}{2}
    \end{equation}
\end{lemma}
\noindent We now proceed with our memoryless lower bound. 
\begin{theorem}\label{thm:memoryless_lb}
    Any algorithm without quantum memory that solves Pauli shadow tomography with high probability requires 
    \begin{equation}
        \Omega\left(\max \left\{\left(\frac{2}{f(\lambda)}\right)^{n}, \frac{3}{(1-\lambda)^{2n}}\right\} \right)
    \end{equation}
    samples.
\end{theorem}
\begin{proof}
As argued previously, we will prove lower bounds on the many-vs.-one distinguishing task from Definition \ref{def:dec_ip_problen}, which will automatically imply a lower bound on our learning task from Definition \ref{def:pauli_shadows}. First, consider Dec-IP$(n, \textnormal{Unif}(\mathcal{P}_n))$. To use Lemma \ref{lemma:toolbox}, we wish to bound the following expression containing the likelihood ratio
\begin{equation}
   \mathbb{E}_{\rho_P\sim \mathcal D}\mathbb{E}_{s\sim p^{\mathds{1}/2^n}(s|u)}\left[\left(\frac{p_{\rho_P}(s|u)}{p_{\mathds{1}/2^n}(s|u)}-1\right)^2\right],
\end{equation}
where $\rho_P = \DN[(\mathds{1}+P)/2^n]$ and we let $\mathcal{D}$ denotes the uniform distribution over $\mathcal{P}$. Letting $F_s^u = 2^nw_s\ketbra{\psi_s^u}{\psi_s^u}$ denote a POVM element and $P_D = \DN[P]$, we rewrite:
\begin{align}
   \mathbb{E}_{\rho_P\sim \mathcal D}\mathbb{E}_{s\sim p^{\mathds{1}/2^n}(s|u)}\left[\left(\frac{p_{\rho_P}(s|u)}{p_{\mathds{1}/2^n}(s|u)}-1\right)^2\right] &= \frac{1}{4^n}\sum_{P\in\mathcal{P}}\mathbb{E}_{s\sim p^{\mathds{1}/2^n}(s|u)}\left[\left(\frac{\tr(F_s^uP_D)}{\tr(F_s^u/2^n)}\right)^2\right] \\
   &= \frac{1}{4^n}\sum_{P\in\mathcal{P}}\sum_s\tr(F_s^u/2^n)\frac{\tr(F_s^{u^{\otimes 2}}(P_D\otimes P_D))}{\tr(F_s^u/2^n)^2} \\
   &= \frac{1}{2^n}\sum_s w_s\frac{\tr(\ketbra{\psi_s^u}{\psi_s^u}^{\otimes 2}(\DN\otimes \DN)[\textsf{SWAP}_n])}{\tr(\ketbra{\psi_s^u}{\psi_s^u})} \\
   &\leq \left(\frac{f(\lambda)}{2}\right)^n.
\end{align}
In the second equality, we use $\tr(A)^2 = \tr(A\otimes A)$, and in the third we use Fact \ref{fact:pauli_to_swap} and linearity of the depolarizing channel. The final inequality follows by Lemma \ref{lemma:depol_on_swap}. Then by Lemma \ref{lemma:toolbox}, any learning tree which succeeds in distinguishing the two hypotheses with high probability requires depth (and thus a number of measurements) at least $\Omega((2/f(\lambda))^n)$. 

Note, however, that in the case when $\lambda = 1$, this lower bound only results in $\Omega(4^n)$ when in reality the two hypotheses are identical and thus impossible to discriminate with any number of measurements. Hence, we prove a second lower bound which is tighter for large $\lambda$. Consider now Dec-IP$(n, \textnormal{Unif}(\{X, Y, Z\}^{\otimes n}))$. Once again, any algorithm solving Pauli shadow tomography to constant precision $<1/2$ can also solve this variant of Dec-IP$(n)$. Using the same argument as above, we have 
\begin{align}
   \mathbb{E}_{\rho_P\sim \mathcal D}\mathbb{E}_{s\sim p^{\mathds{1}/2^n}(s|u)}\left[\left(\frac{p_{\rho_P}(s|u)}{p_{\mathds{1}/2^n}(s|u)}-1\right)^2\right] 
   &= \frac{1}{3^n}\sum_{P\in\{X,Y,Z\}^{\otimes n}}\sum_s\tr(F_s^u/2^n)\frac{\tr(F_s^{u^{\otimes 2}}(P_D\otimes P_D))}{\tr(F_s^u/2^n)^2} \\
   &= \frac{1}{3^n}\tr(\ketbra{\psi_s^u}{\psi_s^u}^{\otimes 2}(\DN\otimes \DN)[(2\textsf{SWAP}_1-\mathds{1}_2^{\otimes 2})^{\otimes n}]) \\
   &\leq \frac{1}{3^n}\left[\frac{3^n+1}{2}(1-\lambda)^{2n}\right] \leq \frac{1}{3}\left(1-\lambda\right)^{2n}.
\end{align}
In the second line we used Fact \ref{fact:pauli_to_swap}, and in the last line we used Lemma \ref{lemma:depol_on_swap_2}. 

Both lower bounds must hold for all $\lambda \in [0, 1]$, so our final lower bound is the maximum of the two. Note that for $\lambda \rightarrow 1$, our sample complexity bound now goes to infinity, while for $\lambda\rightarrow 0$, our bound goes to $2^{n}$, matching the noiseless analysis in \cite{chen2021exponentialseparationslearningquantum}. 
\end{proof}

\subsection{Lower bound with $k$ qubits of memory and constant queries per experiment}
The next model we consider is a $\lambda$-noisy algorithm with access to $k<n$ qubits of quantum memory, but which can only make a fixed number of queries $c$ to the oracle in each experiment before discarding and re-initializing its quantum register. This is analogous to a noisy quantum register which can only interact with the unknown state a constant number of times before projecting and recording a classical outcome, so we say the memory qubits have a limited lifetime $c$. This framework was introduced in \cite{weiyuan_paulis}, and has a learning tree representation characterized as follows.

\begin{definition}[Learning tree with $k$ memory qubits and $c$-query lifetime]
    Any $\lambda$-noisy quantum algorithm with access to a state $\DN[\rho]$ and $k$ qubit memory with a lifetime of $c$ queries can be represented by a learning tree $\mathcal{T}$. At each at each node $u$ of $\mathcal{T}$, the algorithm measures a $cn$-qubit POVM $M = \{F_s\}_u$ from the set $\mathcal{M}^{cn}_{n+k}$ as defined in Appendix \ref{Appendix:basic_quantum_prelim}. 
\end{definition} 

Henceforth, we refer to a learning tree for this model as a $(c, k)$-learning tree. We now show the following lower bound:
\begin{theorem}\label{thm:c,m_lower_bound}
Any quantum algorithm represented by a $(c, k)$-learning tree solving \textnormal{Dec-IP}$(n, \textnormal{Unif}(\mathcal{P}_n))$ with high probability requires a tree of depth at least
    \begin{equation}
        \Omega\left(\min\left(\frac{2^{n}}{f(\lambda)^nc}, \frac{2^{n - k}e^{-2c}}{f(\lambda)^nc^3}\right)\right)
    \end{equation}    
\end{theorem}
\noindent This bound is sensible for $c$ constant. In this regime, our result implies a sample complexity lower bound of $\Omega(2^{n-k}/f(\lambda)^n)$ for \textnormal{Dec-IP}$(n, \textnormal{Unif}(\overline{\mathcal{P}}_n))$ and thus for Pauli shadow tomography (where, as in the memoryless lower bound, $f(\lambda) = 1-\lambda+\lambda^2/2$).

We will proceed with the proof of Theorem \ref{thm:c,m_lower_bound} as in the memoryless case, by bounding the concentration of the likelihood ratio at intermediate nodes of the $(c, k)$-learning tree. We can define the following quantity which characterizes the sample complexity of any single-copy algorithm with bounded quantum memory:
\begin{lemma}
\label{lemma:dc_tree_bound}
    Any noisy quantum algorithm with $n+k$ bits of quantum memory making $c$ calls to $O$ at each node of its learning tree requires $\Omega(c/\Delta_c)$ queries to $O$ to solve $\textnormal{Dec-IP}(n, \textnormal{Unif}(\mathcal{P}_n))$ with probability at least $2/3$, where
    \begin{equation}
        \Delta_{c} = \max_{M\in\mathcal{M}^{cn}_{n+k}}\frac{1}{4^n}\sum_{P\in \mathcal{P}_n}\chi^2_M\left(\rho_P^{\otimes c}\Big\|\rho_{\mathds{1}/2^n}^{\otimes c}\right).
    \end{equation}
\end{lemma}
\begin{proof}
    Let $\pi$ denote $\textnormal{Unif}(\mathcal{P}_n)$. Any algorithm as in the lemma statement for $\textnormal{Dec-IP}(n)$ can be represented as a learning tree $\mathcal{T}$, with some depth $T$. The quantity in Lemma \ref{lemma:toolbox} is
    \begin{align}
        \mathbb{E}_{P\sim \pi}\mathbb{E}_{s\sim p^{\rho_{\mathds{1}/2^n}}(s|u)} \left[\left(\frac{p^{\rho_P}(s|u)}{p^{\rho_{\mathds{1}/2^n}}(s|u)} - 1\right)^2\right] &= \mathbb{E}_{P\sim \pi} \chi^2_M\left(\rho_P^{\otimes c}\Big\|\rho_{\mathds{1}/2^n}^{\otimes c}\right) \\
        &\leq \max_{M\in \mathcal{M}^{cn}_{n+k}}\frac{1}{4^n}\sum_{P\in \mathcal{P}_n }\chi^2_M\left(\rho_P^{\otimes c}\Big\|\rho_{\mathds{1}/2^n}^{\otimes c}\right) \\
        &= \Delta_c\,.
    \end{align}
    The result follows by Lemma \ref{lemma:toolbox}.
\end{proof}
We then find the following bound on $\Delta_c$.
\begin{lemma} \label{lemma:dc_first_bound}
Take a $c\in \mathbb{N}$ and let $\mathcal M$ be the set of $cn$-qubit POVMs for an $n+k$. For any $P\in \mathcal{P}_n$ and $S\subseteq [c]$, let $(\DN[P])^S$ denote the $cn$-qubit operator obtained by applying $\DN[P]$ to all $n$ qubit blocks labeled by $S$ and identity everywhere else. Then
    \begin{equation}
        \Delta_c \leq \left(\sum_{S\subseteq[c]\textnormal{\textbackslash} \varnothing} \sqrt{\max_{M = \{F_s\}\in \mathcal{M}^{cn}_{n+k}} \mathbb{E}_{P\sim \pi} \sum_s \frac{\tr(F_s\DN[P]^{ S})^2}{2^{cn}\tr(F_s)}}
        \right)^2.
    \end{equation}
\end{lemma}
\begin{proof}
Write $\rho_{P}=\rho_{\mathrm{mm}}+\delta_P$ with $\delta_P=\tfrac{1}{2^n}\DN[P]$. By expanding the tensor power, we have 
\begin{equation}
\rho_P^{\otimes c}
=\sum_{S\subseteq[c]}\frac{1}{2^{cn}}(\DN[P])^{S}.    
\end{equation}
Hence
\begin{align}
\frac{p^{\rho_P}(s)}{p^{\rho_{\text{\rm mm}}}(s)} =1+\sum_{\varnothing\neq S\subseteq[c]}\frac{\tr\!\big(F_s\,(\DN[P])^{S}\big)}{\tr(F_s)}.
\end{align}
Applying the definition of $\chi$-squared divergence,
\[
\chi^2_{M}\!\big(\rho_P^{\otimes c}\big\| \rho_{\mathrm{mm}}^{\otimes c}\big)
=\sum_s p^{\rho_{\text{\rm mm}}}(s)\left(\frac{p^{\rho_P}(s)}{p^{\rho_{\text{\rm mm}}}(s)}-1\right)^2
=\sum_{s}\frac{\tr(F_s)}{2^{cn}}
\left(\sum_{\varnothing\neq S\subseteq[c]}\frac{\tr\!\big(F_s\,(\DN[P])^{S}\big)}{\tr(F_s)}\right)^{\!2}.
\]
Applying Cauchy-Schwarz across the sum in $S$, then taking the maximum over all $\mathcal{M}_{n_k}^{cn}$ POVMs and using monotonicity to move the maximum inside the sum: 
\begin{align}
\chi^2_{M}\!\big(\rho_P^{\otimes c}\big\| \rho_{\mathrm{mm}}^{\otimes c}\big) &\le \left(\sum_{\varnothing\neq S\subseteq[c]} \sqrt{\ \sum_s\frac{ \tr\big(F_s\,(\DN[P]^{S})\big)^2}{2^{cn}\tr(F_s)}\ }\right)^{\!2} \\ &\leq \left(\sum_{\varnothing\neq S\subseteq[c]} \sqrt{\max_{M=\{F_s\}\in\mathcal M^{cn}_{n+k}}\ \sum_s\frac{ \tr\big(F_s\,(\DN[P]^{S})\big)^2}{2^{cn}\tr(F_s)}\ }\right)^{\!2}
\end{align}

We can now bound $\Delta_c$ by averaging over $P\sim\pi$ and applying Jensen's inequality to move the expectation under the square root
\begin{align}
\Delta_c \le \left(\sum_{\varnothing\neq S\subseteq[c]} \sqrt{\ \max_{M=\{F_s\}\in\mathcal M_{n+k}^{cn}}\,\mathbb E_{P\sim\pi}\sum_s\frac{ \tr\big(F_s\,(\DN[P]^{S})\big)^2}{2^{cn}\tr(F_s)}\ }\right)^{\!2},
\end{align}
and so we obtain the stated bound.
\end{proof}
To obtain a final bound on $\Delta_c$, we need to bound the expression in Lemma \ref{lemma:dc_first_bound}. We do so in the following Lemma.
\begin{lemma}\label{lemma:dc_moments} For $f(\lambda) = 1-\lambda+\lambda^2/2$,
    \begin{equation}
        \max_{M=\{F_s\}\in\mathcal M_{n+k}^{cn}}\mathbb{E}_{P\sim\pi}\sum_s\frac{ \tr\big(F_s\,(\DN[P]^{S})\big)^2}{2^{cn}\tr(F_s)} \leq
        \begin{cases}
            2^{-n(1-\log f(\lambda))}, \quad &\textnormal{if } |S| = 1 \\
            2^{k-n(1-\log f(\lambda))}, \quad &\textnormal{if } |S| \geq 2
        \end{cases}.
    \end{equation}
\end{lemma}

\begin{proof}
We begin with the $|S| = 1$ case, and consider the maximum over the entire set of POVMs $\mathcal{M}^{cn}$, a strict upper bound on the maximum over $\mathcal M_{n+k}^{cn}$. Fix $j\in[c]$. Then
\begin{align}
\mathbb E_{P\sim\pi_u}\sum_s
\frac{\tr\big(F_s\,(\DN[P])^{\{j\}}\big)^2}{2^{cn}\tr(F_s)}
&=\sum_s\frac{\tr\!\left[(F_s\otimes F_s)\,\mathbb E_{P\sim\pi_u}(\DN\otimes \DN)\!(P^{\{j\}}\otimes P^{\{j\}})\right]}{2^{cn}\tr(F_s)} \\
&\leq  \frac{1}{2^n}\sum_s\frac{\tr\!\left[(F_s\otimes F_s)\,\mathbb (\DN\otimes \DN)\textsf{SWAP}_{j, 2j}\right]}{2^{cn}\tr(F_s)} \\
&\leq  \frac{1}{2^n}\sum_s\frac{\tr(F_s)\tr[\mathbb (\DN\otimes \DN)\textsf{SWAP}_{j, 2j}]}{2^{cn}} \\
&\leq \frac{\tr[\mathbb (\DN\otimes \DN)\textsf{SWAP}_{j, 2j}]}{2^n} \\
&\leq \frac{(1-\lambda+\lambda^2/2)^n}{2^n} \\
&= 2^{-n(1-\log f(\lambda))}\,.
\end{align}
In the second line we have used $\mathbb E_{P\sim\pi_u}(P\otimes P)=2^{-n} \textsf{SWAP}_n$, and in the third line we use $\tr(A B)\le \tr(A)\|B\|$ for $A\succeq 0$. $\textsf{SWAP}_{j, 2j}$ is the operator which swaps the $j$-th block of two $cn$-qubit Hilbert spaces and acts as identity on all other blocks.

When $|S| \geq 2$ we can write $M\in \mathcal{M}_{c, n}^k = \{2^{cn}w_s\ket{L_s}\bra{L_s}\}$ where $\ket{L_s}$ is an $n$-qubit MPS of bond dimension $k$ and $\sum_s w_s = 1$. Let us take
\begin{equation}
    \ket{\psi} = \sum_{i=1}^{2^k} \sqrt{\lambda_i}\ket{\alpha_i}\otimes \ket{\beta_i}
\end{equation}
to be the Schmidt decomposition of a POVM element, where the $\ket{\alpha_i}$ are defined on the first $\ell$ sets of $n$ qubits, and $\ell$ is the smallest element of $S$. Equivalently, we now think of each $n$-qubit block as a qudit of dimension $2^n$. Then $\ket{\alpha_i}$ is defined on the first $\ell$ qubits, and $\ket{\beta_i}$ on the remaining $c-\ell$. Now $P^{\{\ell\}}$ will denote $P^S$ restricted to the first $\ell$ qudits, and $P^{S / \{\ell\}}$ will be the action of $P^S$ on the remaining qudits. Then
\begin{align}
\label{E:forsubstitute1}
   \mathbb E_{P\sim\pi_u}\sum_s &
    \frac{\tr\big(F_s\,(\DN[P])^{S}\big)^2}{2^{cn}\tr(F_s)} = \sum_s \mathbb{E}_{P\sim \pi_u} w_s\bra{L_s}(\DN[P])^{S}\ket{L_s}\\
    &\leq \sum_s \frac{w_s}{4^n}\sum_P \left(\sum_{i,j \in 2^k} \lambda_i\lambda_j\Big|\bra{\alpha_i} \mathcal{D}_\lambda^{\otimes n|\ell|}(P^{\{\ell\}})\ket{\alpha_k}\Big|^2\right)\left(\sum_{i,j \in 2^k} \Big|\bra{\beta_i}\mathcal{D}_\lambda^{\otimes n(c-|\ell|)}(P^{S / \{\ell\}})\ket{\beta_j}\Big|^2\right)
\end{align}
We can upper bound the second sum in the product:
\begin{align}
    \left(\sum_{i,j \in 2^k} \Big|\bra{\beta_i} \mathcal{D}_\lambda^{\otimes c-|\ell|}(P^{S / \{\ell\}})\ket{\beta_i}\Big|^2\right) &\leq 
    \left(\sum_{i \in 2^k}\sum_{j \in 2^{n(c-|\ell|)}} \bra{\beta_i}D_\lambda^{\otimes c-|\ell|}(P^{S / \{\ell\}})^\dagger\ketbra{\beta_j}{\beta_j} \mathcal{D}_\lambda^{\otimes c-|\ell|}(P^{S / \{\ell\}})\ket{\beta_i}\right) 
    \\ &= \left(\sum_{i \in 2^k} |\bra{\beta_i} \mathcal{D}_\lambda^{\otimes c-|\ell|}(P^{S / \{l\}})^\dagger \mathcal{D}_\lambda^{\otimes c-|\ell|}(P^{S / \{\ell\}})\ket{\beta_i}|^2\right) \\
    &\leq 2^k\,,
\end{align}
which follows by extending the basis of the Schmidt decomposition to a full orthonormal basis in the inequality, resolving the resulting identity term, and observing that each term in the last sum is the norm of a normalized quantum state.  Substituting the above into~\eqref{E:forsubstitute1}, we have
\begin{align}
    \sum_s \mathbb{E}_{P\sim \pi_u} w_s\bra{L_s}(\DN[P])^{S}\ket{L_s} &\leq \frac{2^k}{4^n}\sum_P \left(\sum_{i,j \in 2^k} \lambda_i\lambda_j\Big|\bra{\alpha_i}D_\lambda^{\otimes |\ell|}(P^{\{\ell\}})\ket{\alpha_j}\Big|^2\right) \\
    &= \frac{2^k}{4^n}\left(\sum_{i,j \in 2^k} \lambda_i\lambda_j\sum_P \tr(\ket{\alpha_j}\bra{\alpha_i}\mathcal{D}_\lambda^{\otimes n|\ell|}(P^{\{\ell\}}))^2\right) \\
    &= \frac{2^k}{4^n}\left(\sum_{i,j \in 2^k} \lambda_i\lambda_j\sum_P \tr\Big((\ket{\alpha_j}\bra{\alpha_i}\otimes \ket{\alpha_i}\bra{\alpha_j})(\mathcal{D}_\lambda^{\otimes |\ell|} \otimes \mathcal{D}_\lambda^{\otimes n|\ell|})\textsf{SWAP}_{\{\ell\}, \{2 \ell\}}\Big)\right) \\ 
    &\leq \frac{2^k}{2^n}\sum_{i,j \in 2^k} \lambda_i\lambda_j\left(1-\lambda +\frac{1}{2}\lambda^2\right)^n \\
    &= 2^{k - n(1-\log f(\lambda))}\,.
\end{align}
The last two steps follow by the same argument as the $|S| = 1$ case, with the only difference that the $\textsf{SWAP}$ acts on two systems of $n\ell$ qubits; however, $\ell\geq 1$ and $f(\lambda) \leq 1$, so $f(\lambda)^{n\ell} \leq f(\lambda)^n$.
\end{proof}
We are now prepared to prove Theorem \ref{thm:c,m_lower_bound}.
\begin{proof}[Proof of \ref{thm:c,m_lower_bound}]
    Combining Lemmas \ref{lemma:dc_first_bound} and \ref{lemma:dc_moments}, we have
    \begin{align}
    \Delta_c &\leq \left(\sum_{S\subseteq[c]\textnormal{\textbackslash} \varnothing} \sqrt{\max_{M = \{F_s\}\in \mathcal{M}^{cn}_{n+k}} \mathbb{E}_{P\sim \pi} \sum_s \frac{\tr(F_s\DN[P]^{ S})^2}{2^{cn}\tr(F_s)}}
    \right)^2 \\
    &\leq (2^{-n(1-\log f(\lambda))/2}c + (2^c-1-c)2^{[k-n(1-\log f(\lambda))]/2})^2 \\ 
    &\leq 2^{-n(1-\log f(\lambda))}c^2 + c^4e^{2c}2^{k-n(1-\log f(\lambda))}\,.
    \end{align} 
    By Lemma \ref{lemma:dc_tree_bound}, we find that to solve Dec-IP($n, \text{Unif}(\mathcal{P}_n)$) with probability at least $2/3$, a $(c, k)$-learning tree requires depth at least 
    \begin{equation}
        \Omega\left(\min\left(\frac{2^{n}}{f(\lambda)^nc}, \frac{2^{n - k}e^{-2c}}{f(\lambda)^nc^3}\right)\right).
    \end{equation} 
\end{proof}
\subsection{Lower bound with $k$ qubits of memory and unbounded depth}
Finally, we prove an exponential lower bound on a learner which can perform arbitrarily deep $\lambda$-noisy quantum circuits on $k$ qubits of ancillary memory and copies of the unknown $n$-qubit state. This corresponds to Figure \ref{fig:quantum_memory}(b), where the circuit can have arbitrary depth. The learning tree for this model is defined as follows.

\begin{definition}[Learning tree with bounded memory and unbounded depth]
Any $\lambda$-noisy quantum algorithm with access to a state $\DN[\rho]$ and $k$ qubit memory with unbounded quantum lifetime can be represented by the following learning tree $\mathcal{T}$.

\begin{itemize}
    \item Let $\Sigma^\rho(u)$ represent an unnormalized mixed quantum state on $k$ qubits, corresponding to the state of the quantum memory of an algorithm with noisy access to $\rho$ at node $u$. At the root node $r$, $\Sigma^\rho(r)$ is a normalized mixed state.
    \item  At each node $u$ of $\mathcal{T}$, the algorithm measures an arbitrary POVM $M_u= \{F_s^u\}$ on the $n+k$-qubit system $\Sigma^\rho(u)\otimes \DN [\rho]$. Given an outcome $s$ which connects $u$ to a child node $v$, the state of the memory in node $v$ is
    \begin{equation}
        \Sigma^\rho(v) = \tr_{>k}\!\Big(F_s^u(\Sigma^\rho(u)\otimes \rho)\Big) := A_{M_u^s}^\rho(\Sigma^\rho(u))\,,
    \end{equation}
    where $\tr_{>k}$ denotes the partial trace over the ``state" qubits, leaving out the memory qubits.
\end{itemize}
Note that the probability that the algorithm reaches node $u$ is $\tr(\Sigma^\rho(u))$.
\end{definition}
The evident difference between the learning tree in this model and our previous settings is that each node is identified with an unnormalized quantum state rather than a classical bitstring outcome. Thus, the likelihood ratio will be a function of trace distance between memory states rather than a ratio of classical outcome distributions. To lower bound the tree depth for Dec-IP, our strategy will be to bound the total variation between leaf distributions in the two hypotheses by a sum of two terms: one corresponding to a small number of choices of $P$ which contribute large fluctuations, and another corresponding to the vast majority of $P$ which contribute very small fluctuations.
\begin{theorem}
    Any $\lambda$-noisy algorithm with $k$ qubits of quantum memory and unbounded quantum depth solving Pauli shadow tomography with high probability requires at least 
    \begin{equation}
        \Omega(2^{(n-k)/3}(1-\lambda)^{-n/3})
    \end{equation}
    samples.
\end{theorem}
\begin{proof}[Proof]
Consider a learning tree $\mathcal{T}$ for Dec-IP$(n,\textnormal{Unif}(\overline{\mathcal{P}}_n))$ with $k$ qubits of quantum memory and unbounded depth. As before, a bound on the depth of this tree implies a bound on the sample complexity of Pauli shadow tomography. As in previous proofs, let $p^{\rho_p}(\ell)$ and $p^{\mathds{1}/2^n}(\ell)$ denote the distribution over leaves when the algorithm is given access to state $\DN[(\mathds{1} + P)/2^n]$ and $\mathds{1}/2^n$, respectively. The total variation distance between discrete distributions $p, q$ can be written as $\sum_{i:p(i)\geq q(i)} (p(i)-q(i))$. Using this fact and letting $L$ denote the set of leaves where $\mathbb{E}_P[p^{\rho_p}(\ell)] \leq p^{\mathds{1}/2^n}(\ell)$,

\begin{align}
    d_{\text{\rm TV}}(\mathbb{E}_P[p^{\rho_p}(\ell)], p^{\mathds{1}/2^n}(\ell)) &\leq \sum_{\ell\in L} p^{\mathds{1}/2^n}(\ell) - \mathbb{E}_P[p^{\rho_p}(\ell)]  \\
    &\leq \sum_{\ell\in L}\mathbb{E}_P[\min(p^{\mathds{1}/2^n}(\ell), |p^{\mathds{1}/2^n}(\ell) -p^{\rho_p}(\ell)|)]\,,
\end{align}
because for real $a, b$, we have $a-b \leq \min(a, |a-b|)$. From the definition of the learning tree, $$|p^{\mathds{1}/2^n}(\ell) -p^{\rho_p}(\ell)| = \tr(\Sigma^{\mathds{1}/2^n}(\ell) - \Sigma^{\rho_P}(\ell)) \leq \|\Sigma^{\mathds{1}/2^n}(\ell) - \Sigma^{\rho_P}(\ell)\|_{\tr}\,.$$
Substituting, we have 
\begin{equation}\label{eqn:memory_bound_eq_1}
    d_{\text{\rm TV}}(\mathbb{E}_P[p^{\rho_p}(\ell)], p^{\mathds{1}/2^n}(\ell))\leq \sum_{\ell\in L}
    \mathbb{E}_P[\min(p^{\mathds{1}/2^n}(\ell), \|\Sigma^{\mathds{1}/2^n}(\ell) - \Sigma^{\rho_P}(\ell)\|_{\tr})]\,.
\end{equation}
Now we bound the second term in the minimum, which corresponds to the difference between the quantum memory at the end of the algorithm. For this, we utilize the following definition.

\begin{definition}[Good Pauli] We say a Pauli operator $P$ is good for an edge $e_{u, s}$ if 
\begin{equation}
    \|A_{M_s^u}^{\mathds{1}/2^n}(\Sigma^{\mathds{1}/2^n}) - A_{M_s^u}^{\rho_P}(\Sigma^{\mathds{1}/2^n})\|_{\tr}  \leq 
    \frac{(1-\lambda)^{n/3}}{2^{(n-k)/3}2^n}\tr\Big(E^{\otimes 2}(\Sigma^{\mathds{1}/2^n\otimes 2}\otimes\mathds{1}_{2n})E^{\dagger\otimes 2}(\textnormal{\textsf{SWAP}}_{> k}\otimes \mathds{1}_{\text{\rm mem}})\Big)^{1/2}
\end{equation}
where $\mathds{1}_{\text{\rm mem}}$ acts on the two copies of the $k$-qubit memory Hilbert space and the \textnormal{\textsf{SWAP}} operator swaps the two copies of the state Hilbert space. The POVM element $F_s^u\in M_s^u$ has been written in its Cholesky decomposition $E^\dagger E$, omitting sub and superscripts for simplicity.
\end{definition}
\noindent This choice of definition allows us to use Markov's inequality to bound the number of bad Paulis we can have in any root-to-leaf path, which will correspond to large terms in the total variation bound. This is done in the following lemma.
\begin{lemma} \label{lemma:bad_pauli_bound}
    The number of bad Paulis along any root-to-leaf path is at most equal to $T\cdot 4^n2^{-(n-k)/3}(1-\lambda)^{n/3}$.
\end{lemma}
\begin{proof}
We use the form of Markov's inequality
\begin{equation}
    \text{Pr}[{X \geq a}] \leq \frac{\mathbb{E}[|X|^2]}{a^2}\,,
\end{equation}
which follows from the fact that $f(x) = |x|^2$ is a nonnegative and nondecreasing function. This inequality requires that we bound the expectation value $\mathbb{E}[\|A_{M_s^u}^{\mathds{1}/2^n}(\Sigma^{\mathds{1}/2^n}) - A_{M_s^u}^{\rho_P}(\Sigma^{\mathds{1}/2^n}))\|_{\tr}^2]$, which we do as follows. First, note that the unnormalized quantum state corresponding to an edge $e_{u, s}$ can be written as
$$A_M^{\rho}(\Sigma) = \tr_{> k}(E(\Sigma\otimes \mathcal{D}_\lambda^{\otimes n}\rho)E^\dagger).$$ Then
\begin{equation}
    \mathbb{E}_P[\|A_{M_s^u}^{\mathds{1}/2^n}(\Sigma^{\mathds{1}/2^n})-A^{\rho_P}_{M_s^u}(\Sigma^{\mathds{1}/2^n}))\|_{\mathrm{Tr}}^2] \leq 2^k \mathbb{E}_P\Big[\tr\Big(\tr_{> k} (E(\DN[P]/2^n \otimes \Sigma^{\mathds{1}/2^n}))E^\dagger)^2\Big)\Big]\,,
    \label{eqn:memory_node_bound_1}
\end{equation}
which comes from the fact that $\|X\|_{\tr}^2 \leq 2^k \tr(X^2)$. Let $P_D := \DN[P]/2^n$. In the following steps, we drop all subscripts and superscripts $s, u$ for clarity, since we are focused on a specific edge of the tree. Using the swap trick in reverse followed by linearity, we find
\begin{align}
\mathbb{E}_P\Big[\tr\!\Big(\tr_{> k} (E(P_D\otimes \Sigma)E^\dagger)^2\Big)\Big] &= 
\tr\Big(E^{\otimes 2}(\mathbb{E}_P[P_D\otimes \Sigma\otimes P_D\otimes \Sigma])E^{\dagger\otimes 2}(\textsf{SWAP}_{> k}\otimes \mathds{1}_{\text{\rm mem}})\Big)\,,
\end{align}
where the SWAP acts on the two state Hilbert spaces and the $\mathds{1}_{\text{\rm mem}}$ on the two copies of the $k$-qubit memory. Using the property of Paulis that $P_D = (1-\lambda)^{|P|}P$, let us define the 2-qubit operator
$$H_\lambda = 2\lambda(1-\lambda)(I\otimes I) + 2(1-\lambda)^2 \,\textsf{SWAP}$$
Then we have that
$$\mathbb{E}_P[P_D\otimes P_D] = \frac{1}{4^n}\bigotimes_n H_\lambda\,.$$
Substituting this, the above trace becomes
\begin{equation}
    4^{-2n}\tr\Big(E^{\otimes 2}(\Sigma^{\otimes 2}\otimes\bigotimes_n H_\lambda)E^{\dagger\otimes 2}(\textsf{SWAP}_{> k}\otimes \mathds{1}_{\text{\rm mem}})\Big)\,,
\end{equation}
where the $H_\lambda$ terms act on the two copies of the state Hilbert space. Now we can make the following observation. With $\Pi_{\text{sym}}^j = \frac{1}{2}(I_j+\textsf{SWAP}_j)$ and $\Pi_{\text{anti}}^j = \frac{1}{2}(I_j-\textsf{SWAP}_j)$ denoting the symmetric and antisymmetric projectors onto the two copies of the $j$-th qubit, we have
\begin{equation}
    H_\lambda^j = 2(1-\lambda)\,\Pi_{\text{sym}}^j + \eta\,\Pi_{\text{anti}}^j
\end{equation}
where $\eta$ is a constant that we will shortly neglect. Now notice that every factor in the trace we are trying to bound is completely invariant under swapping any qubit of the two copies of the state Hilbert space. Concretely, the tensor 
$$\tau = E^{\otimes 2}(\Sigma^{\otimes 2}\otimes (\cdot))E^{\dagger\otimes 2}(\textsf{SWAP}_{> k}\otimes \mathds{1}_{\text{\rm mem}})\,.$$
is symmetric, and therefore, $\tau = \Pi_{\text{sym}}\tau\Pi_{\text{sym}}$ where $\Pi_{\text{sym}} = \otimes_{i=1}^n \Pi_{\text{sym}}^i$ acting on the state Hilbert spaces. However, $\Pi_{\text{sym}}\Pi_{\text{anti}}^j$ for any $j$ is $0$. Hence, all terms containing a projector onto the antisymmetric subspace vanish while the projector onto the symmetric subspace acts trivially on a symmetric tensor, leaving us with
\begin{equation}
     \mathbb{E}_P\Big[\tr\Big(\tr_{> k} (E(P_D\otimes \Sigma)E^\dagger)^2\Big)\Big] \leq
     2^{-3n}(1-\lambda)^{n}\tr\Big(M^{\otimes 2}(\Sigma^{\otimes 2}\otimes\mathds{1}_{2n})M^{\dagger\otimes 2}(\textsf{SWAP}_{> k}\otimes \mathds{1}_{\text{\rm mem}})\Big)
\end{equation}
Substituting into \eqref{eqn:memory_node_bound_1}, we have
\begin{equation}
    \mathbb{E}_P[\|A_{M_s^u}^{\mathds{1}/2^n}(\Sigma^{\mathds{1}/2^n})-A^{\rho_P}_{M_s^u}(\Sigma^{\mathds{1}/2^n}))\|_{\mathrm{Tr}}^2] \leq \frac{(1-\lambda)^{n}}{2^{n-k}2^{2n}}\tr\Big(E^{\otimes 2}(\Sigma^{\mathds{1}/2^n\otimes 2}\otimes\mathds{1}_{2n})E^{\dagger\otimes 2}(\textsf{SWAP}_{> k}\otimes \mathds{1}_{\text{\rm mem}})\Big)
\end{equation}
Now we can apply Markov to bound the number of bad Paulis on any given edge. With $D$ the uniform distribution over $\mathcal{P}_n$, we have
\begin{equation}
    \text{Pr}_{P\sim D}[P\text{ is bad}] \leq 2^{-(n-k)/3}(1-\lambda)^{n/3}\,.
\end{equation}
Since there are $4^n$ Paulis, the number of bad Paulis for a particular edge is at most $4^n2^{-(n-k)/3}$, and along any given root-to-leaf path, the total number of bad Paulis is at most equal to $T\cdot 4^n2^{-(n-k)/3}(1-\lambda)^{n/3}$, as claimed.
\end{proof}
Equipped with Lemma \ref{lemma:bad_pauli_bound}, we now denote by $P[\ell]$ the set of good Paulis for the root-to-leaf path terminating at $\ell$, and for an intermediate node $u$, we take $P[u]$ to be a superset of these good Paulis that are good during the path up until $u$. Using this notation in tandem with Eq.~\eqref{eqn:memory_bound_eq_1}, we have
\begin{align}
    d_{\text{\rm TV}}(\mathbb{E}_P[p^{\rho_p}(\ell)], p^{\mathds{1}/2^n}(\ell))&\leq \sum_{\ell\in L}
    \mathbb{E}_P[\min(p^{\mathds{1}/2^n}(\ell), \|\Sigma^{\mathds{1}/2^n}(\ell) - \Sigma^{\rho_P}(\ell)\|_{\tr})] \\
    &\leq \sum_{\ell\in L} \left(\text{Pr}(P\notin P[\ell])\cdot  p^{\mathds{1}/2^n}(\ell) + \frac{1}{4^n}\sum_{P\in P[\ell]} \|\Sigma^{\mathds{1}/2^n}(\ell) - \Sigma^{\rho_P}(\ell)\|_{\tr}\right) \\
    &\leq T\cdot 2^{-(n-k)/3}(1-\lambda)^{n/3} + 4^{-n}\sum_{\substack{\ell\in \text{leaf}(\mathcal{T}) \\ P\in P[\ell]}}\|\Sigma^{\mathds{1}/2^n}(\ell) - \Sigma^{\rho_P}(\ell)\|_{\tr}\,.
    \label{eqn:memory_before_induction}
\end{align}
Let us we focus on the second term, having simplified our analysis only to the set of good Paulis. Fix some leaf $\ell$ and let its parent be a node $u$, and focus on one $P\in P[\ell]$. Then by the triangle inequality,
\begin{equation}\label{eqn:memory_two_terms}
    \|\Sigma^{\mathds{1}/2^n}(\ell) - \Sigma^{\rho_P}(\ell)\|_{\tr} \leq \|A_{M_s^u}^{\mathds{1}/2^n}(\Sigma^{\mathds{1}/2^n}(u)) - A_{M_s^u}^{\rho_P}(\Sigma^{\mathds{1}/2^n}(u))\|_{\tr} + \|A_{M_s^u}^{\rho_P}(\Sigma^{\mathds{1}/2^n}(u) - \Sigma^{\rho_P}(u))\|_{\tr}
\end{equation}
Because $P$ is good, the first term is bounded by 
\begin{equation}
    \frac{(1-\lambda)^{n/3}}{2^{(n-k)/3}2^n}\tr\Big(E^{\otimes 2}(\Sigma^{\mathds{1}/2^n\otimes 2}\otimes\mathds{1}_{2n})E^{\dagger\otimes 2}(\textsf{SWAP}_{> k}\otimes \mathds{1}_{\text{\rm mem}})\Big)^{1/2}
\end{equation}
To bound the trace, we note that 
\begin{align}
    \sum_s 2^{-n}\tr\Big(E_s^{\otimes 2}(\Sigma^{\mathds{1}/2^n\otimes 2}\otimes\mathds{1}_{2n})E_s^{\dagger\otimes 2}(\textsf{SWAP}_{> k}\otimes \mathds{1}_{\text{\rm mem}})\Big)^{1/2} &= \sum_s 2^{-n}\tr(\tr_{> k}(E_s( \Sigma^{\mathds{1}/2^n}\otimes \mathds{1}_n)E_s^\dagger))^{1/2} \\&= \tr(\Sigma^{\mathds{1}/2^n})\,,
\end{align}
where we use $\sum_s E_s^\dagger E_s = \mathds{1}_{n+k}$. Computing the expectation over leaves, 
\begin{equation}
    \mathbb{E}_{\ell\in \text{leaf}(\mathcal{T})}\|A_{M_s^u}^{\mathds{1}/2^n}(\Sigma^{\mathds{1}/2^n}(u)) - A_{M_s^u}^{\rho_P}(\Sigma^{\mathds{1}/2^n}(u))\|_{\tr} \leq \frac{(1-\lambda)^{n/3}}{2^{(n-k)/3}2^n}\mathbb{E}_{\ell\in \text{leaf}(\mathcal{T})}\left[\tr(\Sigma^{\mathds{1}/2^n})\right] = \frac{(1-\lambda)^{n/3}}{2^{(n-k)/3}2^n}
\end{equation}
This gives us our bound on the first term in Eq.~\eqref{eqn:memory_two_terms}. For the second term, recall that given a leaf $\ell$ and its parent $u$, $P[\ell]\subseteq P[u]$. Then
\begin{align}
    \sum_{\substack{\text{edge } e_{u,s} \text{ to leaf }\ell\\ \text{ outcome s} \\ P\in P[\ell]}}\|A_{M_s^u}^{\rho_P}(\Sigma^{\mathds{1}/2^n}(u)) - A_{M_s^u}^{\rho_P}(\Sigma^{\rho_P}(u))\|_{\tr} &\leq  \sum_{\substack{\text{ parent } u \text{ of }\ell\\ P\in P[u]}}\sum_s \|A_{M_s^u}^{\rho_P}\big(\Sigma^{\mathds{1}/2^n}(u) -\Sigma^{\rho_P}(u)\big)\|_{\tr} \\
    &\leq \sum_{\substack{\text{ parent } u \text{ of }\ell\\ P\in P[u]}} \|\Sigma^{\mathds{1}/2^n}(u) -\Sigma^{\rho_P}(u)\|_{\tr}
\end{align}
In the second inequality we use the following observation. $A_{M_s^u}^{\rho_P}$ is a quantum circuit that performs a measurement and observes outcome $s$; the sum of $A_{M_s^u}^{\rho_P}$ over outcomes is a quantum channel and thus does not increase the trace norm of the unnormalized input state. Returning to Eq.~\eqref{eqn:memory_two_terms}, and substituting our bounds into the sum from Eq.~\eqref{eqn:memory_before_induction}, we are left with
\begin{equation}
    \sum_{\ell\in \text{leaf}(\mathcal{T}), P\in P[\ell]}\|\Sigma^{\mathds{1}/2^n}(\ell) - \Sigma^{\rho_P}(\ell)\|_{\tr} \leq \frac{(1-\lambda)^{n/3}}{2^{(n-k)/3}2^n} + \sum_{\substack{u \text{ at depth } T-1\\ P\in P[u]}} \|\Sigma^{\mathds{1}/2^n}(u) -\Sigma^{\rho_P}(u)\|_{\tr}\,.
\end{equation}
We see that the right-hand term has become the same expression as the left-hand side, but applied to the $T-1^{st}$ layer of the tree. Inductively using the triangle inequality $T$ times and substituting into Eq.~\eqref{eqn:memory_before_induction}, we are left with 
\begin{equation}
    d_{\text{\rm TV}}(\mathbb{E}_P[p^{\rho_p}(\ell)], p^{\mathds{1}/2^n}(\ell)) \leq T\cdot 2^{-(n-k)/3}(1-\lambda)^{n/3} + T\cdot 2^{-(n-k)/3}(1-\lambda)^{n/3}\,. 
\end{equation}
For $T < \Omega(2^{(n-k)/3}(1-\lambda)^{n/3})$, we find that the total variation distance is $o(1)$. Hence, to achieve a success probability $\geq 2/3$, Lemma \ref{lemma:le_cam} requires that $T \geq \Omega(2^{(n-k)/3}(1-\lambda)^{n/3})$. As $\lambda\rightarrow 1$, our lower bound becomes infinite, which is expected because the I+P state becomes maximally mixed. As $\lambda\rightarrow 0$, we recover the $\Omega(2^{(n-k)/3})$ lower bound from \cite{chen2021exponentialseparationslearningquantum},indicating that our bounds are tight up to constants in the exponents.
\end{proof}
Previous memory-aware lower bounds in \cite{chen2021exponentialseparationslearningquantum} became trivial once $k=n$. Note here that even when $k=n$, the sample complexity lower bound still scales as $\Omega((1-\lambda)^{-n/3})$; hence, no sample efficient two-copy algorithm (including any strategy leveraging Bell measurements) for Pauli shadow tomography, or even the easier many-vs.-one discrimination problem, can exist. 

\subsection{Single-copy noisy strategy}
Here we provide an algorithm using only single copies to solve Pauli shadow tomography in the presence of $\nbqp$ noise. To align with our lower bounds, we consider the setting in which $\epsilon = \delta = 1/3$; the dependence on these parameters in the sample complexity are the standard $\epsilon^{-2}, \log(1/\delta)$. A simple approach is to apply classical shadow tomography \cite{classical_shadow}. We first recall the algorithm.
\begin{algorithm}[h]
  \caption{Classical shadow tomography}
  \label{alg:classical_shadow}

  \DontPrintSemicolon                          
  \SetKwInput{KwInput}{Input}                  
  \SetKwInput{KwOutput}{Output}                
  
  \SetKw{KwAnd}{and}                           
  \SetKw{KwOr}{or}

  \KwInput{Observables $O_1,...,O_M$, failure probability $\delta$, $N = O(\log M\log(1/\delta) \max_i\|O_i\|_{\text{shadow}}/\epsilon^2)$ copies of $n$-qubit state $\rho$, distribution $\mathcal{E}$ over $n$-qubit unitaries }
  \KwOutput{$M$ estimates $\hat{O}_i$ with all $|\tr(O\rho) - \hat{O}_i| \leq \epsilon$ with probability $\geq 1-\delta$ }
  \SetKwFor{RepTimes}{repeat}{times}{end}

  \BlankLine
  Initialize $S, \hat{O} \gets \varnothing$ \\
  \RepTimes{$N$}{
  Sample a $U\sim \mathcal{E}$ \\
  Measure $U^{\dagger}\rho U$ in the computational basis to obtain $n$-bistring $s$\\ 
  With $\mathcal{M}(\rho) = \mathbb{E}_{U\sim \mathcal{E}}[U^{\dagger}\rho U]$, append $\mathcal{M}^{-1}(s)$, computed classically, to $S$.
  }
  \For{$i = 1,2,...,M$}{
    Append $\textsc{MedianOfMeans}(S, O_i, 2\log 2M/\delta)$ to $\hat{O}$.
  }
  \Return{$\hat{O}$}     
\end{algorithm}

The subroutine $\textsc{MedianOfMeans}(S, O, K)$ simply batches the set $S$ into $S/K$ nonoverlapping sets of size $K$, computes $\tr(O\mathcal{M}^{-1}(s))$ for each $\mathcal{M}^{-1}(s)$ stored in each batch, and returns the median of the means of all batches for the estimator $\hat{O}_i$. It is also generally required that $\mathcal{E}$ matches moments of the Haar measure. For details, see \cite{classical_shadow}.

In practice, the ensemble $\mathcal{E}$ should be efficiently sampleable. Moreover, in our model of $\nbqp$ computation, every depth-1 layer in the construction of the random unitary will incur a layer of depolarization. In general, these unitaries must be directly applied to the uncharacterized quantum state, and thus act on physical qubits rather than the codespace of any quantum fault-tolerance scheme. Hence, without relying on substantial ancillary quantum memory, these unitaries cannot be implemented noiselessly. For this analysis, we thus choose particular unitary ensemblse which can be implemented with constant-depth circuits such that the algorithm does not incur extensive errors in system size.

In particular, we implement the simplest choice consisting of a single layer of Haar-random single-qubit rotations. We then comment on an extension to random 1-dimensional brickwork circuits of depth $D$.

\begin{theorem}\label{Thm:classical_shadow_sample}
Algorithm~\ref{alg:classical_shadow}, with $\mathcal{E}$ chosen to be the distribution over $n$-fold tensor products of Haar-random single-qubit unitaries, solves Pauli shadow tomography with $\epsilon=\delta=1/3$ using 
\begin{equation}
O(n3^n(1-\lambda)^{-2n})
\end{equation}
samples and single-copy measurements.
\end{theorem}
\begin{proof}
The classical shadows algorithm has a sample complexity of $O(n\max_i \|P_i\|_{\text{shadow}})$; it only remains to evaluate the shadow norm in Algorithm~\ref{alg:classical_shadow} under our choice of ensemble. 

To simplify notation, we define the noisy scrambling channel $\mathcal{C}_{U,\lambda}[\cdot]=\mathcal D_{\lambda}^{\otimes 2n}[U\mathcal D_{\lambda}^{\otimes 2n}[\cdot]U^\dagger]$. Here, $U$ will be a product of Haar-random single qubit unitaries, and the output of the channel represents the quantum state directly before computational-basis measurement in the classical shadows algorithm. The ideal, noiseless reconstruction channel corresponding to a fixed $U$ is given by $\mathcal{R}_{U}[\cdot]=U\cdot U^\dagger$. Recall that in classical shadow tomography~\cite{cotler2020quantum, classical_shadow}, reconstruction of expectation values from measurements is performed classically, and so we do not incur errors in this step.

With this notation in hand, we state known results on the shadow norm of Pauli observables. ~\cite{hu2023classical} and ~\cite{bu2024classical} demonstrate that when the ensemble $\mathcal{E}$ of quantum channels applied to the given state is invariant under local rotations (which holds trivially for the single-qubit Haar random ensemble and depolarizing channel), the classical shadows channel $\mathcal{M}[\rho]$ as defined in Algorithm \ref{alg:classical_shadow} has the following diagonal representation in the Pauli basis:
\begin{align}
\mathcal{M}[\rho]&=\mathbb{E}_U\sum_s\mathcal{R}_{U}^\dagger[\ketbra{s}{s}] \tr\!\left(\ketbra{s}{s}\mathcal{C}_{U,\lambda}[\rho]\right)\\
&= \frac{1}{2^n}\sum_{P\in \mathcal{P}_{n}}\!\omega(P)\tr(P\rho)P,
\end{align}
where the shadow Pauli weight $\omega(P)$ is defined as 
\begin{align}
\omega(P) := \frac{1}{2^{n}}\sum_s\mathbb{E}_U\left(\tr\!\left(\mathcal{R}_{U}^\dagger[\ketbra{s}{s}]P\right)\tr\left(\mathcal{C}_{U,\lambda}^\dagger[\ketbra{s}{s}]P\right)\right).
\end{align}
where the sum is over all $2n$-bitstrings $s$.

Under this definition, \cite{bu2024classical} shows that $\norm{P}_{\text{shadow}}\leq\omega(P)^{-1}$. Hence, we need only to bound $\omega(P)$ for all Paulis with restricted weight to obtain our sample complexity bound. First, note that the channels $\mathbb{E}_U\mathcal{R}_U[\rho]$ and $\mathbb{E}_U\mathcal{C}_{U, \lambda}[\rho]$ are clearly invariant under conjugation of $\rho$ by any fixed unitary from our product ensemble. Namely, any unitary $U$ such that $U\ket{0}^{\otimes n} = \ket{s}$ is of this form. Hence when $\rho = \ketbra{s}{s}$, we can simply evaluate the channel with input $\ketbra{0}{0}$. With this, we have 
\begin{align}
\omega(P)&=\mathbb{E}_U\left(\tr\!\left(\mathcal{R}_{U}^\dagger[\ketbra{0}{0}]P\right)\tr\!\left(\mathcal{C}_{U,\lambda}^\dagger[\ketbra{0}{0}]P\right)\right)\\
&=\mathbb{E}_U\left(\tr\!\left(U\ketbra{0}{0}U^{\dagger}P\right)\tr\!\left(\ketbra{0}{0}\mathcal{C}_{U,\lambda}[P]\right)\right)\\
&=\mathbb{E}_U\left(\tr\!\left(U\ketbra{0}{0}U^\dagger P\otimes\mathcal{D}_{\lambda}^{\otimes n}[\ketbra{0}{0}]U^\dagger\mathcal{D}_{\lambda}^{\otimes n}[ P]\right)\right),
\end{align}
where in the second line we use the definition of the Hilbert-Schmidt inner product to apply the adjoint channel, and in the final line we apply the hermiticity of the depolarizing channel. Now we use the fact that our distribution over unitaries is product, which gives us 
\begin{equation}
\label{eq:wp_product}
\omega(P)=\prod_{j=1}^{n} \mathbb{E}_{U_j}\!\left[\tr\big(U_j \ketbra{0}{0} U_j^\dagger P_j\big)\tr\big(U_j \mathcal D_\lambda[\ketbra{0}{0}] U_j^\dagger \mathcal D_\lambda[P_j]\big)\right],
\end{equation}
where each $U_j$ is Haar-random on $U(2)$. The remaining traces are easily evaluated using the following observation. The single-qubit state $U_j \ketbra{0}{0} U_j^\dagger$ has a Bloch representation $\tfrac{1}{2}(I+\hat r\cdot\bm{\sigma})$ for some unit vector $\hat r$, where $\bm{\sigma}=(X,Y,Z)$. When $P_j\in \{X, Y, Z\}$, taking the trace $\tr\big(U_j \ketbra{0}{0} U_j^\dagger P_j\big)$ simply picks out the $P_j$-axis component of $\hat r$. When $P_j = I$, the trace is $1$. Then 

\begin{equation}
    \tr\!\big(U_j \ketbra{0}{0} U_j^\dagger P_j\big) = \begin{cases}
        1 &\text{for }P_j = I \\
        \hat r\cdot\hat P_j &\text{for }P_j = X, Y, Z
    \end{cases}.
\end{equation}

\noindent Moreover, since $\mathcal{D}_\lambda[\ketbra{0}{0}] = \tfrac{1}{2}\big(I+(1-\lambda)\hat z\cdot\bm{\sigma}\big)$ and $\mathcal{D}_\lambda[P_j] = (1-\lambda)P$ for $P_j \in \{X, Y, Z\}$, we obtain
\begin{equation}
    \tr\!\big(U_j \mathcal D_\lambda[\ketbra{0}{0}] U_j^\dagger \mathcal D_\lambda[P_j]\big) =  \begin{cases}
        1, \quad P_j = I \\
         (1-\lambda)^2 \hat r\cdot\hat P_j, \quad P_j = X, Y, Z
    \end{cases}
\end{equation}

Since each $U_j$ is Haar-random, the corresponding vectors $\hat r$ are uniform over the unit sphere; then isotropy gives us $\mathbb{E}_R[(R\cdot\hat z)_k^2]=1/3$. So every single-qubit term in Eq.~\eqref{eq:wp_product} for which $P_j\neq I$ contributes a factor of $(1-\lambda)^2/3$, while the remaining terms are $1$. Because $\norm{P}_{\text{shadow}}\leq\omega(P)^{-1}$, our sample complexity bound is obtained by upper-bounding the smallest value $\omega(P)$ can take on, which occurs when $P$ has the maximum allowed weight of $n$. Hence $\omega(P)=\left(\frac{(1-\lambda)^2}{3}\right)^{n}$ implies that
\begin{equation}
O(n\max_i \|P_i\|_{\text{shadow}}) = O(n3^n(1-\lambda)^{-2n})\,,
\end{equation}
as claimed.
\end{proof}

We remark that this result can be generalized slightly in the case where the classical shadows algorithm is implemented by a 1-dimensional brickwork circuit of depth $D$. Ref.~\cite{hu2025demonstration} shows that for such circuits (layered with interstitial depolarizing noise of strength $\lambda$),
\begin{equation}
    \|P\|_{\text{shadow}} \leq 3^{(n+D)\left(\frac43+\frac{(4/5)^D}{D^{3/2}}+\frac{D\lambda}{\log 3}\right)}.
\end{equation}
We immediately obtain that classical shadows with circuits of this form solves Pauli shadow tomography using 
$$O\!\left(n\,3^{(n+D)\left(\frac43+\frac{(4/5)^D}{D^{3/2}}+\frac{D\lambda}{\log 3}\right)}\right)$$
samples.

\subsection{Two-copy noisy Bell sampling and quantum-enhanced advantage}
Note that our lower bound for the model of $n$ qubits of quantum memory has a noise dependence that scales as $(1-\lambda)^{-n}$, up to a constant in the exponent. Hence, all two-copy strategies for Pauli shadow tomography are irrecoverably degraded by just the single layer of noise considered in our lower bounds, including all strategies which utilize Bell measurement. In this section, we give a Bell-measurement based algorithm for Dec-IP which, up to constants in the exponent and a nonleading polynomial factor in $n$, matches this exponential scaling in $\lambda$. 
\begin{theorem} \label{thm:qa_upper_bound}
There exists a $\lambda$-noisy quantum algorithm $Q^O_\lambda$ with the ability to perform joint measurements on at most $2$ copies of the $2n$-qubit unknown state and $\lambda = \Theta(1)$ which solves \textnormal{Dec-IP(}$n, \textnormal{Unif(}\overline{\mathcal{P}_n}\textnormal{))}$ using $O(n(1-\lambda)^{-4n})$ queries to $O_P$.
\end{theorem}
\begin{proof}[Proof of Theorem \ref{thm:qa_upper_bound}]
We proceed by giving an explicit $\lambda$-noisy quantum algorithm for Dec-IP($n, \textnormal{Unif}(\overline{\mathcal{P}_n})$). Let $H_0$ denote the ``null hypothesis" in which the oracle prepares copies of the maximally mixed state, and let $H_1$ denote the uniform-over-$\overline{\mathcal{P_n}}$ hypothesis.
\begin{algorithm}
  \caption{Distinguishing an unknown Pauli with noise}
  \label{alg:find_pauli}
    \DontPrintSemicolon
  \SetKwInput{KwInput}{Input}
  \SetKwInput{KwOutput}{Output}
  \SetKw{KwAnd}{and}
  \SetKwFor{RepTimes}{repeat}{times}{end}

  \KwInput{System size $n$,
           $T = \Theta(n(1-\lambda)^{-4n})$
           copies of $\rho$ from oracle $O$}
  \KwOutput{Guess of $H_0$ or $H_1$}
  \BlankLine

  Initialize set $S \gets \varnothing$ \\
\tcp{Quantum data collection}

  \RepTimes{$T$}{
    Append $\textsc{BellMeasure}(2n,\rho\otimes \rho)$ to $S$, using up two copies of $\rho$
  }
  \BlankLine

  \tcp{Classical postprocessing}
  Initialize $Z_{\max}\gets 0$\\
  \For{all $Q\in\overline{\mathcal{P}_n}$}{
  With  $q\in\F_2^{2n}$ denoting the symplectic bitstring of $Q$, define
  \begin{equation*}
    \hat{Z}_Q \gets \frac{1}{T}
      \sum_{s\in S} (-1)^{\langle s, q\rangle}\,.
  \end{equation*}
    $Z_{\max} \gets \max(Z_{max}, |\hat{Z}_{Q^\star}|)$. \\[2mm]
  }
  \Return{ $H_1$ if $Z_{\textnormal{max}} \geq \frac{1}{2}(1-\lambda)^{2n}$, else return $H_0$}
\end{algorithm}

We now argue that Algorithm \ref{alg:find_pauli} gives us Theorem \ref{thm:qa_upper_bound}. For intuition, we begin by analyzing the noiseless case, $\lambda = 0$, under $H_1$. Within the quantum data collection loop, our algorithm performs multiple rounds of Bell measurement (Definition~\ref{def:BM_subroutine}). Given two copies of an $n$-qubit state $\rho = (\mathds{1} + P)/\tr(\mathds{1}+P)$, the distribution over classical $2n$-bitstring outcomes obtained from performing Bell measurement on $\rho\otimes \rho$ is
\begin{align}
    \text{Pr}[s] &= \tr(\Pi_s \rho\otimes\rho) \\
    &= \frac{1}{4^{2n}}\sum_{Q\in\mathcal{P}}(-1)^{\langle s, q\rangle + \langle q\rangle} \tr((Q \otimes Q)((\mathds{1}+P)\otimes (\mathds{1}+P))\\
    &= \frac{1}{4^{2n}}\sum_{Q\in\mathcal{P}}(-1)^{\langle s, q\rangle + \langle q\rangle} (\tr(Q) +\tr(QP))^2\\
    &= \frac{1}{4^{n}}(\tr(\mathds{1}P) + (-1)^{\langle s, p\rangle}\tr(P^2)) \\
    &= \frac{1}{4^{n}}(1 + (-1)^{\langle s, p\rangle})\,,
\end{align}
which is $0$ if $\langle s, p\rangle = 1$ and $1/4^n$ otherwise. Here, lowercase $p, q$ are the bitstrings over $\mathbb{F}_2$ associated with Pauli observables $P, Q$ as defined in Section \ref{sec:symplectic_paulis}. In other words, each Bell measurement is equivalent to sampling uniformly from the subspace of $\mathbb{F}_2^{2n}$ orthogonal to $p$. Under $H_0$, it is easy to see that every outcome is simply a uniformly random bitstring from all $4^n$ candidates.

Now we analyze Bell measurement under $H_0$ in the noisy setting. Here, we jointly measure two copies of $\DN[\rho]$, where $\lambda = \Theta(1)$ is some constant $>0$. For any $Q\in \overline{\mathcal{P}_n}$, we have $\tr(Q\DN[P]) = (1-\lambda)^{|P|}\tr(PQ)$. Substituting this into the noiseless expression, the distribution over noisy outcomes becomes
\begin{equation}
\text{Pr}_\lambda[s] = \frac{1}{4^n}(1 + (1-\lambda)^{2|P|}(-1)^{\langle s, p\rangle})\,.
\end{equation}
Meanwhile, the distribution over measurement outcomes under $H_0$ is still uniform over all bitstrings, because the maximally mixed state remains unchanged after the noise channel.

With this observation, we can define a test statistic that distinguishes the hypothesis as follows. First, we define the random variable
\begin{equation}
    X_Q = (-1)^{\langle s, q\rangle}\,,
\end{equation}
corresponding to every non-identity $n$-qubit Pauli $Q$ with symplectic bitstring $q$. The randomness is over Bell measurement outcomes $s$. Under $H_0$, note that $\mathbb{E}[X_Q|H_0] = 0$, because all $4^n$ values of $s$ occur with equal probability, and any nonzero symplectic bitstring is orthogonal to exactly half of them. Under $H_1$, we can condition on each value of $P$:
\begin{align}
  \mathbb{E}[X_Q | P]
  &= \sum_{s\in\mathbb F_2^{2n}} X_Q\textnormal{Pr}[s| P] \\
  &= \frac{1}{4^n}\sum_{s}(-1)^{\langle s, q\rangle}
      + \frac{(1-\lambda)^{2|P|}}{4^n}\sum_{s}(-1)^{\langle s, q+p\rangle} \\
  &= \begin{cases}
      (1-\lambda)^{2|P|}\quad &\textnormal{if } q=p \\
      0\qquad &\textnormal{else}
  \end{cases}\,.
\end{align}

Classically enumerating over all non-identity Paulis,w then define our estimator $\hat{Z}_Q$ to be the empirical mean of $X_Q$. In expectation, only $\mathbb{E}[\hat{Z}_P]$ is greater than $0$, so our algorithm chooses the largest empirical $\hat{Z}_Q$ as the final estimator. We note that, under success, our algorithm has also identified $P$ correctly, and is thus \textit{stronger} than simply solving Dec-IP. It is plausible that with an appropriate reframing, e.g. fixing a Pauli $P$ in the alternate hypothesis as in Ref. \cite{huang2022foundations}, we may obtain a tighter gap between our lower and upper bounds.

We now analyze failure probability, proving the correctness of our algorithm. Suppose the oracle outputs the maximally mixed state; i.e. the correct solution is $H_0$. By Hoeffding's inequality, for any fixed $Q$ and any
$\tau>0$,
\begin{equation}
  \textnormal{Pr}\big[|\hat{Z}_Q| \geq \tau | H_0\big]\leq 2\exp(-2T\tau^2)\,.
\end{equation}
Applying a union bound over the $4^n-1$ possible non-identity Paulis
gives
\begin{equation}
  \textnormal{Pr}\big[|Z_{\max}| \ge \tau \mid H_0\big]\leq (4^n-1)\cdot 2\exp(-2T\tau^2)\,.
  \label{eq:type-I-bound}
\end{equation}
Taking $\tau = \frac{1}{2}(1-\lambda)^{2n}$ and using that $(1-\lambda)^{2|P|} \geq (1-\lambda)^{2n}$, we find that for
\begin{equation}
    T\geq C_1n(1-\lambda)^{-4n}
\end{equation}
and some absolute constant $C_1>0$, the failure probability is bounded by $1/3$. Under the ground truth $H_1$, note that bounding the probability of the event $|\hat{Z}_P| \leq \tau$ is sufficient, because $|\hat{Z}_P| \leq Z_{\textnormal{max}}$. We have
\begin{equation}
    \mu_P\coloneqq\mathbb{E}[\hat{Z}_P | H_1] = (1-\lambda)^{2|P|} \geq 2\tau \,.
\end{equation}
Applying Hoeffding's inequality once again, 
\begin{align}
  \textnormal{Pr}\big[Z_{\max} < \tau | P\big]
  &\leq \textnormal{Pr}\big[|\hat{Z}_P - \mu_P|\geq \mu_P - \tau \big| P\big] \\
  &\le 2\exp\!\left(-2T(\mu_P - \tau)^2\right) \\
  &\leq 2\exp\!\left(-\tfrac{T}{2}(1-\lambda)^{4n}\right)\,.
\end{align}
We find that for
\begin{equation}
    T\geq C_2n(1-\lambda)^{-4n}\geq C_2(1-\lambda)^{-4n}
\end{equation}
and some absolute constant $C_2>0$, the failure probability is bounded by $1/3$. This concludes the proof.
\end{proof}
We now immediately obtain Corollary \ref{cor:1}. Let $N_{TC}$ denote the two-copy sample complexity we have derived, and let $N_{SC} = \Omega\big((2/f(\lambda))^n\big)$ denote our lower bound on the sample complexity of any single-copy memoryless strategy from Theorem \ref{thm:memoryless_lb}. Then we find $N_{SC} = \Omega(N_{TC}^{a(\lambda)})$, where 
\begin{equation}
    a(\lambda) = \frac{n(\log 2/f(\lambda))}{\log n - 4n\log(1-\lambda)} = O(\lambda^{-1})\,.
\end{equation}
This establishes a polynomial separation between traditional and quantum-enhanced strategies in the presence of noise, which becomes exponential as $\lambda$ tends to $0$.

Note that for any Pauli observable $P$ whose weight is known beforehand, the sample complexity of our two-copy algorithm actually scales parametrically as $n(1-\lambda)^{-|P|}$; only our worst-case analysis takes $|P| = n$. In reality, high-weight Pauli terms are quantitatively the culprit in degrading the algorithm’s performance from the ideal setting: if instead we chose only to input Paulis with weight at most $O(\log n)$, the two-copy algorithm would incur only polynomial sample cost in $n$. The same, however, is true for single-copy strategies: the classical-shadow protocol of Theorem \ref{Thm:classical_shadow_sample} also becomes polynomial in $n$ when restricted to $O(\log n)$-weight observables. Consequently, high-weight observables are precisely what enable exponential quantum advantages in the ideal setting, and they are also most severely suppressed by noise in natural quantum learning, where the advantage they generate collapses to at best a polynomial separation.

Although local noise exponentially degrades Bell sampling, polynomial quantum-classical separations can persist and still yield substantial experimental gains~\cite{Chen_2024_pauli,seif2024,Oh_2024,Liu_2025}. For near-term experiments at fixed instance size $n$, the more relevant question is how small the noise rate $\lambda$ must be to achieve a desired advantage. Our Dec-IP analysis assumes $\lambda=\Theta(1)$, but if $\lambda=\Theta(1/n)$, the two-copy Bell-sampling strategy regains an $\exp(\Omega(n))$ separation over any single-copy protocol. Allowing $\lambda$ to scale with $n$ is not an asymptotic advantage statement, but it yields a concrete finite-size criterion: when $\lambda \ll 1/n$, the sample-complexity gap can grow essentially exponentially with $n$. Bounds like Corollary~\ref{cor:1} make this dependence explicit, translating target performance into quantitative noise requirements.

\section{A Threshold for Heisenberg-Limited Metrology}

In this Section, we demonstrate that the presence of interstitial noise in quantum-enhanced experiments has fundamental implications for quantum metrology. Here, we work with standard definitions for single-parameter quantum sensing.
\begin{definition}[Single-parameter, single-probe quantum sensing task] Consider an $n$-qubit Hamiltonian $H = \omega G$, where $\omega$ is a real-valued scalar, and a set of Lindblad jump operators $\{L_k\}$ of size $m$. Moreover, fix a precision $\epsilon>0$. Given a single quantum probe $\rho_P$ on $n$ qubits, evolving according to the time-homogeneous Lindblad master equation
\begin{equation}
    \frac{d\rho}{dt} = -i[H, \rho] + \sum_{j=1}^m L_j\rho L_j^\dagger  -\frac{1}{2}\{L_j^\dagger L_j, \rho\} \ ,
\end{equation}
the task of single-parameter quantum sensing is to estimate $\omega$ to within absolute precision $\epsilon$. Ancillary quantum memory and quantum information processing may be used.
\end{definition}
\noindent Strategies that do not require coherent quantum information processing can perform single-parameter sensing using a total sensor evolution time $t = O(\epsilon^{-2})$, known as the Standard Quantum Limit (SQL). However, the gold standard is to achieve the so-called Heisenberg limit (HL) of $t = O(\epsilon^{-1})$. We note that an HL sensing protocol using total time $t$ and a single probe can often be reinterpreted as an $N$-probe protocol requiring sensing time $t/N$ per probe, whereas for SQL protocols, using $N$ parallel sensors will require $t/\sqrt{N}$ sensing time. 

A well-known result in metrology, due to \cite{Zhou_2018}, proposes a necessary and sufficient condition for the achievability of HL metrology. Namely, they define the Hamiltonian-Not-in-Lindblad-Span (HNLS) criterion, which requires that $H$ not be contained in the linear span of $I, L_j, L_j^\dagger, L_j^\dagger L_k$ for all Lindblad jumps and their pairwise products. If HNLS holds, Ref.~\cite{Zhou_2018} demonstrates the existence of a quantum error-correcting code which can correct for jumps while preserving logical time evolution under the signal Hamiltonian (when executed using fast quantum control and between short-time queries to the signal). Conversely, they prove that when HNLS is violated, there is no protocol which can outperform the SQL. 

At a high level, the error-correction scheme in \cite{Zhou_2018} is enabled by the assumption that quantum error-correction is effectively instantaneous, so that the dissipator, which only evolves for short times, is unlikely to generate high-weight errors consisting of multiple jump terms. In this regime, error-correction handles first-order-in-time contributions from the dissipator, while higher weight errors occur with exponentially low probability. In light of our results, we can make several practical refinements. First, real-world hardware can experience noise at the interface between sensor and computing degrees of freedom, whenever the protocol alternates between signal queries and error-correction rounds. Moreover, quantum control takes finite time, and any noise that is native to the quantum computer, which may not be captured in the signal Lindbladian, can still accrue. The strength of this interface-level noise channel is thus parametrized both by the native noise rate of the quantum computer and the time each round of quantum control takes. As a result, the continuous-time approximations made in \cite{Zhou_2018}, where noise and error correction occur as simultaneous processes, require refinement in a more practical setting. 

To make this precise in an instructive but simple way, we assume each round of quantum control occupies a time at least $\tau$. Moreover, in line with $\nbqp$, we assume that the ``interface" noise channel which appears between signal queries and quantum control is depolarizing. To align with the Lindbladian formalism, we assume this noise is generated by the channel $e^{\mathcal{L}_{\textnormal{dep}}^\eta\tau}$, where 
\begin{equation}
    \mathcal{L}_{\textnormal{dep}}^\eta(\rho) = \eta\sum_{j=1}^n \sum_{\sigma\in\{X, Y, Z\}}(\sigma^{(j)}\rho \sigma^{(j)} - \rho)
\end{equation}
is the Lindbladian which generates a local depolarizing channel. Here, we parametrize the depolarization strength by $\eta$ and $\tau$, noting that as quantum channels, $\exp(\mathcal{L}_{\textnormal{dep}}^\eta \tau)$ and $\DN$ are equivalent when $\lambda = 1 - \exp(-4\eta \tau)$. We then work with an effective description of interstitial noise, in which every round of quantum control is assumed to be a noiseless black box interleaved with $\exp(\mathcal{L}_{\textnormal{dep}}^\eta \tau)$. 

If at every round of quantum control in the procedure from \cite{Zhou_2018}, we introduce this interface noise and assume HNLS is satisfied, how does the protocol's HL sensitivity change?

\begin{theorem} \label{thm:HL_sensing}
Consider a single-parameter sensing task satisfying the HNLS criterion, with interstitial depolarizing noise applied after each round of quantum control. There exist thresholds $\tilde{\epsilon}(\eta, \tau, n)$ and $\tilde{t}(\eta, \tau, n)$ such that the protocol of Ref.~\cite{Zhou_2018} achieves Heisenberg-limited sensitivity for target precision $\epsilon > \tilde{\epsilon}$ and total evolution time $t < \tilde{t}$, and crosses over to the Standard Quantum Limit otherwise. For $n\eta\tau \ll 1$, the thresholds scale as $\tilde{t} = O(1/\eta n))$ and $\tilde{\epsilon} = O(\eta n)$. 
\end{theorem}
\noindent Note that because our bounds treat the quantum control as a black box, they extend to any protocol achieving Heisenberg-limited sensing in the absence of interface noise if the protocol uses many interleaved rounds of signal queries and quantum control, only making measurements at the end. This result places a fundamental ceiling on the achievability of the Heisenberg limit in the presence of experimental noise, absent other problem-dependent error-mitigation strategies. Without leveraging such strategies to correct for both signal noise \textit{and} interface noise, HL sensitivity is preserved up to only inverse-polynomial total time in the number of qubits.

Our proof will elucidate the hardware dependence of HL sensitivity to the physical noise rate and the speed of error correction. This perspective is crucial to quantify precision limits and requisite experimental parameters for practical error-correction-based metrology protocols. 

To proceed with the proof, we require some notation. In the language of \cite{Zhou_2018}, each round of sensing and error-correction (without interface-level noise) is described by the quantum channel $\Psi(\rho) = \mathcal{R}\circ e^{\mathcal{L}t_{single}}$, where $t_1$ is the time for a single round and $\mathcal{R}$ is a recovery channel which projects the joint state of probe and ancillas onto a chosen codespace which protects the Hamiltonian against Lindblad jumps. This high-level channel description will suffice for our argument, but the full construction of $\mathcal{R}$ is detailed in \cite{Zhou_2018}. Upon incorporating $\nbqp$-style noise, the channel becomes $\Phi(\rho) = \mathcal{R}\circ e^{\mathcal{L}_{dep}^\eta\tau} \circ e^{\mathcal{L}t_{1}}$.  The entire protocol consists of $N$ rounds and runs for time $t$; trivially, $N\leq t/\tau$. With these definitions, we can prove a lemma which will immediately yield Theorem \ref{thm:HL_sensing}.

\begin{lemma}
    Consider a single-parameter sensing problem which satisfies the HNLS criterion. Then for any initial state $\rho$,
    \begin{equation}
        \| \Psi^N(\rho)-\Phi^N(\rho)\|_1 \leq 2pN
    \end{equation}
    where $p = 1-(1-\lambda)^n$ and $\lambda = 1-e^{-4\eta\tau}$. When $\eta\tau \ll 1$, this implies
    \begin{equation}
        \| \Psi^N(\rho)-\Phi^N(\rho)\|_1 \leq 6\eta n\tau N \ .
    \end{equation}
\end{lemma}
\begin{proof}
    Via telescoping,
    \begin{equation}
        \Psi^N(\rho)-\Phi^N(\rho) = \sum_{i=0}^{N-1} \Psi^{N-i-1}\circ[(\Psi-\Phi)\circ \Phi^i(\rho)] \ .
    \end{equation}
    Applying the triangle inequality and using the fact that $\Psi, \Phi$ are trace-preserving,
    \begin{equation}
        \| \Psi^N(\rho)-\Phi^N(\rho)\|_1\leq \sum_{i=0}^{N-1} \|(\Psi-\Phi)(\Phi^i(\rho))\|_1 \leq N\sup_\sigma \|(\Psi-\Phi)(\sigma)\|_1
    \end{equation}
    where in the last line we absorb $\Phi^i$ into the state and maximize over all states. The remaining task is to bound the difference in channel output states after a single application of the channel. Note that by assumption, $\mathcal{R}$ will correct, in each round, any errors due to the signal dissipator. Only if errors due to the interface depolarizing channel occur will there be any difference between $\Psi$ and $\Phi$. We can then rewrite $\Phi = (1-p)\Psi + p\chi(\rho)$, where $\chi(\rho)$ is \textit{some} quantum channel not equal to $\Psi$ and $p$ is the probability of any error due to this noise source, which we will calculate shortly. Then for any state $\sigma$,
    \begin{equation}
        (\Psi-\Phi)(\sigma) = p(\Psi - \chi)(\sigma)
    \end{equation}
    and taking trace norm and using that density matrices have trace 1,
    \begin{equation}
        \sup_\sigma \|(\Psi-\Phi)(\sigma)\| \leq p \sup_\sigma \|(\Psi - \chi)(\sigma)\| \leq 2p
    \end{equation}
    Precisely, $p$ is the probability that evolution by $\mathcal{L}_{dep}^\eta$ for time $\tau$ causes a jump of weight 1 or larger. While it is possible that the signal Lindbladian has terms which overlap with the depolarizing Lindbladian, this single-jump probability is still a valid upper bound on the channel distinguishability. Moreover, because $e^{\mathcal{L}_{\textnormal{dep}^\eta}\tau} = \DN$ for $\lambda = 1-e^{-4\eta\tau}$, the probability of a weight $1$ or higher error is exactly $p = 1-(1-\lambda)^n$.
\end{proof}
This trace-distance condition, as argued in \cite{Zhou_2018}, implies that the Quantum Fisher Information of the protocol increases quadratically in time only up to $\tilde{t}\propto \tau/p$, before crossing over to an asymptotic SQL sensitivity. Equivalently, attempting to resolve $\omega$ to precision $\epsilon$ is only achievable with Heisenberg-limited sensitivity for $\epsilon \ll 1/ \tilde{t}$. For $\eta n\tau \ll  1$, these thresholds scale as $\tilde{t} = O(1/\eta n)$ and $\tilde{\epsilon} = O(\eta n)$. This gives us Theorem \ref{thm:HL_sensing}.

\section{Deferred Proofs} \label{app:deferred}
\subsection{Proof of Lemma \ref{lemma:o2h}}
To prove this lemma, we want to bound the probability that any algorithm querying the encoded SSP oracle performs queries inside the hidden wrappers, then argue that it is unlikely to notice if we swap out the oracle on those wrappers with one implementing a shadow function. We start by defining an encoded unitary that flags whether a computation contains a query to some element in $S_{d-1}$, the most crucial hidden domain.

For this definition, we adapt notation from \cite{arora2023quantum}. Let the input register \textbf{I} hold the query inputs to $\mathcal{F}^{\textnormal{Enc}}$ and let the output responses be applied to register \textbf{O}. The remaining quantum memory will be denoted by \textbf{W}. For simplicity, we assume these are all logical registers embedded into the FT-QEC scheme associated with the problem, and implicitly assume that all oracle calls in the following definition are encoded, dropping the Enc superscript.
\begin{definition}[Encoded flag unitary]
    Suppose we have some unitary $U$ acting jointly on \textnormal{\textbf{IOW}}. Let $\mathcal{F}$ be an encoded oracle acting logically on \textbf{IO}, and let $S$ be a subset of the classical query domain of $\mathcal{F}$. We define the encoded flag unitary $$U^{\mathcal{F}/S}\ket{\psi}_{\textnormal{\textbf{IOW}}}\ket{0}_{\textnormal{FLAG}} = \mathcal{F}U_SU\ket{\psi}_{\textnormal{\textbf{IOW}}}\ket{0}_{\textnormal{FLAG}}$$
    where \textnormal{FLAG} denotes a single-qubit register, and 
    \begin{equation}
        U_S\ket{\zv}_\textbf{I}\ket{b}_\textnormal{FLAG} = \begin{cases} \ket{\zv}_\textbf{I}\ket{b}_\textnormal{FLAG} & \text{\rm if } \zv \cap S = \varnothing \\
    \ket{\zv}_\textbf{I}\ket{b\oplus 1}_\textnormal{FLAG} & \text{\rm else }
    \end{cases},
    \end{equation}
    where $\ket{\zv}$ is a query input state that contains a combination of encoded computational basis states labeled by bitstrings $z\in \zv$, and $\zv$ is the set of those bitstrings.
\end{definition}
The encoded flag unitary simply takes a query register which we prepare and flips a single flag qubit if any part of the query overlaps with the set $S$. Naturally, this will be used with $S$ being a hidden wrapper or hidden domain. We now proceed with the proof.
\begin{proof}[Proof of Lemma \ref{lemma:o2h}]
Recall that $\mathcal{F}$ is a $d$-level shuffling Simon's function with hidden domains $S_1,...,S_d$, and $F^{\text{Enc}}$ is its corresponding encoded quantum oracle acting on the logical codespace of a FT-QEC scheme. Moreover, we fix collections of $k$-level hidden wrappers $\bar{S}^{(k)}$ for $k = 1,...,d$. Using an encoded flag unitary corresponding to any particular $\bar{S}^{(k)}$, note that we can always find a decomposition of the following form (with unnormalized states): 
\begin{equation}
U^{\mathcal{F}^{\text{Enc}}/S}\ket{\psi}_{\textnormal{\textbf{IOW}}}\ket{0}_{\textnormal{FLAG}} = \ket{\psi_{\text{in}}}_{\textnormal{\textbf{IOW}}}\ket{1}_{\textnormal{FLAG}} + \ket{\psi_{\text{out}}}_{\textnormal{\textbf{IOW}}}\ket{0}_{\textnormal{FLAG}}    \label{eq:encoded_orthog}
\end{equation}
where $\ket{\psi_{\text{in}}}$ (respectively $\ket{\psi_{\text{out}}})$ is a combination of basis strings inside (outside) $S$, and  $\braket{\psi_{\text{in}}|\psi_{\text{out}}} = 0$. By the same reasoning, 
\begin{equation}
U^{\mathcal{F}_{\text{\rm sh}}^{(k), \text{Enc}}/S}\ket{\psi}_{\textnormal{\textbf{IOW}}}\ket{0}_{\textnormal{FLAG}} = \ket{\perp}_{\textnormal{\textbf{IOW}}}\ket{1}_{\textnormal{FLAG}} + \ket{\psi_{\text{out}}}_{\textnormal{\textbf{IOW}}}\ket{0}_{\textnormal{FLAG}}\,,
\end{equation}
where $\braket{\psi_{\text{in}}|\perp} = 0$ because, by definition, $\ket{\perp}$ corresponds to a closed linear subspace that is only mapped to by the action of the shadow oracle within the hidden domain.

From these expressions, the state of our quantum register $\textnormal{\textbf{IOW}}$, which starts in a fixed state $\ket{\psi}$, is acted upon by an arbitrary unitary $U$ uncorrelated with $\mathcal{F}$, and is then hit by either the encoded oracle or its shadow, is:
\begin{align}
    &\ket{\psi_{\text{true}}}:= \mathcal{F}^{\text{Enc}}U\ket{\psi} = \ket{\psi_{\text{in}}} + \ket{\psi_{\text{out}}}\\
      & \ket{\psi_{\text{\rm sh}}}:= \mathcal{F}_{\text{\rm sh}}^{(k), \text{Enc}}U\ket{\psi} = \ket{\perp} + \ket{\psi_{\text{out}}}\,.
\end{align}
Both of these states depend on the choice of shuffling $\mathcal{F}$. Averaging over shufflings sampled from $D(f, d)$, we have states $\rho_{\text{true}} = \mathbb{E}_{\mathcal{F}}[\ketbra{\psi_{\text{true}}}{\psi_{\text{true}}}], \rho_{\text{\rm sh}} = \mathbb{E}_{\mathcal{F}}[\ketbra{\psi_{\text{\rm sh}}}{\psi_{\text{\rm sh}}}]$. These represent the state of our quantum register after a single layer of unitary operations and oracle queries. We want to bound their discrepancy to demonstrate that the shadow oracle acts almost exactly like the true oracle, up to the probability of finding the hidden wrapper via classical bitstring sampling.

To do this, we make use of the Fuchs-van de Graaf inequality, $d_{\tr}(\rho_1, \rho_2) \leq \sqrt{2-2F(\rho_1, \rho_2)}$, where $F$ denotes the fidelity. The fidelity between $\rho_{\text{true}}$ and $\rho_{\text{\rm sh}}$ can be lower bounded as follows, utilizing concavity of the fidelity and Jensen's inequality:
\begin{align}
    F(\rho_{\text{true}}, \rho_{\text{\rm sh}}) &\geq \mathbb{E}_{\mathcal{F}}\big[F(\ketbra{\psi_{\text{true}}}{\psi_{\text{true}}}, \ketbra{\psi_{\text{\rm sh}}}{\psi_{\text{\rm sh}}})\big] \\
    &\geq 1 - \frac{1}{2}\mathbb{E}_{\mathcal{F}}\Big[ \|\ket{\psi_{\text{true}}} - \ket{\psi_{\text{\rm sh}}}\|^2\Big] \\
    &\geq 1 - \mathbb{E}_{\mathcal{F}}[\|\ket{\psi_{\text{true}}}\|^2]\,.
\end{align}
We remark that the proof up to this point follows the proof of the O2H Lemma from \cite{Chia_2023}, with the important distinction that obtaining Eq.~\eqref{eq:encoded_orthog} requires a modified definition of the flag unitary. However, the following lemma from \cite{Chia_2023} allows us to directly bound the 2-norm of the the unnormalized logical state $\ket{\psi_{\text{true}}}$ by a classical combinatorial argument, which is independent of whether this state is defined on a physical Hilbert space or a code subspace. 
\begin{lemma}[Lemma 5.8 in \cite{Chia_2023}]
Suppose the shadow wrappers satisfy $$\textnormal{Pr}\big[x\in S_i^{(k)}\big|x \in S_i^{(k-1)}\big] \leq p$$ for all $i, k$. Then take any initial state $\rho$ and unitary $U$ consisting of $q$ oracle queries and depth-1 layers. Given that $\rho$ and all depth-1 layers in $U$ are uncorrelated to $\bar{S}^{(k)}$ and $\mathcal{F}$ restricted to $\bar{S}^{(k)}$, we have
\begin{equation}
    \mathbb{E}_{\mathcal{F}}[\| \!\ket{\psi_{\text{\rm true}}}\! \|^2] \leq q\cdot p
\end{equation} 
\end{lemma}
Substituting this into the Fuchs-van de Graaf inequality, we obtain
\begin{equation}
    d_{\tr}(\rho_{\text{true}}, \rho_{\text{\rm sh}}) \leq \sqrt{2q\cdot p}\,.
\end{equation}
Lemma \ref{lemma:o2h} follows immediately, using the fact that the trace distance is the maximum distinguishability between the shadow and true states under any observable.
\end{proof}

\subsection{Proofs of depolarizing channel Lemmas} \label{appendix:depol_lemma_proofs}
\begin{proof}[Proof of Lemma \ref{lemma:depol_on_swap}]
    We observe that
    \begin{align}
    \textsf{SWAP} = \textsf{SWAP}_1^{\otimes k}
    \end{align}
    and
    \begin{align}
    (\mathcal{D}_\lambda \otimes \mathcal{D}_\lambda)[\textsf{SWAP}_1] = \frac{1}{2}\lambda(2-\lambda)\mathds{1}_2^{\otimes 2} + (1-\lambda)^2 \textsf{SWAP}_1.
    \end{align}
    Thus we have
    \begin{align}
    (\mathcal{D}_\lambda^{\otimes n} \otimes \mathcal{D}_\lambda^{\otimes n})[\textsf{SWAP}_n] = \bigotimes_{j = 1}^n \left(\frac{1}{2}\lambda(2-\lambda)\mathds{1}_2^{\otimes 2} + (1-\lambda)^2 \textsf{SWAP}_1\right).
    \end{align}
    Now, let $I\subseteq [n]$ denote a subset of indices labeling single-qubit Hilbert spaces. Then we define the following $2n$ qubit operator, which acts symmetrically across the two $n$-qubit Hilbert spaces:
    $$\textsf{SWAP}_1^I := \left(\bigotimes_{j\notin I}\mathds{1}_2^{\otimes 2}\right) \!\otimes\! \left(\bigotimes_{j\in I}\textsf{SWAP}_1\right).$$
    Using the \textsf{SWAP} trick, we notice that for any density matrix $\rho$ we have
    \begin{align}
    \tr\left(\rho^{\otimes 2} \textsf{SWAP}_2^I\right) = \tr(\rho_I^2) \leq 1\,,
    \end{align}
    where $\rho_I$ means we trace $\rho$ down to the subset of sites in $I$. Accordingly,
    \begin{align}
    \sup_{|\phi\rangle} \tr(|\phi\rangle \langle \phi|^{\otimes 2} (\mathcal{D}_\lambda^{\otimes n} \otimes \mathcal{D}_\lambda^{\otimes n})[\textsf{SWAP}]) &\leq \sum_{a = 0}^n \binom{n}{a} \left(\frac{1}{2}\lambda(2-\lambda)\right)^{n-a} \left((1-\lambda)^2\right)^a \\
    &= \left[\frac{1}{2}\lambda(2-\lambda) + (1-\lambda)^2\right]^n \\
    &= \left(1 - \lambda + \frac{1}{2}\lambda^2\right)^n = f(\lambda)^n.
    \end{align}
\end{proof}
\begin{proof}[Proof of Lemma \ref{lemma:depol_on_swap_2}]
    We have
    \begin{align}
        (\mathcal{D}_\lambda \otimes \mathcal{D}_\lambda)[2\textsf{SWAP}_1-\mathds{1}_2^{\otimes 2}] = -(1-\lambda)^2\mathds{1}_2^{\otimes 2} + 2(1-\lambda)^2 \textsf{SWAP}_1\,,
    \end{align}
    which gives us
   \begin{align}
    ((\mathcal{D}_\lambda\otimes \mathcal{D}_\lambda)[2\textsf{SWAP}_1-\mathds{1}_2^{\otimes 2}])^{\otimes n}= \bigotimes_{j = 1}^n \left(-(1-\lambda)^2\mathds{1}_2^{\otimes 2} + 2(1-\lambda)^2 \textsf{SWAP}_1\right).
    \end{align}
    Using the \textsf{SWAP} trick to once again assert that $\tr\left(\rho^{\otimes 2} \textsf{SWAP}_2^I\right) \leq 1$, we proceed as in Lemma \ref{lemma:depol_on_swap}:
    \begin{align}
    \tr(|\phi\rangle \langle \phi|^{\otimes 2} ((\mathcal{D}_\lambda\otimes \mathcal{D}_\lambda)[2\textsf{SWAP}_1-\mathds{1}_2^{\otimes 2}])^{\otimes n})&\leq \sum_{a \geq 0,\,a\text{ even}}^n \binom{n}{a} \left((1-\lambda)^2\right)^{a} \left(2(1-\lambda)^2\right)^{n-a} \\
    &=\frac{3^n+1}{2}(1-\lambda)^{2n}\,.
    \end{align}
\end{proof}

\bibliographystyle{alpha}
\bibliography{refs}

\end{document}